\documentclass[11pt]{article}

\usepackage{amsthm,amsmath,amssymb}
\usepackage{natbib}
\usepackage{multirow}
\usepackage[pdftex]{graphicx}
\usepackage{subfigure}
\usepackage{makecell}
\usepackage{booktabs}
\usepackage{array}
\usepackage{fullpage}
\usepackage{url}
\usepackage{algorithm}
\usepackage{algorithmic}
\usepackage{bm}
\usepackage{smile}
\usepackage{mathtools}
\usepackage{wrapfig}
\usepackage{lipsum}
\usepackage{mathrsfs}
\usepackage{dsfont}
\usepackage{titling}
\usepackage{epstopdf}
\usepackage{setspace}
\usepackage{lscape}
\usepackage{color}
\usepackage{ragged2e}

\usepackage{natbib}
\usepackage{multirow}

\usepackage[usenames,dvipsnames,svgnames,table]{xcolor}
\usepackage[colorlinks=true,
            linkcolor=blue,
            urlcolor=blue,
            citecolor=blue]{hyperref}

\newcommand\independent{\protect\mathpalette{\protect\independenT}{\perp}}
\def\independenT#1#2{\mathrel{\rlap{$#1#2$}\mkern2mu{#1#2}}}

\newtheorem{thm}{Theorem}

\newtheorem*{defn*}{Definition}

\newtheorem{cor}{Corollary}

\begin{document}

\title{\huge Regression Discontinuity Design under Self-selection \thanks{We would like to thank Zhuan Pei, Yanqing Fan, Peter Hull and participants in seminars and conferences at which this paper was presented. All remaining errors are ours.}}

\author{Sida Peng\thanks{Microsoft Research, sidpeng@microsoft.com}~~~and ~~Yang Ning\thanks{Department of Statistics and Data Science, Cornell University, yn265@cornell.edu}}

\date{\today}

\maketitle

\vspace{0.1in}

\begin{abstract}
\noindent In Regression Discontinuity (RD) design, self-selection leads to different distributions of covariates on two sides of the policy intervention, which essentially violates the continuity of potential outcome assumption. The standard RD estimand becomes difficult to interpret due to the existence of some indirect effect, i.e. the effect due to self selection. We show that the direct causal effect of interest can still be recovered under a class of estimands. Specifically, we consider a class of weighted average treatment effects tailored for potentially different target populations. We show that a special case of our estimands can recover the average treatment effect under the conditional independence assumption per \cite{Angrist2015}, and another example is the estimand recently proposed in \cite{FrolichHuber2018}. We propose a set of estimators through a weighted local linear regression framework and prove the consistency and asymptotic normality of the estimators. Our approach can be further extended to the fuzzy RD case. In simulation exercises, we compare the performance of our estimator with the standard RD estimator. Finally, we apply our method to two empirical data sets: the U.S. House elections data in \cite{Lee2008} and a novel data set from Microsoft Bing on  Generalized Second Price (GSP) auction.    \\[1em]
\end{abstract}

\section{Introduction}
Regression discontinuity (RD) design is an important policy evaluation tool that has been widely used in empirical studies.  Under the continuity assumptions \citep{Lee2008}, the RD design gives rise to many testable restriction similar to a randomized control trial, and it allows for the identification of causal effects \citep{Hahn2001}. Standard nonparametric tools like series expansion method or kernel regression method can be applied to estimate this quantity under this minimal assumption, see \cite{Imbens2008} and \cite{CattaneoRD2017}.

An implication of the continuity assumption is that the distribution of the covariates conditional on the assignment variable is continuous at the cutoff. To valid this assumption, a variety of statistical tests and nonparametric inference procedures have been proposed, including \cite{Cattaneo2015} and \cite{Canay2018}. However, this assumption may not hold in many applications. In particular, we consider the following two motivating examples. 

\begin{example}[College scholarship]
Consider the classical scholarship example per \cite{Thistlethwaite1960}. Students receive scholarships if their test score is higher than a threshold. The authors applied the RD framework to estimate the causal effect of the scholarship on student's future education outcomes. In this example, the continuity assumption implies that the distributions of the covariate information for students whose score is just above and below the threshold are similar. However, if students are aware of this policy before they take the SAT exam, some of them may study harder to pass this threshold. As a result, students whose score is on two sides of the threshold may have distinct characteristics (e.g., gender). 
\end{example} 

\begin{example}[Generalized Second Price (GSP) auction]
GSP is widely used by internet search engines like Google or Bing to allocate sponsored search advertisements. In reality, it is implemented through a reservation score system. Each bidder is assigned a score as a function of their bids and characteristics. The bidder's advertisement is displayed if her score passes the reservation. However, due to design of the score system, bidders with low quality may be selected below the reservation score while bidders with high quality may be selected above the reservation score. Given that low quality bidders may have a high willingness to pay, we often observe the distribution of bids to be different on two sides of the reservation score. We will further elaborate this example in the real data analysis. 
\end{example} 
  
Due to different distributions of the covariates at the cutoff in these examples, the standard RD estimand is no longer valid.  We show in those cases the standard RD estimand can be decomposed into a direct treatment effect and an indirect treatment effect. The indirect effect is due to the unbalanced covariates near the cut-off. For example, the policy intervention may result more boys than girls to receive scholarship. If in general boys perform different from girls in SAT test, the difference in SAT score due to gender will also be accounted into the standard RD estimand as if the policy intervention is selecting on genders. However as no direct causal mechanism is assumed between the unbalanced covariates and running variable, the selection could be generated due to an unknown equilibrium, completely reversed or purely spurious. For example, in a GSP auction, the reservation score is designed to separate the bidders by quality and in the meanwhile bidders may self-select such that low quality bidders may have incentive to bid higher. As a result, policy changes to move reservation score can be risky if the equilibrium between reservation score and quality of bidders are not disentangled.  

In this paper, we propose a new framework to address this problem by adjusting unbalanced covariates due to self-selection. Consider the following sharp RD setup: $T$ is a binary treatment variable, $Y(1)$ and $Y(0)$ are the potential outcomes under $T=1$ and $T=0$ and $X$ is the running variable so that the treatment is fully determined by $T = \bold{1}(X>c)$ for a known threshold $c$. In the classical RD framework, it is assumed that $\EE(Y(t)|X=x)$ is continuous at $x=c$ for $t=0, 1$; see Assumption \ref{asscont1}. This essentially assumes away self-selection based on both observed and unobserved covariates. To account for the self-selection effect, we assume that further covariate information can be collected. In particular, we define $Z(1)$ and $Z(0)$ as the potential covariates with or without treatment. By using the potential ``outcome" formulation, we allow the distribution of the covariates on two sides of the threshold to be different, i.e., discontinuity of the covariate distribution. By controlling all unbalanced covariates, we assume that $\EE(Y(t)|X=x, Z(t)=z)$ is continuous in $x$ at the threshold; see Assumption \ref{asscont2}. This is the main assumption made in this paper.   

In the classical RD framework, the standard causal parameter of interest is the marginal treatment effect $\mathbb{E}(Y(1) - Y(0)|X=c)$ at the cutoff. In the presence of self-selection, the continuity assumption on $\EE(Y(t)|X=x)$ may fail. Thus, the marginal treatment effect can be confounded by the discontinuity of the conditional mean of covariates at the cutoff, as seen in equation (\ref{eqtaulinear}). In this paper, we propose a class of estimands for RD design in the framework of weighted average treatment effect (WATE). We show that our estimands can tease out the effect of discontinuity of the conditional distribution of covariates through re-weighting the marginal treatment effects. For instance, under the constant treatment effect model with unbalanced covariates, our estimands reduce to the direct treatment effect. One special case of our estimands is \cite{FrolichHuber2018}.  Another special case can be interpreted as the ``global" average treatment effect (ATE) $\mathbb{E}(Y(1) - Y(0))$ under the conditional independence assumption (CIA) as in \cite{Angrist2015}. The CIA assumes that the treatment is mean independent of the running variable near the cutoff conditional on the covariates. Intuitively, after projecting the outcome variable  onto a rich set of covariates (excluding running variable), the residual should not depend on the running variable and can be viewed as an experiment with randomly assigned treatment. We show that our method identifies ATE under CIA assumption, whereas the standard RD estimand remains a ``local" treatment effect at the cutoff \citep{Lee2008}.

We further provide the nonparametric identification for our estimands and propose nonparametric estimators based on the inverse propensity score weighted (IPW) approach, see \cite{HorvitzThompson1952} and \cite{AbadieImbens2016}. However, notice that the treatment assignment is degenerate with respect to the running variable. We get around this problem by considering the marginal effect of the treatment and re-weighting on the other covariates first. The kernel method is applied to estimate the conditional mean function. The consistency and asymptotic normality of the proposed estimator are established. We further extend our method to the fuzzy RD design where the treatment compliance is imperfect.  Similarly, we provide the nonparametric identification of the causal effect under the fuzzy RD design, and propose a nonparametric estimator. The proposed estimator is similar to that of the local average treatment effect (LATE) in a fractional format but with numerator being adjusted to incorporate additional selections.

This work is connected to the growing literature on the RD design with covariates. In particular, two recent papers provide insightful guidance on the subject. \cite{Matias2017} estimated the marginal treatment effect by a local linear regression with the linear-in-parameters specification for the covariates. The main advantage of their method is that the nonparametric estimation of $\EE(Y(t)|X=x, Z(t)=z)$ is avoided.  In another paper, \cite{FrolichHuber2018} proposed a fully nonparametric estimator of the marginal treatment effect by estimating $\EE(Y(t)|X=x, Z(t)=z)$ nonparametrically. They allowed the conditional density of $Z(t)$ given $X$ to be discontinuous. Our work differs from the above papers by considering a different estimand that is less local under CIA  and identifies the direct treatment effect under the constant treatment effect model. Unlike \cite{Matias2017},  we do not require the continuity of the conditional mean of  $Z(t)$ given $X$. Instead, our Assumption \ref{asscont2} is similar to assumption 1 (iv) in \cite{FrolichHuber2018} under the sharp RD design. The proposed IPW estimator is also different from the above regression based estimators.


The rest of the paper is organized as follows. In Section \ref{sec_estimator}, we propose our new estimands in the form of weighted average treatment effects and establish its nonparametric identification. In Section  \ref{sec_est}, we propose a set of weighted local linear estimators. The theoretical properties are established in Section \ref{sec_asymptotics}. In Section \ref{sec_simulation}, we conduct simulation studies and also apply the method to two empirical data sets: the U.S. House elections data in \cite{Lee2008} and a novel data set from Microsoft Bing on  Generalized Second Price (GSP) auction. Finally, we consider the extension to the fuzzy RD case in Section \ref{sec_Fuzzy}. The proofs are deferred to the appendix.



\section{Sharp RD Design}
\label{sec_estimator}


\subsection{Problem Setup and Continuity Assumption}\label{sec_setup}
In the standard RD design setting, we observe $n$ i.i.d. random samples $\{Y_i, X_i, Z_i, T_i\}_{i = 1}^n$, where $Y_i$ is the outcome variable of interest  for the $i$th sample, $T_i\in\{0,1\}$ is the binary treatment variable, $X_i\in\RR$ is the running variable and $Z_i\in\RR^p$ is the covariate. In the sharp RD design, the treatment $T_i$ is perfectly assigned through the running variable $X_i$ relative to a known cutoff $c$. For example, we have
\[T_i = \bold{1}(X_i>c).\]  
Adopting a potential outcome framework, we can write the observed outcome variable $Y_i$ as
\[Y_i = Y_i(0)\cdot (1-T_i)+Y_i(1)\cdot T_i,\]
where $Y_i(0)$ and $Y_i(1)$ represent the potential outcomes without or with treatment. The average treatment effect is defined as $\mathbb{E}(Y_i(1) - Y_i(0))$. However, this estimand is not identifiable under the RD design as the treatment assignment $T_i$ is a deterministic function of $X_i$. In this framework,  \cite{Hahn2001}, \cite{Lee2008} and \cite{Cattaneo2015} showed that one can still identify the treatment effect at the cutoff
\[\tau_{SRD} = \mathbb{E}(Y_i(1) - Y_i(0)|X_i=c),
\]
under the following continuity assumption. 

\begin{assumption}[Continuity Assumption of $\EE(Y_i(t)|X_i=x)$]\label{asscont1}
\begin{equation*}
\EE(Y_i(1)|X_i=x) \quad \text{   and   } \quad \EE(Y_i(0)|X_i=x)
\end{equation*}
are continuous in $x$ at $x = c$.
\end{assumption}
This assumption implies that the conditional mean of the potential outcomes near the cutoff $x=c$ are similar. There is no discontinuity of the conditional mean functions at the cutoff. This assumption enables us to identify $\tau_{SRD}$ in the RD design. We refer to \cite{Hahn2001}, \cite{Lee2008} and \cite{Cattaneo2015} for further discussion on this assumption.

Now let us consider the case that additional covariates $Z_i$ are observed. Denote
\[Z_i = Z_i(0)\cdot (1-T_i)+Z_i(1)\cdot T_i,\]
where $Z_i(0)$ and $Z_i(1)$ represent the potential covariates without or with treatment. In the presence of covariates $Z_i$, the causal parameter $\tau_{SRD}$ can be rewritten as  
\begin{equation}\label{eqtau}
\tau_{SRD} = \mathbb{E}\Big(\mathbb{E}[Y_i(1) |X_i=c, Z_i(1)]- \mathbb{E}[Y_i(0) |X_i=c, Z_i(0)] \Big|X_i=c\Big).
\end{equation}

In a recent work, \cite{Matias2017} proposed a kernel based estimator of $\tau_{SRD}$ by accounting for the additional covariates $Z_i$. In addition to the continuity Assumption \ref{asscont1}, it is also assumed that $\EE(Z_i(1)|X_i=c) = \EE(Z_i(0)|X_i=c)$ for the consistency of the resulting kernel estimator, that is the potential covariates $Z_i(1)$ and $Z_i(0)$ have the same conditional mean at the cutoff $X_i=c$.

However, in some applications we may observe $\EE(Z_i(1)|X_i=c) \neq \EE(Z_i(0)|X_i=c)$, when self-selection based on the covariates exists. For example, consider the classical scholarship example. Students with SAT score higher than a threshold will receive scholarship. The treatment effect of interest is the effect of scholarship on the students' first semester GPAs. If the cutoff is pre-released, it might be possible that students with some common characteristics (i.e. gender) may study harder to pass the bar. This leads to an ex-ante selection based on the covariates. Thus, one may observe that the conditional mean functions of $Z_i$ given $X_i$ right below or above the threshold are different, i.e., 
\begin{equation}\label{eqself}
\lim_{x\rightarrow c^+}\EE(Z_i |X_i=x)\neq \lim_{x\rightarrow c^-}\EE(Z_i |X_i=x).
\end{equation}
The following simple lemma essentially says that the self-selection based on the covariates (i.e., eq \ref{eqself}) implies $\EE(Z_i(1)|X_i=c) \neq \EE(Z_i(0)|X_i=c)$.

\begin{lemma} 
If $\EE(Z_i(t)|X_i=x)$ is continuous at $x=c$ for $t\in\{0,1\}$, then
$$
\EE(Z_i(1)|X_i=c) = \EE(Z_i(0)|X_i=c) \Longleftrightarrow \lim_{x\rightarrow c^+}\EE(Z_i |X_i=x)=\lim_{x\rightarrow c^-}\EE(Z_i |X_i=x).
$$
\end{lemma}

\begin{proof}
By the continuity assumption on $\EE(Z_i(t)|X_i=x)$, 
$$
\EE(Z_i(1)|X_i=c)=\lim_{x\rightarrow c^+}\EE(Z_i(1)|X_i=x)=\lim_{x\rightarrow c^+}\EE(Z_i |X_i=x),
$$
and similarly $\EE(Z_i(0)|X_i=c)=\lim_{x\rightarrow c^-}\EE(Z_i |X_i=x)$. The lemma holds.
\end{proof}

The above lemma provides a convenient way to check whether $\EE(Z_i(1)|X_i=c) = \EE(Z_i(0)|X_i=c)$ holds in empirical studies. One may simply plot the observed covariates $Z_i$ against $X_i$ and examine whether there is a discontinuity of the trend around $x=c$. The method is applied in the real data analysis. 

In the following, we investigate the consequence of $\EE(Z_i(1)|X_i=c) \neq \EE(Z_i(0)|X_i=c)$. To be specific, we consider the following constant treatment effect model 
\begin{equation}\label{eqconstantmodel}
\EE(Y_i(t)| X_i, Z_i (t)) = \alpha + \tau \bold{1}(t=1) + g(X_i)+Z_i(t)\gamma,
\end{equation}
for $t\in\{0,1\}$, where $g(\cdot)$ is an arbitrary continuous function. By (\ref{eqtau}), one can show that
\begin{equation}\label{eqtaulinear}
\tau_{SRD} = \tau +  \Big(\mathbb{E}(Z_i(1)|X_i = c) - \mathbb{E}(Z_i(0)|X_i = c)\Big)\gamma.
\end{equation}
The estimand $\tau_{SRD}$ can be decomposed into two terms. The first term $\tau$ represents the direct treatment effect after controlling the running variable $X_i$ and the covariates $Z_i(t)$. The second term in the right hand side of (\ref{eqtaulinear}) can be interpreted as the indirect effect of the policy due to the unbalanced covariates near the cutoff or self-selection, which is nonzero if $\gamma\neq 0$ and $\EE(Z_i(1)|X_i=c) \neq \EE(Z_i(0)|X_i=c)$. In many applications, the direct treatment effect $\tau$ is usually more meaningful and interpretable than $\tau_{SRD}$, as $\tau_{SRD}$ is confounded by the self-selection effect. 

In this example, when the indirect effects in $\tau_{SRD}$ is assumed away by requiring the continuity of the conditional density of the covariates $Z_i$ at the cutoff $X_i = c$, the data around the cutoff can be viewed as a natural experiment and continuity on the covariates implies a balanced design for this experiment so that we can estimate a local average treatment effect. However when the self-selection exits, the experiment is no longer balanced and the average treatment effect $\tau_{SRD}$ will typically differ from the direct effect $\tau$. 



\subsection{Weighted Average Treatment Effect}

As seen in (\ref{eqtaulinear}), if $\EE(Z_i(1)|X_i=c) \neq \EE(Z_i(0)|X_i=c)$ (i.e., the covariates are unbalanced at the cutoff) and $\gamma\neq 0$, $\tau_{SRD}$ can be different from the causal parameter of interest. To overcome this difficulty, we propose a new class of causal parameters, called the weighted average treatment effect (WATE), which are defined as
\begin{align}
\tau^{w}_{SRD}&=\EE\{Y(1)w_1(Z(1))|X=c\}-\EE\{Y(0)w_0(Z(0))|X=c\}\nonumber\\
&=\int \Big[\EE(Y(1)|X=c, Z(1)=z)w_1(z)f_{Z(1)|X}(z|c)\nonumber\\
&~~~~~~ -\EE(Y(0)|X=c, Z(1)=z)w_0(z)f_{Z(0)|X}(z|c)\Big]dz,\label{eqtaunew}
\end{align}
where $w_1(\cdot)$ and $w_0(\cdot)$ denote different choices of weights to form the estimand, and $f_{Z(1)|X}(\cdot |\cdot)$ and $f_{Z(0)|X}(\cdot |\cdot)$ are the conditional density of $Z(1)$ and $Z(0)$ given $X$. In order to interpret (\ref{eqtaunew}) as the WATE, we require the following normalization condition for $w_1(\cdot)$ and $w_0(\cdot)$:
$$
\int w_1(z)f_{Z(1)|X}(z|c)dz=\int w_0(z)f_{Z(0)|X}(z|c)dz=1.
$$
In particular, by choosing appropriate $w_1(\cdot)$ and $w_0(\cdot)$, (\ref{eqtaunew}) can be interpreted as the average of the difference of the conditional mean functions corresponding to a target population. To see this, we consider the following examples. Denote 
$$
\Delta(c,z)=\mathbb{E}(Y(1)|X=c, Z(1)=z) - \mathbb{E}(Y(0)|X=c, Z(0)=z).
$$

\begin{itemize}
\item Average treatment effect over entire population: 
$$
\tau^{w1}_{SRD}=\int \Delta(c,z) f_Z(z)dz,
$$
where $f_Z(\cdot)$ is the p.d.f of the covariates $Z$. In $\tau^{w1}_{SRD}$, we average the conditional mean difference over the entire population whose covariates follow from the {\it marginal distribution} $f_Z(z)$. 
It is easy to see that WATE defined in (\ref{eqtaunew}) recovers $\tau^{w1}_{SRD}$ by taking 
$$
w_1(z)=\frac{f_Z(z)}{f_{Z(1)|X}(z|c)},~~ \textrm{and}~~w_0(z)=\frac{f_Z(z)}{f_{Z(0)|X}(z|c)}.
$$ 
\item Average treatment effect over locally untreated population:
$$
\tau^{w2}_{SRD}=\int \Delta(c,z) f_{Z(0)|X}(z|c)dz.
$$
In this causal parameter, we average the conditional mean difference over the untreated population right below the threshold whose covariates follow from the {\it conditional distribution} $f_{Z(0)|X}(z|c)$. Similarly, we obtain $\tau^{w2}_{SRD}$ by taking
$$
w_1(z)=\frac{f_{Z(0)|X}(z|c)}{f_{Z(1)|X}(z|c)}~~ \textrm{and}~~w_0(z)=1.
$$
\item Average treatment effect over locally randomized population:
$$
\tau^{w3}_{SRD}=\int \Delta(c,z) \frac{f_{Z(0)|X}(z|c) + f_{Z(1)|X}(z|c)}{2} dz.
$$ 
This is the estimand studied by \cite{FrolichHuber2018} in the sharp RD case. Under the proposed WATE framework, $\tau^{w3}_{SRD}$ can be viewed as the average treatment effect over the population around the threshold which is randomized so that their covariates follow from $f_{Z(0)|X}(z|c)$ and $f_{Z(1)|X}(z|c)$ with equal probability. Similarly, we obtain $\tau^{w3}_{SRD}$ by taking
$$
w_1(z)=\frac{f_{Z(1)|X}(z|c)+f_{Z(0)|X}(z|c)}{2f_{Z(1)|X}(z|c)}~~ \textrm{and}~~w_0(z)=\frac{f_{Z(1)|X}(z|c)+f_{Z(0)|X}(z|c)}{2f_{Z(0)|X}(z|c)}.
$$
\item Average treatment effect via classical RD estimand: We note that the proposed WATE reduces to the classical RD estimand $\tau_{SRD}$ (\ref{eqtau}) by taking $w_1(z)=w_0(z)=1$. However, we note that unlike the previous three examples, $\tau_{SRD}$ may not be written as the average treatment effect over one well defined population. To see this, recall that when there exists self-selection, the conditional distributions $f_{Z(1)|X}(z| c)$ and $f_{Z(0)|X}(z| c)$ usually differ from each other. Then $\tau_{SRD}$ in (\ref{eqtau}) can be written as the difference of the average of $\mathbb{E}(Y_i(t)|X_i=c, Z_i(t)=z)$ over two populations (i.e., the population right below and above the threshold) with covariate distributions $f_{Z(1)|X}(z| c)$ and $f_{Z(0)|X}(z| c)$ respectively. This is the reason for which $\tau_{SRD}$ is confounded by the unbalanced covariates. 
\end{itemize} 


A summary of the first three estimands is provided in Table \ref{tab0}. In practice, which causal estimand in above examples to use should depend on the target population of interest and is often determined on a case-by-case basis. Indeed, our framework opens a door towards designing new causal parameters tailored to specific applications. For instance, similar to $\tau^{w2}_{SRD}$, one can define the average treatment effect over locally treated population i.e., $\int \Delta(c,z) f_{Z(1)|X}(z|c)dz$. 
Since the goal of the paper is to deal with unbalanced covariates, to fix the idea we will mainly focus on the first three examples. 


In the following, we comment on two properties of our estimands $\tau^{w1}_{SRD}$, $\tau^{w2}_{SRD}$ and $\tau^{w3}_{SRD}$. First, under the constant treatment effect model (\ref{eqconstantmodel}), direct calculation shows that $\Delta(c,z)=\tau$ and thus $\tau^{w1}_{SRD}=\tau^{w2}_{SRD}=\tau^{w3}_{SRD}$ equals to the direct treatment effect $\tau$ without any further assumption. In contrast, the classical RD estimand $\tau_{SRD}$ reduces to $\tau$ under the extra assumption that $\gamma= 0$ or $\EE(Z_i(1)|X_i=c) = \EE(Z_i(0)|X_i=c)$.

Second, our estimand $\tau^{w1}_{SRD}$ generalizes to the overall average treatment effect (ATE) under the conditional independence assumption (CIA) proposed by \cite{Angrist2015}. Assume that $Z_i$ are the pre-treatment covariates, i.e, $Z_i(1)=Z_i(0)=Z_i$. The CIA is defined as  
\begin{equation}\label{eqcia}
\mathbb{E}(Y_i(1)|X_i, Z_i) = \mathbb{E}(Y_i(1)|Z_i), \qquad \mathbb{E}(Y_i(0)|X_i, Z_i) = \mathbb{E}(Y_i(0)|Z_i),
\end{equation}
which says that the potential outcomes are mean independent of the running variable conditional on the covariates. By controlling a rich set of covariates, CIA seems to be a reasonable assumption as the link between the running variable and outcomes can be blocked \citep{Angrist2015}. Since the CIA (\ref{eqcia}) implies $\Delta(c,z)=\Delta(z)$, our estimand $\tau^{w1}_{SRD}$ reduces to
$$
\tau^{w1}_{SRD}=\int \Delta(z) f_Z(z)dz=\mathbb{E}(Y_i(1) - Y_i(0)),
$$
which is the overall ATE. In contrast, 
$$
\tau^{w2}_{SRD}=\tau^{w3}_{SRD}=\tau_{SRD}=\mathbb{E}(Y_i(1) - Y_i(0)|X_i=c)
$$
remains a ``local" treatment effect at the cutoff. It requires further conditions to generalize to the overall ATE. For instance, if the constant treatment assumption $\mathbb{E}(Y_i(1)|Z_i)-\mathbb{E}(Y_i(0)|Z_i)=a$ holds for some constant $a$, then $\tau_{SRD}$ reduces to $\mathbb{E}(Y_i(1) - Y_i(0))$. Thus, the new estimand $\tau^{w1}_{SRD}$ can represent a causal effect that is less local than the standard RD estimand $\tau_{SRD}$.


\begin{landscape}
\begin{table}
\singlespacing
	\begin{center}
		\begin{tabular}{ccccc}
			\toprule 

Estimand & $w_1(z)$ & $w_0(z)$ & $\pi_1(z)$ & $\pi_0(z)$\\
\hline
$\int \Delta(c,z) f_Z(z)dz$ & $\frac{f_Z(z)}{f_{Z(1)|X}(z|c)}$ & $\frac{f_Z(z)}{f_{Z(0)|X}(z|c)}$ & $\frac{f_{X, Z(1)}(c,z)}{2f_{Z}(z)}$ &$\frac{f_{X, Z(0)}(c,z)}{2f_{Z}(z)}$\\
$\int \Delta(c,z) f_{Z(0)|X}(z|c)dz$ & $\frac{f_{Z(0)|X}(z|c)}{f_{Z(1)|X}(z|c)}$ &$1$ &  $\frac{f_{X, Z(1)}(c,z)}{2f_{Z(0)|X}(z|c)}$ & $\frac{f_X(c)}{2}$\\
$\int \Delta(c,z) \frac{f_{Z(0)|X}(z|c) + f_{Z(1)|X}(z|c)}{2} dz$ & $\frac{f_{Z(1)|X}(z|c)+f_{Z(0)|X}(z|c)}{2f_{Z(1)|X}(z|c)}$ & $\frac{f_{Z(1)|X}(z|c)+f_{Z(0)|X}(z|c)}{2f_{Z(0)|X}(z|c)}$ & $\frac{f_{X, Z(1)}(c,z)}{f_{Z(1)|X}(z|c)+f_{Z(0)|X}(z|c)}$ & $\frac{f_{X, Z(0)}(c,z)}{f_{Z(1)|X}(z|c)+f_{Z(0)|X}(z|c)}$\\
 \bottomrule
		\end{tabular}
	\end{center}
\vspace{-.2in}
\caption{Three examples of the weighted average treatment effect $\tau^w_{SRD}$, where $\Delta(c,z)=\mathbb{E}(Y_i(1)|X_i=c, Z_i(1)=z) - \mathbb{E}(Y_i(0)|X_i=c, Z_i(0)=z)$.} \label{tab0}
\end{table}
\end{landscape}

\subsection{Nonparametric Identification}
In this subsection, we study the nonparametric identification of $\tau^{w1}_{SRD}$, $\tau^{w2}_{SRD}$ and $\tau^{w3}_{SRD}$. Instead of Assumption \ref{asscont1}, we impose the following continuity assumption.

\begin{assumption}[Continuity Assumption of $\EE(Y_i(t)|X_i=x, Z_i(t)=z)$]\label{asscont2}
\begin{equation*}
\EE(Y_i(1)|X_i=x, Z_i(1)=z) \quad \text{   and   } \quad \EE(Y_i(0)|X_i=x, Z_i(0)=z)
\end{equation*}
are right and left continuous in $x$ at $x=c$ for any $z\in\cZ$ respectively.
\end{assumption}

Intuitively, this assumption says, once all unbalanced covariates are controlled, there is no further discontinuity between the running variable and outcomes at the threshold. This assumption is similar to assumption 1 (iv) in \cite{FrolichHuber2018} under the sharp RD design. In addition, this assumption is weaker than CIA, as  (\ref{eqcia}) implies our Assumption \ref{asscont2} when the covariates are pre-determined.


The following theorem shows that $\tau^{w1}_{SRD}$ is identifiable based on the distribution of the observed data under  Assumption \ref{asscont2}. In addition, $\tau^{w2}_{SRD}$ and $\tau^{w3}_{SRD}$ are identifiable under some extra continuity assumptions. 

\begin{theorem}[Nonparametric Identification]
Under Assumption \ref{asscont2}, $\tau^{w1}_{SRD}$ is identifiable: 
$$
\tau^{w1}_{SRD}=\int \big[\mathbb{E}(Y|X=c^+, Z=z)-\mathbb{E}(Y|X=c^-, Z=z)\big]f_Z(z)dz ,
$$
where $\mathbb{E}(Y|X=c^+, Z)=\lim_{x\rightarrow c^+} \mathbb{E}(Y|X=x, Z)$ and $\mathbb{E}(Y|X=c^-, Z)=\lim_{x\rightarrow c^-} \mathbb{E}(Y|X=x, Z)$. In addition, if $f_{Z(0)|X}(z|x)$ is left continuous in $x$ at $x=c$ for any $z\in\cZ$, then $\tau^{w2}_{SRD}$ is identifiable: 
$$
\tau^{w2}_{SRD}=\int \big[\mathbb{E}(Y|X=c^+, Z=z)-\mathbb{E}(Y|X=c^-, Z=z)\big]f_{Z|X}(z|c^-)dz ,
$$
where $f_{Z|X}(z|c^-)=\lim_{x\rightarrow c^-}f_{Z|X}(z|x)$. Furthermore, if $f_{Z(1)|X}(z|x)$ is right continuous in $x$ at $x=c$ for any $z\in\cZ$, then $\tau^{w3}_{SRD}$ is identifiable: 
$$
\tau^{w3}_{SRD}=\int \big[\mathbb{E}(Y|X=c^+, Z=z)-\mathbb{E}(Y|X=c^-, Z=z)\big]\frac{f_{Z|X}(z|c^-)+f_{Z|X}(z|c^+)}{2}dz.
$$
\end{theorem}

\begin{proof}
To show the identifiability of $\tau^{w1}_{SRD}$, we note that Assumption \ref{asscont2} implies
\begin{align*}
\int \mathbb{E}(Y(1) |X=c, Z(1)=z) f_Z(z) dz&=\int \lim_{\delta\rightarrow 0^+}\mathbb{E}(Y(1) |X=c+\delta, Z(1)=z) f_Z(z) dz, 
\end{align*}
By the definition of the potential outcome, we have $\mathbb{E}(Y(1) |X=c+\delta, Z(1)=z)=\mathbb{E}(Y |X=c+\delta, Z=z)$ for any $\delta>0$. Thus,
$$
\int \mathbb{E}(Y(1) |X=c, Z(1)=z) f_Z(z) dz=\int \lim_{\delta\rightarrow 0^+}\mathbb{E}(Y |X=c+\delta, Z=z) f_Z(z) dz.
$$
Following the same step, we can show that 
$$
\int \mathbb{E}(Y(0) |X=c, Z(0)=z) f_Z(z) dz=\int \lim_{\delta\rightarrow 0^-}\mathbb{E}(Y |X=c+\delta, Z=z) f_Z(z) dz.
$$
This implies $\tau^{w1}_{SRD}$ is identifiable. To show $\tau^{w2}_{SRD}$ is identifiable, similarly we have $\mathbb{E}(Y(1) |X=c, Z(1)=z)=\mathbb{E}(Y |X=c^+, Z=z)$. In addition, by the left continuity of $f_{Z(0)|X}(z|x)$, we further have $f_{Z(0)|X}(z|c)=\lim_{\delta\rightarrow 0^-} f_{Z(0)|X}(z|c+\delta)=f_{Z|X}(z|c^-)$. Thus, we have
\begin{align*}
\int \mathbb{E}(Y(1) |X=c, Z(1)=z) f_{Z(0)|X}(z|c) dz&=\int \mathbb{E}(Y |X=c^+, Z=z) f_{Z|X}(z|c^-) dz.
\end{align*}
This implies $\tau^{w2}_{SRD}$ is identifiable. The identifiability of $\tau^{w3}_{SRD}$ follows from the same argument. This completes the proof.
\end{proof}

\section{Nonparametric Estimation}\label{sec_est}

In the causal inference literature, inverse propensity score weighting (IPW) is one of the most widely used tools to handle the unbalanced covariates in the treatment and control groups \citep{HorvitzThompson1952}. However, the standard IPW method is not directly applicable because in the sharp RD design the treatment assignment is a deterministic function of the running variable and thus the propensity score is degenerate. In this section, we propose a class of nonparametric estimators of $\tau^{w1}_{SRD}$, $\tau^{w2}_{SRD}$ and $\tau^{w3}_{SRD}$ by modifying the inverse propensity score weighting approach. 

To motivate our nonparametric estimator, we consider the following notation. Denote by $\pi_1(Z_i)$ a function of the covariate $Z_i$ to be chosen later, $K(\cdot)$ a symmetric kernel function and $h$ a bandwidth that shrinks to 0. The detailed conditions on the kernel function and bandwidth are deferred to the next section. Consider the following inverse weighted kernel estimator 
\begin{equation}\label{eqkernel}
\frac{1}{n}\sum_{i=1}^n\frac{Y_iT_i}{\pi_1(Z_i)} \cdot h^{-1}K\Big(\frac{X_i-c}{h}\Big),
\end{equation}
for the estimand $\EE\{Y(1)w_1(Z(1))|X=c\}$, where $\pi_1(Z)$ plays the same role as the propensity score in the IPW method. Since in the RD design the propensity score function is degenerate, in the following we will show that the choice of $\pi_1(Z_i)$ differs from the standard propensity score model. The rationale is to choose $\pi_1(Z_i)$ so that the estimator (\ref{eqkernel}) is asymptotically unbiased,
\begin{equation}\label{unbiasedeq}
\mathbb{E}\left( \frac{Y_iT_i}{\pi_1(Z_i)} \cdot h^{-1}K\Big(\frac{X_i-c}{h}\Big)\right) \approx \EE\{Y_i(1)w_1(Z_i(1))|X_i=c\}.
\end{equation}
Thus, it suffices to calculate the expectation of the kernel estimator as follows
\begin{align}
&\mathbb{E}\left( \frac{Y_iT_i}{\pi_1(Z_i)} \cdot h^{-1}K\Big(\frac{X_i-c}{h}\Big)\right) \nonumber\\
&= \mathbb{E}\left( \frac{T_i}{\pi_1(Z_i(1))} \cdot h^{-1}K\Big(\frac{X_i-c}{h}\Big)\cdot\EE(Y_i(1)|X_i,Z_i(1))\right)\nonumber\\
& = \int \int_{x>c} \mathbb{E}(Y_i(1)|X_i=x, Z_i(1)=z) \frac{ h^{-1}}{\pi_1(z)} K\Big(\frac{x-c}{h}\Big) f_{Z(1),X}(z, x)dxdz,\label{eqstep1}
\end{align}
where the first step follows from $Z_i=Z_i(1)$ under $T=1$ and the last step follows from the definition $T_i = \bold{1}(X_i>c)$. Thus, (\ref{eqstep1}) implies
\begin{align}
&\mathbb{E}\left( \frac{Y_iT_i}{\pi_1(Z_i)} \cdot h^{-1}K\Big(\frac{X_i-c}{h}\Big)\right)\nonumber\\ 
&= \int \int_{u>0} \mathbb{E}(Y_i(1)|X_i=uh+c, Z_i(1)=z) \frac{K(u)}{\pi_1(z)} f_{Z(1),X}(z,uh+c)dudz\nonumber\\
&\approx\int \mathbb{E}(Y_i(1)|X_i=c, Z_i(1)=z) \frac{1}{\pi_1(z)} \frac{f_{Z(1),X}(z, c)}{2}dz,\label{eqstep2}
\end{align}
where the last step follows from the symmetry of the kernel function $\int_{u>0} K(u)du=1/2$ and can be made rigorous given the regularity conditions specified in the next section. Comparing with
$$
\EE\{Y_i(1)w_1(Z_i(1))|X_i=c\}=\int \mathbb{E}(Y_i(1)|X_i=c, Z_i(1)=z) w_1(z)f_{Z(1)|X}(z|c)dz,
$$
we can see that (\ref{unbiasedeq}) holds provided
\begin{equation}\label{eqpi1}
\pi_1(z)=\frac{f_X(c)}{2w_1(z)},
\end{equation}
where $f_X(c)$ is the p.d.f of $X$ at $x=c$. Following from the same argument, one can show that $\pi_0(z)=\frac{f_X(c)}{2w_0(z)}$.  In Table \ref{tab0}, we give the detailed expression of $\pi_1(z)$ and $\pi_0(z)$ for the proposed three estimands $\tau^{w1}_{SRD}$, $\tau^{w2}_{SRD}$ and $\tau^{w3}_{SRD}$. 
Since the weight $\pi_1(z)$ depends on the unknown density functions, we propose to estimate those densities by the following kernel estimators
\begin{equation}\label{eqkernelf1}
\begin{split}
\hat f_{X,Z(1)}(c,z) = 2 \cdot (nh_1^2)^{-1}\sum_{x_i>c}K_1(\frac{c-x_i}{h_1}, \frac{z-z_i}{h_1}),
\end{split}
\end{equation}
\begin{equation}\label{eqkernelf0}
\begin{split}
\hat f_{X,Z(0)}(c,z) = 2 \cdot (nh_1^2)^{-1}\sum_{x_i<c}K_1(\frac{c-x_i}{h_1}, \frac{z-z_i}{h_1}),
\end{split}
\end{equation}
\begin{equation}\label{eqkernelfz}
\hat f_{Z}(z)  =  (nh_2)^{-1}\sum_{i=1}^nK(\frac{z-z_i}{h_2}),~~\textrm{and}~~\hat f_{X}(c)  =  (nh_2)^{-1}\sum_{i=1}^nK(\frac{c-x_i}{h_2}),
\end{equation}
where $h_1$ and $h_2$ are bandwidth parameters. Note that the kernel estimator is known to suffer from the curse of dimensionality and is only applicable when the dimension of $Z_i$ is small. For the applications in which a large number of covariates can be collected, one may consider alternative parametric or semiparametric approaches for density estimation. In this work, we only focus on the above kernel estimators and leave the alternatives for future investigation.


Replacing the unknown density functions in $\pi_1(z)$ and $\pi_0(z)$ as shown in Table \ref{tab0} with the corresponding kernel estimators, we can obtain $\hat\pi_1(z)$ and $\hat\pi_0(z)$. 
While we can construct the final estimator by plugging $\hat\pi_1(z)$ into (\ref{eqkernel}), for practical use and theoretical analysis we recommend the local linear estimator, since it has smaller asymptotic bias and better finite sample behavior near the boundary \citep{fan1996local}. Motivated by the formulation of the kernel estimator (\ref{eqkernel}), we propose the following weighted local linear (WLL) estimator 
\begin{equation}\label{est}
\begin{split}
\Big(\hat{\alpha}_{0}, \hat{\beta}_{0}\Big) &= \arg \min_{\alpha, \beta} \sum_{\{i:X_i<c\}} \frac{1-T_i}{\hat\pi_0(Z_{i})} \Big( Y_i- \alpha - (X_i-c)\beta  \Big)^2 K\Big( \frac{X_i-c}{h}\Big),\\
\Big(\hat{\alpha}_{1}, \hat{\beta}_{1}\Big) &= \arg \min_{\alpha, \beta} \sum_{\{i:X_i>c\}} \frac{T_i}{\hat\pi_1(Z_{i})} \Big( Y_i- \alpha - (X_i-c)\beta\Big)^2 K\Big( \frac{X_i-c}{h}\Big),
\end{split}
\end{equation}
Thus, we can estimate the WATE ${\tau}_{SDR}^w$ by 
\begin{equation}\label{eqhattau}
\hat{\tau}_{SDR}^w = \hat{\alpha}_{1} - \hat{\alpha}_{0}. 
\end{equation}

\section{Theoretical Results}
\label{sec_asymptotics}

In this section, we study the asymptotic properties of the proposed estimator. We focus on the local linear estimator (\ref{est}) for the estimand $\tau^{w1}_{SRD}$, due to the nice properties of $\tau^{w1}_{SRD}$ as explained in Section 2.2. The analysis of the estimators for the other two estimands $\tau^{w2}_{SRD}$ and $\tau^{w3}_{SRD}$ is similar, but may require different technical conditions. 

Let $\delta$ denote a small positive constant. For notational simplicity, define $\cF^-$ as the class of functions of $x\in(c-\delta, c]$ and $z\in\cZ$ such that for any $f\in\cF^-$, $\frac{\partial^3}{\partial a\partial b\partial c}f(x,z)$ is continuous, where $a,b,c=\{x,z\}$. Here, the derivatives of $f(x,z)$ with respect to $x$ at $x=c$ (say $\frac{\partial^3}{\partial x^3}f(c,z)$) are interpreted as left derivatives. Similarly, we define $\cF^+$ as the class of functions of $x\in[c, c+\delta)$ and $z\in\cZ$ such that for any $f\in\cF^+$, $\frac{\partial^3}{\partial a\partial b\partial c}f(x,z)$ is continuous, where $a,b,c=\{x,z\}$. The derivatives at $x=c$ correspond to the right derivatives. For simplicity, we only consider the case that $\textrm{dim}(Z)=1$. The generalization to multivariate covariates follows from the similar argument.


\begin{assumption}[Smoothness condition]\label{ContinuityofConditionalDensity}
Assume $f_{X, Z(0)}(x,z)\in \cF^-$ and $f_{X, Z(1)}(x,z)\in \cF^+$. Denote $m_t(x,z)= \EE(Y_i(t)|X_i=x, Z_i(t)=z)$ for $t=\{0,1\}$. Assume $m_0(x,z)\in \cF^-$ and $m_1(x,z)\in \cF^+$.
\end{assumption}

Since the proposed method requires to estimate the unknown density functions $f_{X, Z(0)}(x,z)$ and $f_{X, Z(1)}(x,z)$, we assume they are sufficiently smooth. In addition, we also need the smoothness of $m_t(x,z)$ in order to study the local linear estimator (\ref{est}). This assumption implies Assumption \ref{asscont2}, which is required for identification purpose. 



\begin{assumption}\label{SharpRDdesign}
Assume the following conditions hold:
\begin{enumerate}
\item $X\in[x_l, x_u]$ such that $x_l<c<x_u$ and $Z$ also has bounded support.
\item The kernel $K(u)$ is a non-negative, symmetric and bounded function with compact support which satisfies
  \begin{displaymath}
    \int K(z)\,dz=1, ~~~\textrm{and}~~~ K(u) = K(-u).
  \end{displaymath}
The bivariate kernel $K_1(u,v)$ is a product kernel  $K_1(u,v)=K(u)K(v)$, where $K(\cdot)$ satisfies the above conditions. 
\item $f_{X,Z(1)}(x,z)$ is bounded way from 0 by a constant for $x\in[c, c+\delta)$ and $z\in\cZ$ and $f_{X,Z(0)}(x,z)$ is bounded way from 0 by a constant for $x\in(c-\delta, c]$ and $z\in\cZ$.
\item $\sigma_t^2=\EE[(Y(t)-m_{t}(X,Z(t)))^2|X=x,Z(t)=z]<\infty$ for $t\in\{0,1\}$.
\end{enumerate} 
\end{assumption}

Assumption \ref{SharpRDdesign} has four parts. The first and second parts are standard assumptions in kernel density estimation problems. Since our RD estimator applies local linear estimators at the boundary, the uniform kernel or triangular kernel has been shown to have good performance under such scenario \citep{Matias2017}. 
The third part requires $f_{X,Z(1)}(x,z)$ and $f_{X,Z(0)}(x,z)$ to be bounded away from 0 so that the inverse weights in the local linear estimator can be well controlled. This is similar to IPW estimator which requires the propensity score to be bounded away from 0.  The last part assumes that the (homoscedastic) noise has finite variance. 

Denote $\alpha_t=\int \mathbb{E}(Y(t)|X=c, Z(t)=z) f_Z(z)dz$, and recall that $\tau^{w1}_{SRD}=\alpha_1-\alpha_0$. The main theorem in this section shows the rate of convergence of the local linear estimators $\hat{\alpha}_{1}$ and $\hat{\alpha}_{0}$ and their limiting distributions. 
\begin{thm}\label{thm2}
Assume that assumptions \ref{ContinuityofConditionalDensity} and \ref{SharpRDdesign} hold, and we choose $h\asymp \sqrt{h_1}\asymp h_2$. Then
\begin{equation}\label{eqthm2_rate}
|\hat{\alpha}_{t} - \alpha_{t}| = O_p(h^2+(nh^2)^{-1/2}).
\end{equation}
Denote $d_t(x_i, z_i) = m_t(x_i, z_i) - \alpha_t$. Furthermore, if $h^{3}n^{1/2}=o(1)$ holds, we have
\begin{equation*}
\begin{split}
\sqrt{nh^2}(\hat{\alpha}_{1} - \alpha_{1}) & \sim  N\left(0,  C_v\cdot \mathbb{E}_Z\left(\frac{f_Z(z_i)}{f_{X, Z(1)}(c^+, z_i)}d_1(c^+, z_i)^2 \right)  \right),
\end{split}
\end{equation*}
and
\begin{equation*}
\begin{split}
\sqrt{nh^2}(\hat{\alpha}_{0} - \alpha_{0}) & \sim  N\left(0,  C_v\cdot\mathbb{E}_Z\left(\frac{f_Z(z_i)}{f_{X, Z(0)}(c^-, z_i)}d_0(c^-, z_i)^2 \right)  \right),
\end{split}
\end{equation*}
where $\EE_Z$ denotes the expectation under the marginal distribution of $Z$, and 
\[C_v = \frac{\kappa_2^2\kappa_{20}  + \kappa_1^2\kappa_{22} - 2\kappa_1\kappa_2\kappa_{21}}{\left(\frac{1}{2}\kappa_2-\kappa_1^2\right)^2},\]
with
\[\kappa_{q} = \int_{u>0} K(u)u^q du \quad \text{  and  } \quad \kappa_{2q} = \int_{u>0} K(u)^2u^q du, \quad \text{ for  } q = 0, 1, 2. \] 
\end{thm}
We note that from (\ref{eqthm2_rate}) the optimal choice of $h$ is of order $O(n^{-1/6})$ and the corresponding convergence rate is $|\hat{\alpha}_{t} - \alpha_{t}|=O_p(n^{-1/3})$ due to the boundary effect. Specifically, the estimated weights $\hat\pi_1(z)$ and $\hat\pi_0(z)$ depend on the density estimators at the boundary. Since we only require $f_{X, Z(0)}(x,z)\in \cF^-$ and $f_{X, Z(1)}(x,z)\in \cF^+$ to be smooth from one side,  the corresponding density estimators have a slower rate. So, the plug-in error becomes the dominant term when establishing the rate of $\hat{\alpha}_{t}$. However, if $Z_i$ are the pre-treatment covariates, i.e, $Z_i(1)=Z_i(0)=Z_i$, we can estimate the density $f_{X,Z}(c,z)$ by 
\begin{equation}\label{eqkernelf00}
\begin{split}
\hat f_{X,Z}(c,z) =  (nh_1^2)^{-1}\sum_{i=1}^n K_1(\frac{c-x_i}{h_1}, \frac{z-z_i}{h_1}).
\end{split}
\end{equation}
The following corollary shows that in this case $\hat{\alpha}_{t}$ has an improved rate. With an optimal choice of the bandwidth parameters, we prove that $|\hat{\alpha}_{t} - \alpha_{t}|=O_p(n^{-2/5})$, that is the boundary effect for estimating $\alpha_t$ is automatically removed without applying any additional bias correction procedures. 

\begin{cor}\label{corrd}
Assume that assumptions \ref{ContinuityofConditionalDensity} and \ref{SharpRDdesign} hold and  $Z_i$ are the pre-treatment covariates. Choosing $h\asymp h_1\asymp h_2$, we have
\begin{equation*}
\begin{split}
|\hat{\alpha}_{t} - \alpha_{t}| = O_p(h^2+(nh)^{-1/2}).
\end{split}
\end{equation*}
If $h^5n=o(1)$, we have
\begin{equation*}
\begin{split}
\sqrt{nh}(\hat{\alpha}_{t} - \alpha_{t}) & \sim  N\left(0,  C_v\cdot \sigma^2\int\frac{f_Z(z)^2}{f_{X, Z}(c, z)}dz   \right).
\end{split}
\end{equation*}

\end{cor}

\section{Simulation and Empirical Examples}
\label{sec_simulation}

\subsection{Simulation Study}

In this first setting, consider the following data generating process:
\[y_i(1) = 2+x_i+\beta z_i + \epsilon_i,\]
\[y_i(0) = 1+x_i+\beta z_i + \epsilon_i,\]
where $x_i$, $z_i$, and $\epsilon_i$ are generated independently from $N(0,1)$ distribution. The treatment $T_i$ is assigned at the cutoff 0: $T_i = \bold{1}(x_i>0)$. In this case, there is no discontinuity of the conditional distribution of $z_i$ given $x_i=0$. Thus, our estimaind $\tau^{w1}_{SRD}$ equals the standard RD estimand $\tau_{SRD}$ (both are equal to 1). We vary $\beta$ from 0 to 5 and compare weighted local linear (WLL) estimator with the standard RD estimator (\cite{Imbens2008}) in terms of bias, variance, root-mean-squared error (MSE), coverage probability of 95\% confidence intervals (Coverage) and its length (CI length).  When implementing both methods, we set the bandwidth parameter for standard RD estimator using cross-validation and then use the same bandwidth for WLL.
The results based on 500 simulations are shown in Table \ref{tab1}. When $\beta=0$, there is no covariates involved in the outcome function. Standard RD estimator performs neck to neck with our estimator. When $\beta\neq 0$, the standard RD estimator performs slightly better in terms of bias, however, our estimator consistently has smaller variance and MSE. 

\begin{table}[H]
\singlespacing
	\begin{center}
	\scalebox{0.8}{
		\begin{tabular}{llcccccccccc}
			\toprule 

&  &  \multicolumn{2}{c}{bias} &  \multicolumn{2}{c}{variance}  & \multicolumn{2}{c}{MSE}   & \multicolumn{2}{c}{Coverage}   & \multicolumn{2}{c}{CI length}   \\
 \cmidrule(r){3-4} \cmidrule(r){5-6} \cmidrule(r){7-8} \cmidrule(r){9-10}  \cmidrule(r){11-12}
       n & $\beta$       &    RD & WLL   & RD & WLL  & RD & WLL & RD & WLL  & RD & WLL \\
\midrule
\\
\multirow{5}{*}{500} & 0   & -0.0106 & \textbf{-0.0102} & 0.3731 & \textbf{0.3767} & 0.3733 & \textbf{0.3768}& 0.9650 & 0.9750 & 1.6666 & 1.6853\\
                     & 1   & 0.0236 & \textbf{0.0226} & 0.6090 & \textbf{0.5423} & 0.6094, & \textbf{0.5427} & 0.9250 & 0.9400 & 2.2335 & 2.1753\\
                     & 2   & -0.0583 & \textbf{-0.0548} & 0.9004& \textbf{0.7886} & 0.9023 & \textbf{0.7905}& 0.9600 & 0.9600 & 3.6103 & 3.2369\\
                     & 5   & -0.0522  & \textbf{-0.0423} & 2.1757 & \textbf{1.8264} & 2.1763 & \textbf{1.8268}& 0.8800 & 0.9000 & 7.1005 & 6.1439\\
                     \\
                     \hline
                     \\
\multirow{5}{*}{1000} & 0   & \textbf{0.0021} & 0.0034 & 0.2853 & \textbf{0.2841} & 0.2854 & \textbf{0.2841}& 0.9900 & 0.9900 & 1.3447 & 1.4083\\
                     & 1    & \textbf{-0.0088} & -0.0118 & 0.4379 & \textbf{0.4387} & 0.4740& \textbf{0.4388} &0.9200 & 0.9450 & 1.6833 & 1.6661\\
                     & 2   & \textbf{0.0380} & 0.0431 & 0.6046 & \textbf{0.5118} & 0.6058 & \textbf{0.5136} &0.9400 & 0.9550 & 2.2516 & 2.2635\\
                     & 5   & \textbf{0.1106} & 0.1129 & 1.4676 & \textbf{1.2770} & 1.4718 & \textbf{1.2820} &0.9600 & 0.9850 & 6.3293 & 6.9139\\
                     \\
                     \hline
                     \\
\multirow{5}{*}{2000} & 0   & \textbf{-0.0107} & -0.0111 & 0.2098 & \textbf{0.2096} & 0.2101 & \textbf{0.2099} & 0.9300 & 0.9200 & 0.7400 & 0.7229\\
                     & 1   & -0.0320 & \textbf{-0.0225} & 0.3151 & \textbf{0.2920} & 0.3167 & \textbf{0.2929} &0.9850 & 0.9850 & 1.3572 & 1.3237\\
                     & 2   & -0.0177 & \textbf{-0.0151} & 0.5234 & \textbf{0.4595} & 0.5237 & \textbf{0.4598} &0.9500 & 0.9650 & 2.0911 & 1.9680\\
                     & 5   & \textbf{0.1331} & 0.1360 & 1.1707 & \textbf{1.0019} & 1.1783 & \textbf{1.0111} &0.9150 & 0.9200 & 4.2082 & 3.7723\\                     
                     \\
                     \hline
                     \\
\multirow{5}{*}{5000} & 0   & -0.0149 & \textbf{-0.0136} & \textbf{0.1606} & 0.1616 & \textbf{0.1613} & 0.1622 &0.9200 & 0.9250 & 0.5699 & 0.5725\\
                     & 1   & \textbf{0.0034} & 0.0049 & 0.2451 & \textbf{0.2275} & 0.2451 & \textbf{0.2275} &0.9350 & 0.9400 & 0.9012 & 0.8793\\
                     & 2   & \textbf{-0.0130} & -0.0137 & 0.3301 & \textbf{0.2940} & 0.3304 & \textbf{0.2943} &0.9450 & 0.9600 & 1.3654 & 1.2931\\
                     & 5   & \textbf{0.0053} & 0.0139 & 0.8355 & \textbf{0.7016} & 0.8355 & \textbf{0.7018} &0.9250 & 0.9750 & 3.1578 & 3.1535\\                     

 \bottomrule
		\end{tabular}}
	\end{center}
\vspace{-.2in}
\caption{Comparison of the standard RD estimator and the proposed weighted local linear estimator (WLL) in the first setting.} \label{tab1}
\end{table}

Figure \ref{bandwidth} compares the MSE of the standard RD estimator with our WLL estimator across different bandwidth choices. The figure is consistent with corollary \ref{corrd} as our estimator is asymptotically more efficient through including additional covariates into the estimation. Moreover, the advantage of WLL estimator is greater when perform under-smoothing and the difference of the two estimators becomes smaller as the bandwidth increases.

\begin{figure}[H]
\begin{center}
\subfigure[][]{
\label{bandwidth200}
    \includegraphics[height=2.2in]{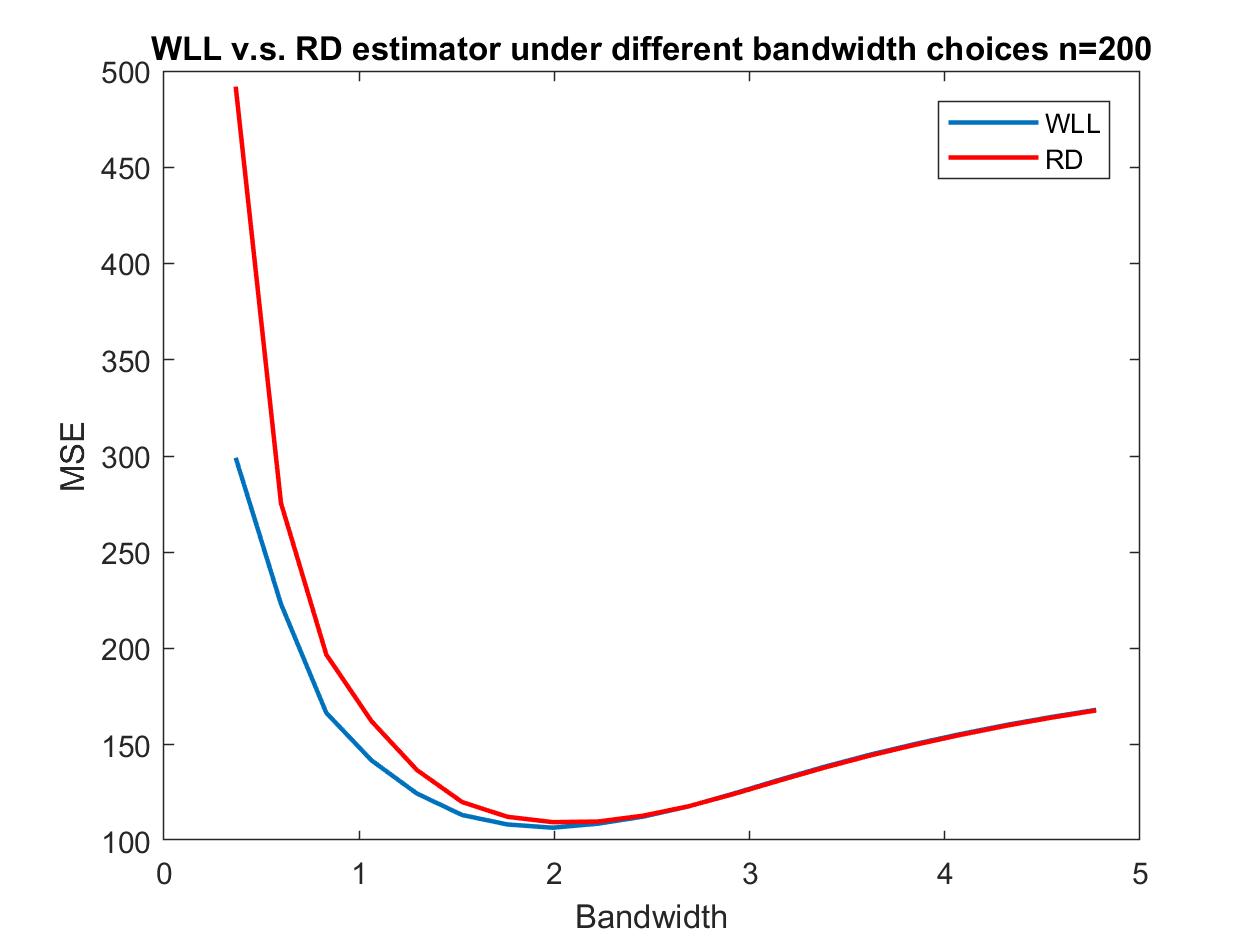}}
    \subfigure[][]{
    \hspace{8pt}
    \label{bandwidth500}
    \includegraphics[height=2.2in]{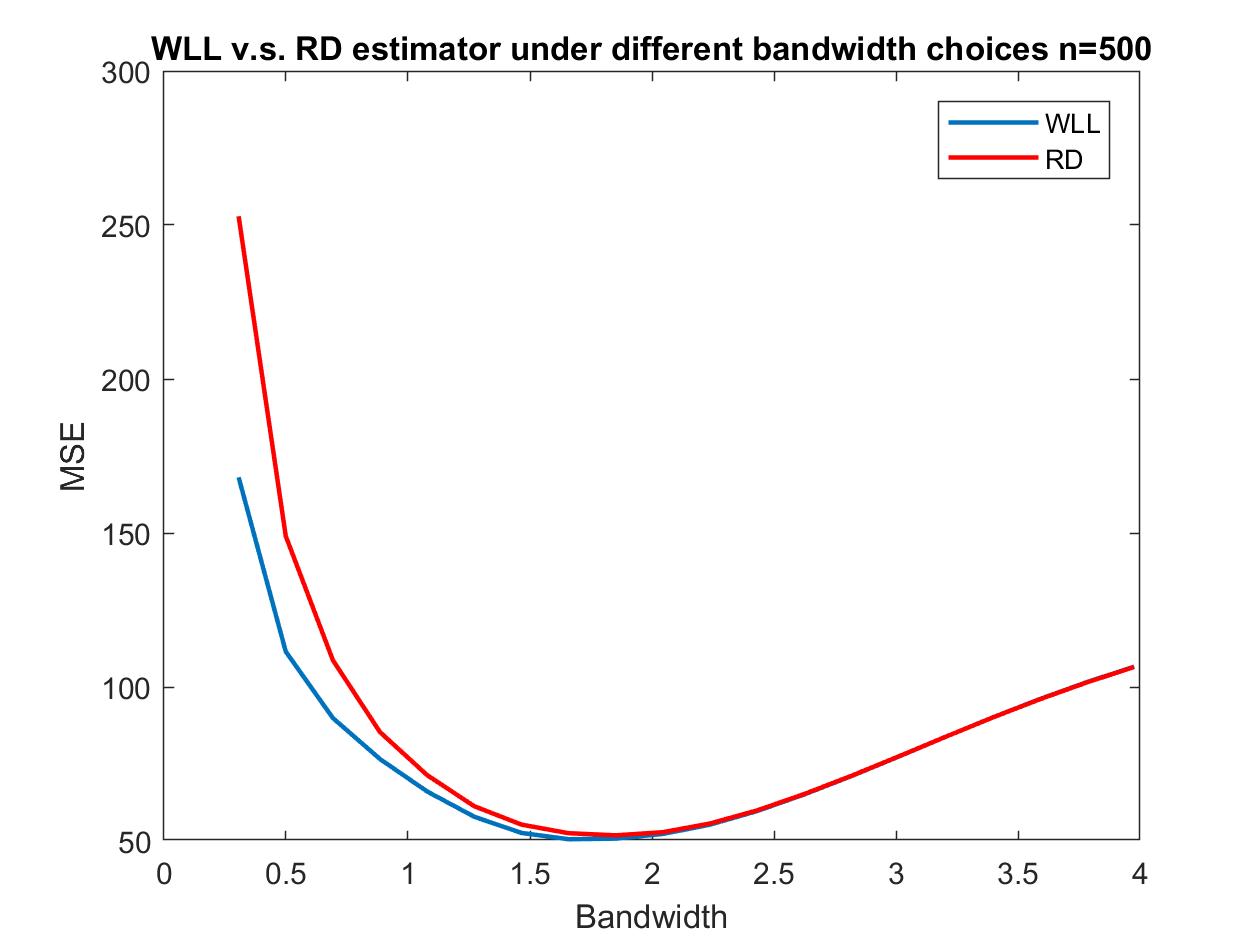}
    }\\

\subfigure[][]{  
\label{bandwidth1000}  
    \includegraphics[height=2.2in]{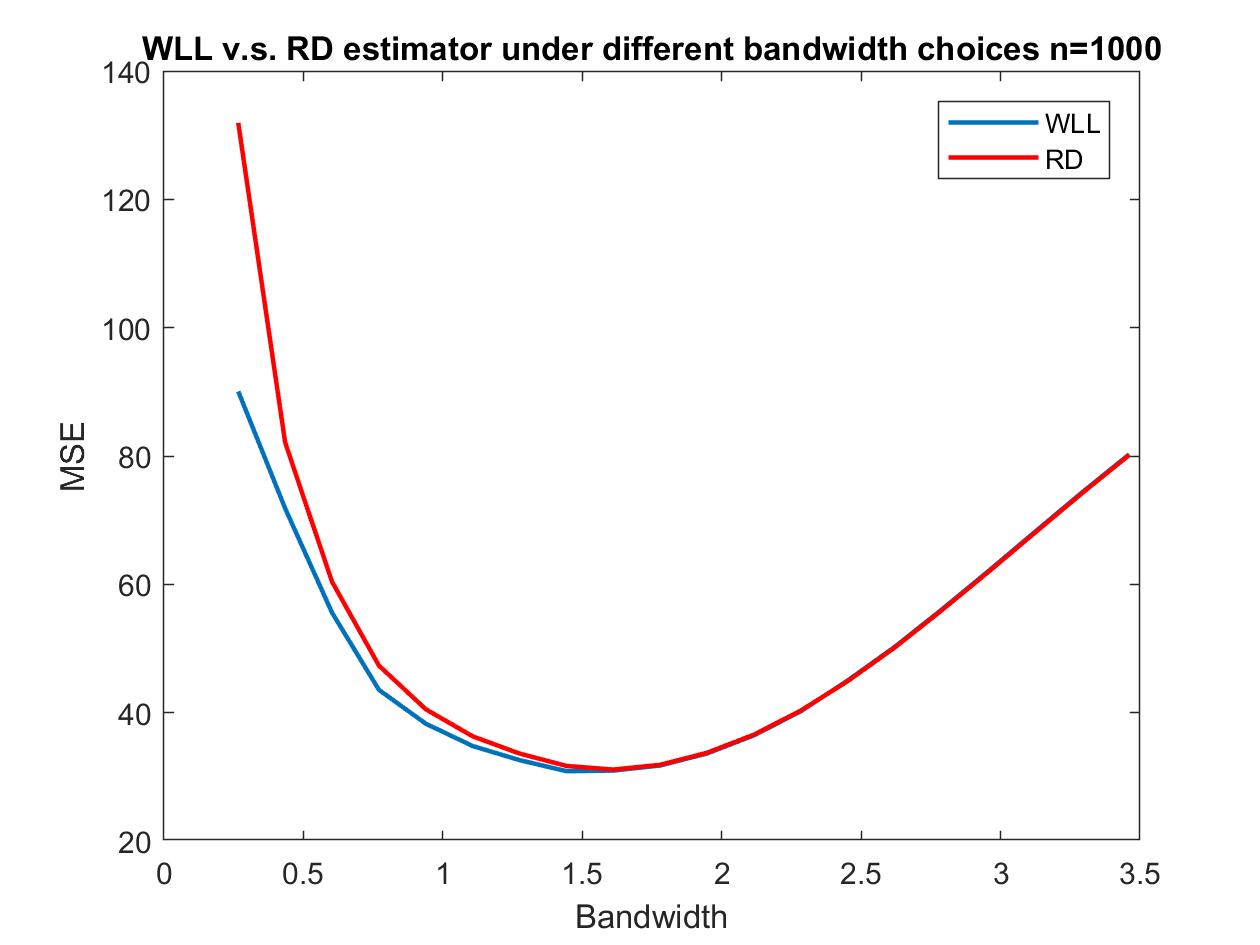}}
    \hspace{8pt}
    \subfigure[][]{
    \label{bandwidth2000}
    \includegraphics[height=2.2in]{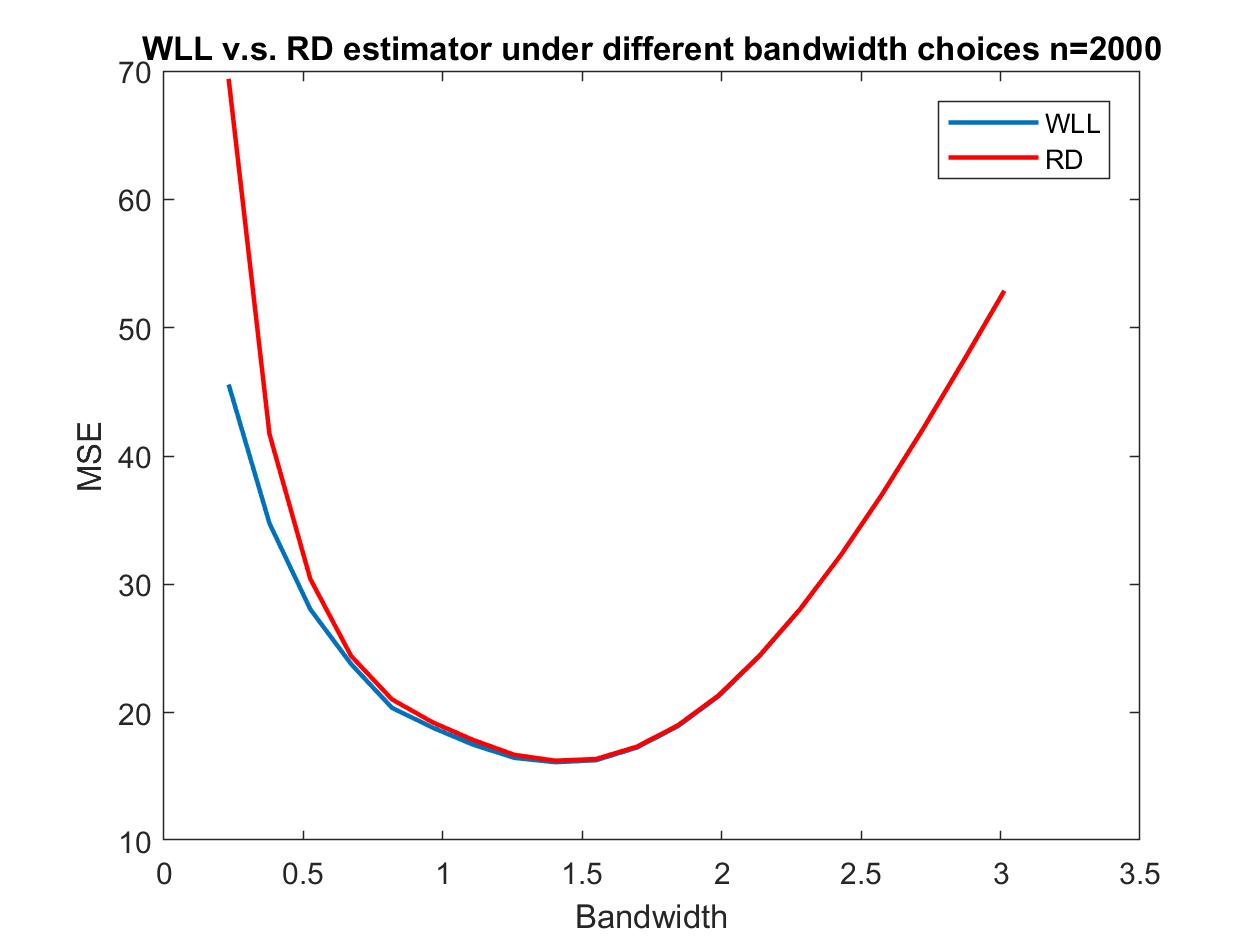}
}

\end{center}
\caption{Comparing MSE of WLL and RD estimator under the same bandwidth choices }
\footnotesize{\singlespace\justify  
}\label{bandwidth}
\end{figure}


In the second setting, we consider the following data generating process:
\[y_i(1) = 3+x_i+z_i + \epsilon_{1i},\]
\[y_i(0) = 1+x_i+z_i + \epsilon_{0i},\]
where $x_i$ and $\epsilon_i$ are generated independently from $N(0,1)$ distribution, however, $z_i$ is generated from another independent $N(0,1)$ process with a discontinuity at $X>0$, i.e. $z_i = \gamma \cdot \bold{1}(x_i>0)+z_i^*$, where $z_i^* \sim N(0,1)$. The treatment $T_i$ is assigned at the cutoff 0: $T_i = \bold{1}(x_i>0)$. When $\gamma=0$, again there is no discontinuity of the conditional distribution of $z_i$ given $x_i=0$. Both our estimaind $\tau^{w1}_{SRD}$ and the standard RD estimand $\tau_{SRD}$ are equal to 2.  However, as $\gamma$ differs from 0, the conditional distribution of $z_i$ given $x_i$ is discontinuous at $x_i=0$.  Our estimand $\tau^{w1}_{SRD}$ is still equal to 2, which is the direct causal effect of interest. If we adopt the standard RD framework and ignore the discontinuity of the conditional distribution of $z_i$ given $x_i$, we would expect that the standard RD estimator is biased for estimating the direct causal effect of interest (which is 2 in this example). In the data generating process, we vary $\gamma$ from 0 to 1 and compare our estimator with the standard RD estimator. The results are shown in Table \ref{tab2}. The standard RD estimator has large bias and very poor coverage probability when $\gamma$ is close to 1, which agrees with our expectation. In contrast, the proposed estimator has relatively small MSE and accurate coverage probabilities across different choices of the sample size $n$ and the parameter $\gamma$. In summary, our simulation studies confirm that one should apply the proposed framework to the RD study if there exists some potential discontinuity of the conditional distribution of the covariates given the running variable.

\begin{table}[H]
\singlespacing
	\begin{center}
	\scalebox{0.8}{
		\begin{tabular}{llcccccccccc}
			\toprule 

&  &  \multicolumn{2}{c}{bias} &  \multicolumn{2}{c}{variance}  & \multicolumn{2}{c}{MSE}   & \multicolumn{2}{c}{Coverage}   & \multicolumn{2}{c}{CI length}  \\
 \cmidrule(r){3-4} \cmidrule(r){5-6} \cmidrule(r){7-8} \cmidrule(r){9-10}  \cmidrule(r){11-12} 
       n & $\gamma$       &    RD & WLL   & RD & WLL  & RD & WLL & RD & WLL & RD & WLL\\
\midrule
\\
\multirow{5}{*}{500}     & 0.2 & 0.1858 & \textbf{0.0762} & 0.6181 & \textbf{0.4318} & 0.4166 & \textbf{0.1922} & 0.9350 & 0.9550 & 2.4229 & 1.6925\\
 & 0.4 & 0.4457 & \textbf{0.1150} & 0.5519 & \textbf{0.4327} & 0.5032 & \textbf{0.2005} & 0.8700 & 0.9450 & 2.1634 & 1.6963\\
 & 0.6 & 0.5721 & \textbf{0.1817} & 0.5898 & \textbf{0.4913} & 0.6751 & \textbf{0.2744} & 0.8500 & 0.9400 & 2.3119 & 1.9260\\
 & 0.8 & 0.7620 & \textbf{0.2445} & 0.5686 & \textbf{0.4698} & 0.9040 & \textbf{0.2805} & 0.7100 & 0.9250 & 2.2290 & 1.8416\\
 & 1 & 0.9753 & \textbf{0.3628} & 0.5346 & \textbf{0.4404} & 1.2369 & \textbf{0.3256} & 0.5750 & 0.8800 & 2.0956 & 1.7264\\
                     \\
                     \hline
                     \\
\multirow{5}{*}{1000}   & 0.2 & 0.1851 & \textbf{0.0869} & \textbf{0.4124} & 0.4897 & \textbf{0.2044} & 0.2473 & 0.9200 & 0.9650 & 1.6168 & 1.9196\\
 & 0.4 & 0.4082 & \textbf{0.1271} & 0.4101 & \textbf{0.3360} & 0.3349 & \textbf{0.1291} & 0.8450 & 0.9350 & 1.6077 & 1.3171\\
 & 0.6 & 0.5534 & \textbf{0.1795} & 0.4034 & \textbf{0.3428} & 0.4690 & \textbf{0.1497} & 0.7200 & 0.9200 & 1.5814 & 1.3437\\
 & 0.8 & 0.8240 & \textbf{0.2458} & 0.4625 & \textbf{0.3859} & 0.8930 & \textbf{0.2093} & 0.5750 & 0.900 & 1.8131 & 1.5125\\
 & 1 & 1.0130 & \textbf{0.3193} & \textbf{0.4445} & 0.5142 & 1.2237 & \textbf{0.3664} & 0.3950 & 0.9400 & 1.7423 & 2.0157\\
                    \\
                     \hline
                     \\
\multirow{5}{*}{2000} & 0.2 & 0.2052 & \textbf{0.0592} & 0.3162 & \textbf{0.2705} & 0.1421 & \textbf{0.0767} & 0.9050 & 0.9700 & 1.2393 & 1.0602\\
 & 0.4 & 0.3755 & \textbf{0.0887} & 0.3072 & \textbf{0.2500} & 0.2354 & \textbf{0.0704} & 0.7550 & 0.9400 & 1.2042 & 0.9802\\
 & 0.6 & 0.5787 & \textbf{0.1697} & 0.3196 & \textbf{0.2568} & 0.4370 & \textbf{0.0947} & 0.5500 & 0.9000 & 1.2529 & 1.0067\\
 & 0.8 & 0.8448 & \textbf{0.2208} & 0.3173 & \textbf{0.2979} & 0.8143 & \textbf{0.1375} & 0.2450 & 0.9100 & 1.2438 & 1.1676\\
 & 1 & 1.0107 & \textbf{0.2234} & \textbf{0.2876} & 0.3556 & 1.1043 & \textbf{0.1763} & 0.0450 & 0.9500 & 1.1272 & 1.3938\\
                    \\
                     \hline
                     \\
\multirow{5}{*}{5000} & 0.2 & 0.2110 & \textbf{0.0418} & 0.2249 & \textbf{0.1748} & 0.0951 & \textbf{0.0323} & 0.8300 & 0.9450 & 0.8817 & 0.6852\\
 & 0.4 & 0.4138 & \textbf{0.0650} & 0.2397 & \textbf{0.1947} & 0.2287 & \textbf{0.04212} & 0.6200 & 0.9450 & 0.9396 & 0.7631\\
 & 0.6 & 0.6099 & \textbf{0.1301} & 0.2319 & \textbf{0.2230} & 0.4257 & \textbf{0.06666} & 0.3000 & 0.9350 & 0.9091 & 0.8741\\
 & 0.8 & 0.8327 & \textbf{0.1826} & 0.2305 & \textbf{0.2741} & 0.7465 & \textbf{0.1085} & 0.0500 & 0.9150 & 0.9035 & 1.0744\\
 & 1 & 1.0259 & \textbf{0.2331} & 0.2250 & \textbf{0.2130} & 1.1032 & \textbf{0.0997} & 0.0050 & 0.9750 & 0.8821 & 0.8350\\

 \bottomrule
		\end{tabular}}
	\end{center}
\caption{Comparison of the standard RD estimator and the proposed weighted local linear estimator (WLL) in the second setting.} \label{tab2}
\end{table}

\subsection{Incumbency Advantage}
We further apply our method to study the ``incumbency advantage" in the U.S. House elections as in \cite{Lee2008}. The ``incumbency advantage" states that current incumbent party in a district are more likely to win the next. The dataset contains 6560 observations on elections to the United States House of Representatives (1946-1998). We evaluate the probability a Democrat both running in and winning election $t+ 1$ as a function of the Democratic vote share margin of victory in election $t$. The other covariates we are considering includes the Democratic vote share at time $t-1$, Democrat winning at $t-1$, Democrat's political experience, opponent's political experience, Democrat's electoral experience, opponent's electoral experience. This is the same setting as studied in \cite{Lee2008} where he uses local 4th order polynomial RD estimates and includes these additional covariates as robustness check. We apply our method in the same five settings and results are reported in Table \ref{tab3} below:

\begin{table}[H]
\begin{center}
\singlespacing
\scalebox{0.8}{
		\begin{tabular}{lccccccccccc}
			\toprule 		
  Vote share $t+1$  &&  \multicolumn{2}{c}{(1)} &  \multicolumn{2}{c}{(2)}  & \multicolumn{2}{c}{(3)}  &     \multicolumn{2}{c}{(4)} &  \multicolumn{2}{c}{(5)}   \\
\\
 && Polyn & WLL & Polyn & WLL & Polyn & WLL & Polyn & WLL & Polyn & WLL \\ 
 \cmidrule(r){3-4}\cmidrule(r){5-6}\cmidrule(r){7-8} \cmidrule(r){9-10}  \cmidrule(r){11-12}
Victory, election $t$ && 0.077   & 0.076    &  0.078  & 0.076   &  0.077  & 0.076   & 0.077   & 0.076   & 0.078   & 0.076\\
                      && (0.011) & (0.010) & (0.011) & (0.011) & (0.011) & (0.010) & (0.011) & (0.010) & (0.011) & (0.010)\\
                      \\
Dem. vote share   && -   & -    &  Y  & Y   &  -  & -   & -   & -   & Y   & Y  \\
$t-1$\\
\\
Dem. win, $t-1$ && -   & -    &  Y  & Y   &  -  & -   & -   & -   & Y   & Y  \\
\\
Dem. political && -   & -    &  -  & -   &  Y  & Y   & -   & -   & Y   & Y  \\
experience\\
\\
Opp. political && -   & -    &  -  & -   &  Y  & Y   & -   & -   & Y   & Y  \\
experience\\
\\
Dem. electoral && -   & -    &  -  & -   &  -  & -   & Y   & Y   & Y   & Y  \\
experience\\
\\
Opp. electoral && -   & -    &  -  & -   &  -  & -   & Y   & Y   & Y   & Y  \\
experience \\
 \bottomrule
		\end{tabular}}
\end{center}
\caption{Effect of winning an election on subsequent party electoral success. Polyn method is replicated as in \cite{Lee2008} using 4th order polynomial. WLL represents the results using our method. $Y$ indicates a covariates is been included in the calculation. Standard error is provided in the bracket using bootstrap for both method.} \label{tab3}
\end{table}

As shown in Table \ref{tab3}, there is no significant change on the estimated effect when using our method. The main reason is that the conditional densities of covariates are truly continuous at the cutoff in this data set. We expect that both standard RD estimator and our method are valid and therefore the results are very similar.  

\subsection{Generalized Second Price Auction (GSP)}

Next we apply our method to study the generalized second price auction (GSP) problem. GSP is an auction mechanism for multiple items and it has been used widely for the assignment of advertisement positions by internet search engine like google and bing. Let $n$ be the number of bidders, and let $b^1 \geq b^2 \geq \cdots\geq b^n$ be the bids from high to low. Denote by $v_{(1)}, v_{(2)}, \cdots, v_{(n)}$ the bidders' valuation associated with the rank of bids and $r^k$ the click through rate for the $k$th position. The $k$th bidder's payoff in a GSP is given as $(v_{(k)} - b^{k+1})r^k$.     

An important metrics is the click through rate $r^k$ for the $k$th position. Bidders are interested in the potential growth in their search traffic by winning the auction. And furthermore, in a Vickrey-Clarke-Groves (VCG) auction, $r^k$ will determine the total cost for placing each bidder in the sponsored advertisement region. In real world, GSP is usually implemented through a reservation score. A search score is formed for every bidder based on their bid and other quality measures. When the search score is bigger than a pre-set reservation score, the bidder's link will be displayed in the sponsored area. Otherwise, they will be displayed after all the sponsored advertisements. The reservation score cut-off creates a natural regression discontinuity setting to evaluate $r^k$.   

We study the Microsoft Bing search data from Oct 2nd to Oct 22nd in 2015 and estimate the effect of advertisement positions. We focus on a set of searches with first advertisement positions displayed but without third advertisement position displayed.\footnote{Microsoft Bing allows a maximum of 4 advertisement to be displayed at the time of study.} This allows us to analyze the effect of second advertisement positions by comparing click-abilities for bidders near the search score cutoff. Figure \ref{scorevsclick} plots the connection between search score and click-ability for the second position in the sponsored advertisement area. Once the score passed the cut-off at 15.17, the customers' links will be placed at the sponsored advertisement area and a significant increase in the click traffic can be observed.     
\begin{figure}
\begin{center}
\includegraphics[scale=0.30]{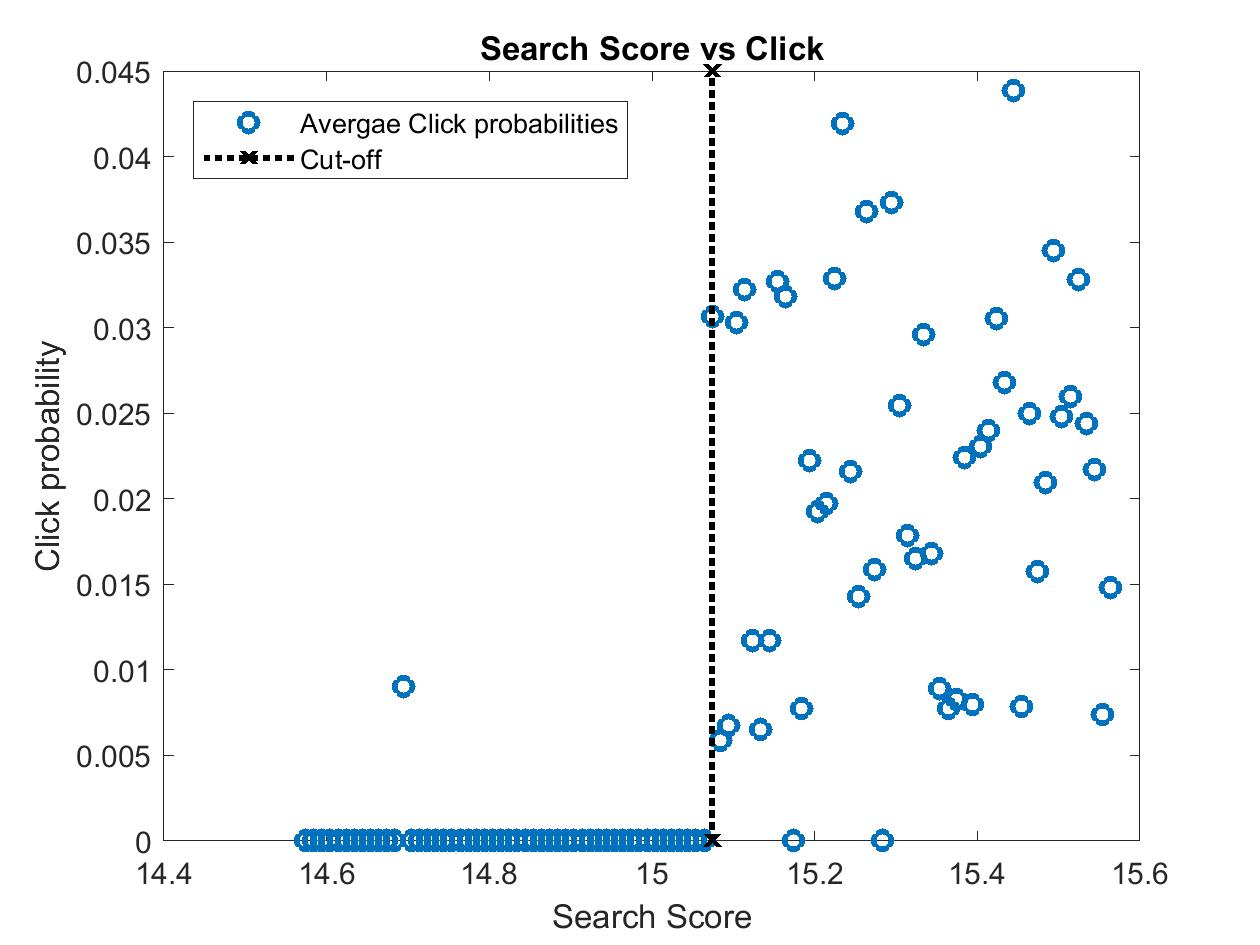}
\end{center}
\caption{Plot of search score and click-ability. If the search score is above 15.17, the bidder's advertisement will be displayed at the second position. Otherwise, the advertisement will be displayed based on its regular search orders.      }\label{scorevsclick}
\end{figure}
Consider the bid as a covariate. Figure \ref{scorevsbid} plots the mean bidding price before and after the search score cut-off. The mean bids before the cut-off is higher than the mean bids after the cut-off, implying the discontinuity of the conditional distribution of the covariate. See also Figure \ref{density} for the conditional density of the covariate before and after the cut-off. Although a local envy free equilibrium exists when bidders are all bidding their true valuation \citep{Edelman2007}, information asymmetry or bidder inertia may still lead to bidder selections. For example, active bidders may have the incentive to bid more aggressively to take advantage of the bidders with high inertia.     
\begin{figure}
\begin{center}
\includegraphics[scale=0.30]{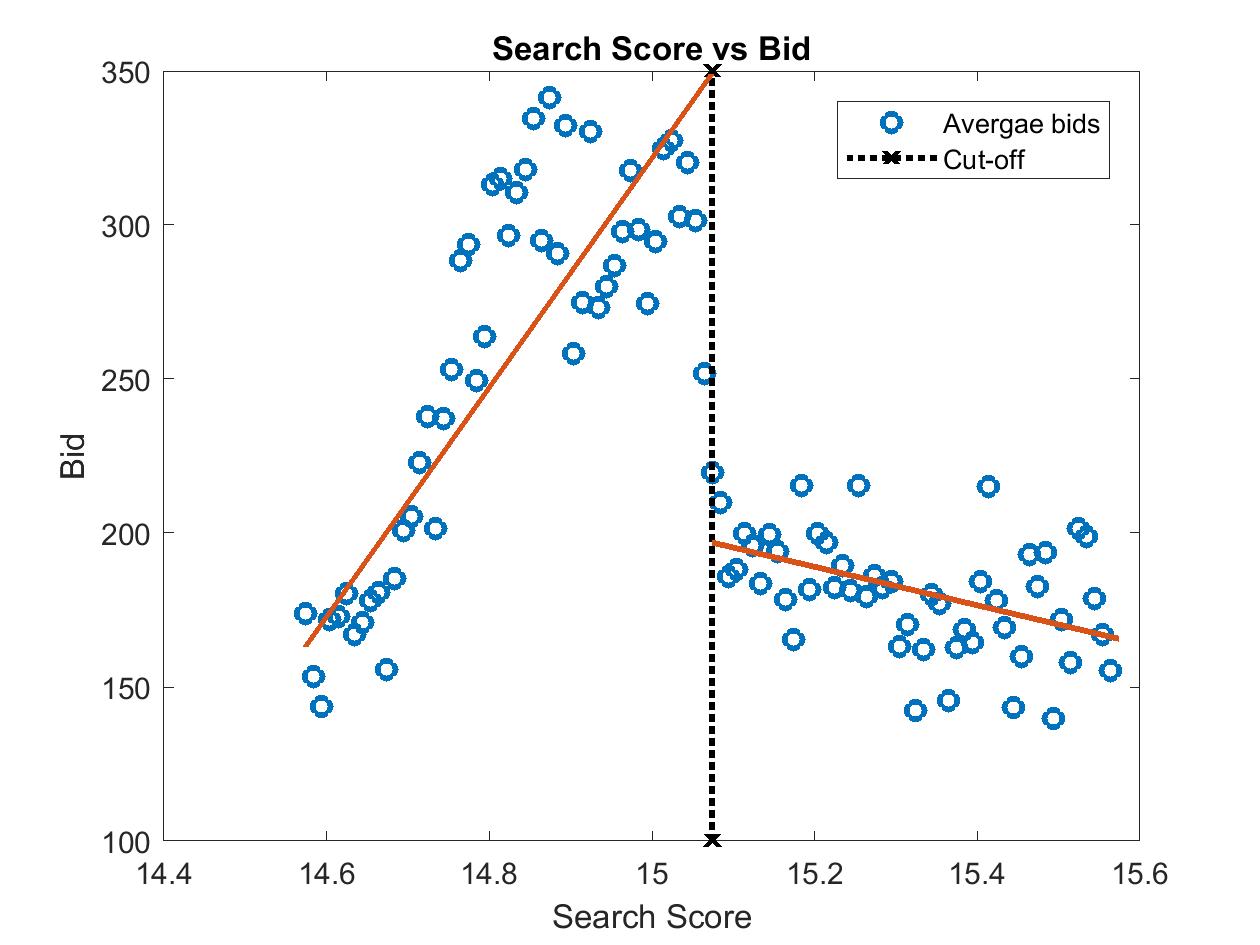}
\end{center}
\caption{Search Score v.s. Bids}\label{scorevsbid}
\end{figure}

\begin{figure}
\begin{center}
\includegraphics[scale=0.30]{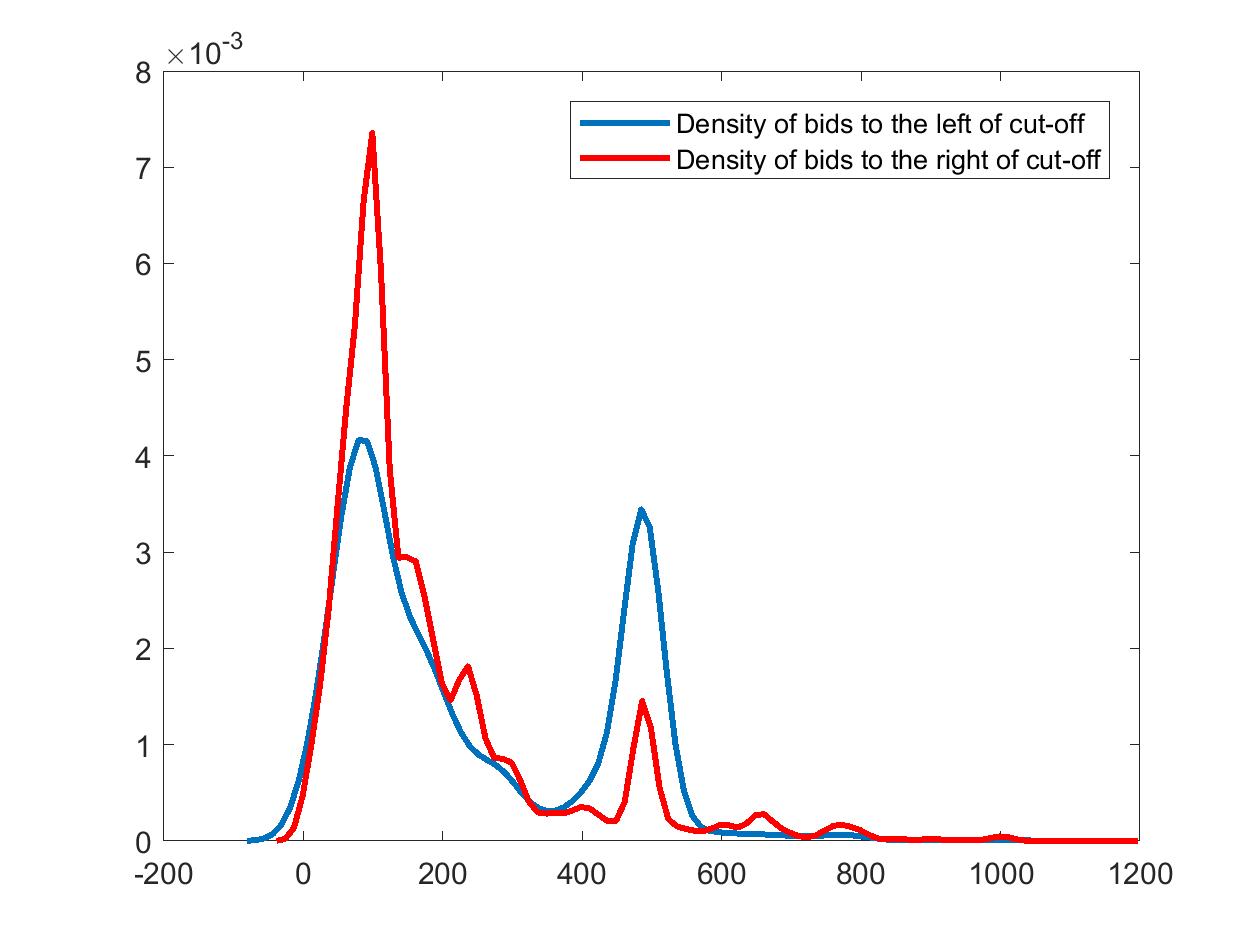}
\end{center}
\caption{Density plot for bids before and after the cut-off}\label{density}
\end{figure}
Table \ref{tab4} presents the results of our estimator and a polynomial RD estimator. The classic RD estimator may not be valid in this case due to the discontinuity in the covariates. It estimates that placing the advertisement on the second position can increase the click-ability by 1.91\% and it is statistically significant. On the other hand, the proposed estimator delivers only 1.20\%, 57\% less than the RD estimator and it is not statistically significant at 5\% level. The difference is mainly because low quality bidders with high willingness to pay have the incentive to bid higher to move their search scores pass the threshold. But when we match bidders with similar bids before and after the cut-off, the effect goes away.  
\begin{table}[H]
\begin{center}
\singlespacing
\scalebox{0.8}{
		\begin{tabular}{lcc}
		\toprule
  & RD & WLL\\
  Estimates & 1.91\%$^{***}$ & 1.20\%\\
  Standard Error & 0.0033 & 0.0090\\
  \bottomrule
		\end{tabular}}
\end{center}
\vspace{-.2in}
\caption{Estimated effect of adversing at second position using the standard RD estimator and the proposed method. The bandwidth parameters are selected using cross-validation and standard errors are obtained via bootstrap. $^{***}$ represents significance at $p<0.05$} \label{tab4}
\end{table}

\section{Extension to Fuzzy RD Design}
\label{sec_Fuzzy}

Fuzzy RD design refers to the case of a RD design when the treatment compliance is imperfect. Although $T_i$ is no longer a deterministic function of $X_i$, there is still a discontinuity in $\mathbb{P}(T_i|X_i, Z_i)$. \cite{Hahn2001} showed the equivalence between the fuzzy RD design and local average treatment effect (LATE). Under independence and monotonicity assumption, the cutoff can be used as an instrument for the treatment status and the LATE  can be interpreted as an intention-to-treat effect on the compliers: subjects take treatment as their assigned one. 

Introducing self-selection under the Fuzzy RD design is the same as adding an additional endogenous variable. As a result, additional instrumental variable is required for identification. We consider a scenario when the self-selection is based on the cut-off instead of the treatment assignment mechanism. This setting allows us to deal with unobservables that affect the covariates. Let $\tau$ denote the type of individual: complier ($co$), always-taker ($at$) and never taker ($nt$). A post-selection may exist and we can write 
\[Z_i = \bold{1}(X_i>c)Z_i(1) + \bold{1}(X_i\leq c)Z_i(0).\]
Notice that in sharp RD design, the self-selection based on the cut-off is the same as the selection on treatment assignment mechanism since the assignment mechanism is a deterministic function of $X_i$. We study the selection based on the cut-off in fuzzy RD design as it is a more practical design. For example, students may notice that with SAT score higher than a threshold may increase their chance to receive scholarship, but the true assignment mechanism (gender, family income, etc.) may not be disclosed to them.      

To consider self-selection under this setting, we first need to modify the independence and monotonicity assumption to incorporate additional covariates.    

\begin{assumption}[Fuzzy RD design].\label{fuzzyrd}
\begin{enumerate}
\item Independence: $\mathbb{P}(\tau_\epsilon  | X_i = c^+, Z_i(1) = z) = \mathbb{P}(\tau_\epsilon | X_i = c^-, Z_i(0) = z)$ for $\tau_\epsilon \in \{co, at, nt\}$.
\item Monotonicity: $\mathbb{P}(\tau_\epsilon = co) + \mathbb{P}(\tau = at) + \mathbb{P}(\tau_\epsilon = nt) = 1$
\item Exclusion: 
\begin{align*}
\mathbb{E}(Y_i(1)|Z_i(1)=z, \tau_\epsilon = at) = \mathbb{E}(Y_i(1)|Z_i(0)=z, \tau_\epsilon = at)\\
\mathbb{E}(Y_i(1)|Z_i(1)=z, \tau_\epsilon = nt) = \mathbb{E}(Y_i(1)|Z_i(0)=z, \tau_\epsilon = nt)
\end{align*}
\end{enumerate}
\end{assumption}

Due to self-selection, $Z_i$ is discontinuous at $X_i = c$, thus the independence assumption needs to be modified to insure individual types are remain constant under self-selection. Monotonicity assumption is essential to the identification by restricting the direction of selection. We may only have compliers, always-takers, never-takers in the population. As a result, we can view the treatment as an instrument as in \cite{Imbens1994}.

By allowing types of individuals to depend on $Z_i$, our framework allows selection on treatment compliance to vary across $Z_i$. For example, individuals with family incomes less than a threshold are more likely to be always takers for the scholarship. An extreme scenario when $Z_i$ is binary would be that always takers all have $Z_i=1$. In a similar spirit to the sharp RD design, we consider the following three estimands under our WATE framework:  
\begin{equation} \label{fuzzy}
\begin{split}
\mathbb{E}_{Z(0)}\Big(\mathbb{E}(Y_i(1)|Z_i(1)=z, \tau_\epsilon = co) - \mathbb{E}(Y_i(0)|Z_i(0)=z, \tau_\epsilon = co)\Big|  X_i = c^-, \tau_\epsilon = co\Big),
\end{split}
\end{equation}
\begin{equation} \label{fuzzy2}
\begin{split}
\lim_{\epsilon \rightarrow 0} \mathbb{E}\Big(\mathbb{E}(Y_i(1)|Z_i(1)=z, \tau_\epsilon = co) - \mathbb{E}(Y_i(0)|Z_i(0)=z, \tau_\epsilon = co)\Big|X \in \mathcal{N}_\epsilon,  \tau_\epsilon = co\Big),
\end{split}
\end{equation}
\begin{equation} \label{fuzzy3}
\begin{split}
\mathbb{E}_{Z}\Big(\mathbb{E}(Y_i(1)|Z_i(1)=z, \tau_\epsilon = co) - \mathbb{E}(Y_i(0)|Z_i(0)=z, \tau_\epsilon = co)\Big),
\end{split}
\end{equation}
where $\mathcal{N}_\epsilon$ is a symmetric $\epsilon$ neighborhood around $c$.
\begin{thm}[Identification under Fuzzy RD Design]\label{thm4}
Under assumption \ref{fuzzyrd}, we have 
\begin{align*}
&\mathbb{E}_{Z(0)}\Big(\mathbb{E}(Y_i(1)|Z_i(1)=z, \tau_\epsilon = co) - \mathbb{E}(Y_i(0)|Z_i(0)=z, \tau_\epsilon = co)\Big|  X_i = c^-, \tau_\epsilon = co\Big)\\
& = \frac{\mathbb{E}_{Z(0)}\left(\mathbb{E}(Y_i|X_i = c^+,Z_i=z) - \mathbb{E}(Y_i|X_i = c^-,Z_i=z)\Big| X_i = c^-\right)}{\mathbb{E}_{Z(0)}\left(\mathbb{E}(T_i|X_i = c^+,Z_i=z) - \mathbb{E}(T_i|X_i = c^-,Z_i=z)\Big| X_i = c^-\right)},\\
&\lim_{\epsilon \rightarrow 0} \mathbb{E}\Big(\mathbb{E}(Y_i(1)|Z_i(1)=z, \tau_\epsilon = co) - \mathbb{E}(Y_i(0)|Z_i(0)=z, \tau_\epsilon = co)\Big|X \in \mathcal{N}_\epsilon,  \tau_\epsilon = co\Big)\\
& = \frac{\int \left(\mathbb{E}(Y_i|X_i = c^+,Z_i=z) - \mathbb{E}(Y_i|X_i = c^-,Z_i=z)\right) (f_{Z(0)|X}(z_i|c) + f_{Z(1)|X}(z_i|c)) dz }{\int \left(\mathbb{E}(T_i|X_i = c^+,Z_i=z) - \mathbb{E}(T_i|X_i = c^-,Z_i=z)\right) (f_{Z(0)|X}(z_i|c) + f_{Z(1)|X}(z_i|c)) dz},\\
&\text{If we further have the conditional independence assumption (CIA), }\\
&\mathbb{E}_{Z}\Big(\mathbb{E}(Y_i(1)|Z_i(1)=z, \tau_\epsilon = co) - \mathbb{E}(Y_i(0)|Z_i(0)=z, \tau_\epsilon = co)\Big)\\
& = \frac{\mathbb{E}_{Z}\left(\mathbb{E}(Y_i|X_i = c^+,Z_i=z) - \mathbb{E}(Y_i|X_i = c^-,Z_i=z)\right)}{\mathbb{E}_{Z}\left(\mathbb{E}(T_i|X_i = c^+,Z_i=z) - \mathbb{E}(T_i|X_i = c^-,Z_i=z)\right)}.
\end{align*}
\end{thm}

Theorem \ref{thm4} proves the identification results for the three WATEs under the fuzzy RD design. The proof for the first two results are straightforward. The third result relies on the CIA in addition. This is because the sharp RD estimand is to ignore the conditioning on $X_i$, however, the fuzzy RD has to condition on $X_i$ to cancel with the denominator. The CIA is to directly remove the conditioning on $X_i$ and thus is necessary for the results to hold. The local linear estimators of these three estimands can be developed following the argument in Section \ref{sec_est}. However, the technical details and the theoretical results are much more involved, and are beyond the scope of this work. Given the practical importance of the fuzzy RD design, we think this is an important future research problem to explore.

\section{Discussion}\label{sec_discussion}
In this paper we study the RD problem when the conditional distribution of the covariates given running variables is discontinuous at the cutoff, for example, due to self-selection.  The standard RD design is no longer valid as the continuity of potential outcomes assumption is violated. We show that casual effect can still be recovered when the covariates related to self-selection are observed. We thus propose a set of estimands under the framework of WATE and show that these estimands can be estimated using a class of weighted local linear estimators. We derive the theory for our estimators include consistency and asymptotic normality. We further compare our estimator with the standard RD estimator in simulation exercises to demonstrate its finite sample performance. 

We apply our estimator to two empirical examples. First, we study the U.S. House elections data in \cite{Lee2008}, which is the classical data for RD estimator.  Our estimator and local polynomial RD estimator have similar performance. Second, we apply our method to evaluate the effect of a GSP auction. We show that the result from our estimator is different from standard RD estimator due to self-selection. In particular, the obtained effect of advertising by using our estimator is smaller and statistically insignificant when self-selection is taking into account.  


\newpage
\appendix

\section{Proofs}
\label{sec_proofs}





In the proof, we use $C$ to denote a generic constant which may change from line to line. Define the following notations:  
\[\kappa_{\iota} = \int_{u>0} K(u)u^\iota du \quad \text{  and  } \quad \kappa_{2\iota} = \int_{u>0} K(u)^2u^\iota du, \quad \text{ for  } \iota = 0, 1, 2, \cdots \] 
For the ease of presentation, we introduce the following notations. Let $f_{X|Z(1)}'(c^+|z_i)$, $f_{X|Z(1)}''(c^+|z_i)$ and $f_{X|Z(1)}'''(c^+|z_i)$ denote the right derivatives $\frac{\partial f_{X|Z(1)}(c^+|z_i)}{\partial x}$, $\frac{\partial^2 f_{X|Z(1)}(c^+|z_i)}{\partial x^2}$ and $\frac{\partial^3 f_{X|Z(1)}(c^+|z_i)}{\partial x^3}$. Denote 
$$m_1(x,z)=\EE(Y(1)| X=x, Z(1)=z),
$$
 \[\alpha_1 = \EE\{Y(1)w_1(Z(1))|X=c^+\} = \int   \frac{f_{X|Z(1)}(c^+|z)}{\pi_1(z)}  m_1(c^+, z) f_{Z(1)}(z)d{z},\]
 and
\[\beta_1=  \int   \frac{f'_{X|Z(1)}(c^+|z)}{\pi_1(z)}  m_1(c^+, z) f_{Z(1)}(z)d{z}+\int   \frac{f_{X|Z(1)}(c^+|z)}{\pi_1(z)}  m'_1(c^+, z) f_{Z(1)}(z)d{z},
\]
where $m'_1(c^+, z)=\frac{\partial m_1(c^+,z)}{\partial x}$. Define $R_{i} = y_i -\alpha_1-  \beta_1 (x_i-c)$,  
\begin{align*}
A_n &= \sum_{i=1}^n \frac{T_i}{\pi_1(z_{i})}  K\Big( \frac{x_i-c}{h}\Big)(x_i-c)^2,\\
B_n &= -\sum_{i=1}^n \frac{T_i}{\pi_1(z_{i})}  K\Big( \frac{x_i-c}{h}\Big)(x_i-c),\\
C_n &= \sum_{i=1}^n \frac{T_i}{\pi_1(z_{i})}  K\Big( \frac{x_i-c}{h}\Big),\\
D_n &= A_nC_n-B_n^2.
\end{align*}
Since $A_n, B_n, C_n, D_n$ all depend on $\pi_t(z)$, we denote the corresponding version with $\hat\pi_1(z)$ as $\hat A_n, \hat B_n, \hat C_n, \hat D_n$.


\begin{lemma}\label{lemABCD}
Under assumption \ref{asscont2}, \ref{ContinuityofConditionalDensity} and \ref{SharpRDdesign}, we have
\[\Big| \frac{\hat A_n}{n} - h^3\kappa_2- h^4\mathbb{E}\left(\frac{f_{X|Z(1)}'(c^+|z_i)}{\pi_1(z_i)}\right)\kappa_3\Big|   = O_p(h^5+rh^3+\frac{h^{5/2}}{n^{1/2}}),\]
\[\Big|\frac{\hat B_n}{n} + h^2\kappa_1 +h^3\mathbb{E}\left(\frac{f_{X|Z(1)}'(c^+|z_i)}{\pi_1(z_i)}\right)\kappa_2\Big|   = O(h^4+rh^2+\frac{h^{3/2}}{n^{1/2}}),\]
\begin{equation}\label{eqC}
\Big|\frac{\hat C_n}{n} - h/2 \Big| =O_p(h^2+rh+\frac{h^{1/2}}{n^{1/2}}),
\end{equation}
\[\Big|\frac{\hat D_n}{n^2} -  h^4 \left(\frac{1}{2}\kappa_2 - \kappa_1^2\right)\Big| = O_p(h^5+rh^4+\frac{h^{7/2}}{n^{1/2}}),\]
where $r$ satisfies $\sup_{z\in\cZ}|\hat\pi_1(z)-\pi_1(z)|=O_p(r)$. 
\end{lemma}

\begin{proof}
We will focus on the proof of (\ref{eqC}). The remain results can be shown following the similar steps. By triangle inequality,
\begin{align*}
&\Big|\frac{\hat C_n}{n} - \frac{h}{2} \Big|\\
&\leq \Big|\frac{1}{n}\sum_{i=1}^n \frac{T_i}{\pi_1(z_{i})}  K\Big( \frac{x_i-c}{h}\Big)- \frac{h}{2}\Big|+ \Big|\frac{1}{n}\sum_{i=1}^n \frac{T_i(\pi_1(z_i)-\hat\pi_1(z_i))}{\pi_1(z_{i})\hat\pi_1(z_i)}  K\Big( \frac{x_i-c}{h}\Big)\Big|:=I_1+I_2.
\end{align*}
For term $I_1$,  we further decompose it into two terms,
\begin{equation}\label{eqlemABCD1}
I_1\leq \Big|\frac{1}{n}\sum_{i=1}^n \frac{T_i}{\pi_1(z_{i})}  K\Big( \frac{x_i-c}{h}\Big)- \EE\Big(\frac{T_i}{\pi_1(z_{i})}  K\Big( \frac{x_i-c}{h}\Big)\Big)\Big|+\Big| \EE\Big(\frac{T_i}{\pi_1(z_{i})}  K\Big( \frac{x_i-c}{h}\Big)\Big)- h/2\Big|:=I_{11}+I_{12}.
\end{equation}
The bias term $I_{12}$ is computed as
\begin{align*}
I_{12}&=\Big|\int\int_{x>c}\frac{1}{\pi_1(z)}K\Big( \frac{x_i-c}{h}\Big) f_{X,Z(1)}(x,z)dxdz-\frac{h}{2}\Big|\\
& =h\Big|\int\int_{u>0}\frac{1}{\pi_1(z)}K(u) f_{X,Z(1)}(uh+c,z)dudz-\frac{1}{2}\Big|\\
&\leq h\Big|\int\int_{u>0}\frac{1}{\pi_1(z)}K(u) f_{X,Z(1)}(c^+,z)dudz-\frac{1}{2}\Big|\\
&~~~~+h\Big|\int\int_{u>0}\frac{1}{\pi_1(z)}K(u) f_{X|Z(1)}'(\tilde u|z) uh f_{Z(1)}(z)dudz\Big|,
\end{align*}
where the last step follows from the mean value theorem for some intermediate value $\tilde u$. Our assumption implies that $|f_{X|Z(1)}'(\tilde u|z)|$ is bounded. In addition, by assumption $\pi_1(z)$ is bounded away from 0 by a constant,  thus the second term is of order $h^2$. For the first term, by the choice of $\pi_1(z)$ we get 
\begin{equation}\label{weightrestriction}
\int\frac{1}{\pi_1(z)}f_{X|Z(1)}(c^+|z)f_{Z(1)}(z)dz=1,
\end{equation}
which implies
$$
\int\int_{u>0}\frac{1}{\pi_1(z)}K(u) f_{X|Z(1)}(c^+|z)f_{Z(1)}(z)dudz=\int\frac{1}{\pi_1(z)}f_{X|Z(1)}(c^+|z)f_{Z(1)}(z)dz \int_{u>0}K(u) du=\frac{1}{2}.
$$
Thus, we have $I_{12}=O(h^2)$. Now we consider $I_{11}$. By the Markov inequality, $I_{11}\lesssim (\EE I_{11}^2)^{1/2}$. Thus, it suffices to compute $\EE I_{11}^2$, 
\begin{align*}
\EE I_{11}^2&=\frac{1}{n} \EE\Big(\frac{T_i}{\pi_1^2(z_{i})}  K^2\Big( \frac{x_i-c}{h}\Big)\Big)-\frac{1}{n}\Big[\EE\Big(\frac{T_i}{\pi_1(z_{i})}  K\Big( \frac{x_i-c}{h}\Big)\Big)\Big]^2\\
&\leq \frac{1}{n}\int\int_{x>c}\frac{1}{\pi_1^2(z)}K^2\Big( \frac{x-c}{h}\Big) f_{X,Z(1)}(x,z)dxdz\\
&=\frac{h}{n}\int\int_{u>0}\frac{1}{\pi_1^2(z)}K^2(u) f_{X,Z(1)}(uh+c,z)dudz\\
&=O(\frac{h}{n}), 
\end{align*}
where the last step follows from the same argument in $I_{21}$. Putting them together into (\ref{eqlemABCD1}), we have $I_1=O_p(h^2+\frac{h^{1/2}}{n^{1/2}})$. For the last term $I_2$, we have
\begin{align*}
I_2&\leq \sup_{z\in\cZ}|\hat\pi_1(z)-\pi_1(z)| \frac{1}{n}\sum_{i=1}^n \frac{T_i}{\pi_1(z_{i})\hat\pi_1(z_i)}  K\Big( \frac{x_i-c}{h}\Big)\\
&\leq O_p(r)  \frac{1}{n}\sum_{i=1}^n \frac{T_i}{\pi_1(z_{i})}  K\Big( \frac{x_i-c}{h}\Big)=O_p(rh),
\end{align*}
where the last step holds by the bound for the $I_1$ term. This implies (\ref{eqC}). 
\end{proof}

\begin{lemma}\label{lemkernel}
Under assumption \ref{asscont2}, \ref{ContinuityofConditionalDensity} and \ref{SharpRDdesign}, we have
\begin{align*}
&\mathbb{E}\left( (h_2)^{-1}K(\frac{z_i-z_j}{h_2}) \Big|z_i\right) = f_Z(z_i) + O_p(h_2^2),\\[10pt]
&\mathbb{E}\left(  2(h_1^2)^{-1}T_j K_1(\frac{c-x_j}{h_1}, \frac{z_i-z_j}{h_1}) \Big| z_i,x_i \right) = f_{X, Z(1)}(c^+, z_i) + 2h_1\kappa_1 \frac{\partial f_{X, Z(1)}(c^+, z_i)}{\partial x_i} + O_p(h_1^2),\\[10pt]
&\sup_{z_i}\Big|(nh_2)^{-1}\sum_{j=1}^nK(\frac{z_i-z_j}{h_2}) - f_Z(z_i)\Big|=O_p\left(\sqrt{\frac{\log n}{nh_2}} + h_2^2\right),\\[10pt]
&\sup_{z_i}\Big|2(nh_1^2)^{-1}\sum_{j=1}^nT_jK(\frac{c-x_j}{h_1}, \frac{z_i-z_j}{h_1}) - f_{X, Z(1)}(c^+, z_i)\Big|=O_p\left(\sqrt{\frac{\log n}{nh_1^2}} + h_1\right).\\[-2pt]
\end{align*}

In case when $Z = Z(1) = Z(0)$,
\begin{align*}
&\mathbb{E}\left(  (h_1^2)^{-1}T_j K_1(\frac{c-x_j}{h_1}, \frac{z_i-z_j}{h_1}) \Big| z_i,x_i \right) = f_{X, Z}(c, z_i) + O_p(h_1^2),\\[10pt]
&\sup_{z_i}\Big|(nh_1^2)^{-1}\sum_{j=1}^nK(\frac{c-x_j}{h_1}, \frac{z_i-z_j}{h_1}) - f_{X, Z}(c, z_i)\Big|=O_p\left(\sqrt{\frac{\log n}{nh_1^2}} + h_1^2\right).
\end{align*}
\end{lemma}

\begin{proof}
First, note that
$$
\PP(Z=z)=\PP(Z=z,T=1)+\PP(Z=z,T=0)=\PP(Z(1)=z,X>c)+\PP(Z(0)=z,X<c).
$$
Then, we have 
$$
f_Z(z)=\int_{x>c} f_{X,Z(1)} (x,z)dx+\int_{x<c} f_{X,Z(0)} (x,z)dx,
$$
which implies that $f_Z(z)$ is second order continuously differentiable by the continuously differentiable property of $f_{X,Z(1)} (x,z)$ and $f_{X,Z(0)} (x,z)$ in assumption \ref{ContinuityofConditionalDensity}. Thus, the standard calculation in nonparametric density estimation yields
\begin{equation*}
\begin{split}
\mathbb{E}\left( (h_2)^{-1}K(\frac{z_i-z_j}{h_2}) \Big|z_i\right) &= \frac{1}{h_2}\int  K\left(\frac{z_i-z_j}{h_2}\right) f_Z(z_j)dz_j = \int  K\left(u\right) f_Z(z_i + uh_2)du\\
& = f_Z(z_i) +O_p(h_2^2).
\end{split}
\end{equation*}
To show the second result, following the similar argument, we get
\begin{equation*}
\begin{split}
&\mathbb{E}\left(  2(h_{11}h_{12})^{-1}T_j K_1(\frac{c-x_j}{h_{11}}, \frac{z_i-z_j}{h_{12}}) \Big| z_i,x_i \right) = \frac{2}{h_{11}h_{12}}\int \int T_j K_1(\frac{c-x_j}{h_{11}}, \frac{z_i-z_j}{h_{12}}) f_{X, Z(1)}(x_j, z_j)dz_jdx_j\\
& = \frac{2}{h_{11}}\int T_j K(\frac{c-x_j}{h_{11}}) \int K(v) f_{X, Z(1)}(x_j, z_i+vh_{12})dvdx_j\\
& = \frac{2}{h_{11}}\int T_j K(\frac{c-x_j}{h_{11}}) \left(f_{X, Z(1)}(x_j, z_i) + Ch_{12}^2\right) dx_j\\
& = 2\int_{u>0} K(u) \left(f_{X, Z(1)}(c+uh_{11}, z_i) + Ch_{12}^2\right) du\\
& = f_{X, Z(1)}(c^+, z_i) + 2h_{11}\kappa_1 \frac{\partial f_{X, Z(1)}(c^+, z_i)}{\partial x_i} + O(h_{12}^2+h_{11}^2),
\end{split}
\end{equation*} 
where $C$ is a generic constant. The last two results follow from \cite{Fan1993} together with the above bias calculation. 
\end{proof}

\subsection{Proof of Theorem \ref{thm2}}
Since the local linear estimator is invariant to the scale of $\pi_1(z)$, we can simply take $\frac{1}{\pi_1(z)} = \frac{f_Z(z)}{f_{X, Z(1)}(c, z)}$ in the rest of the proof. It can be estimated by the following kernel estimator: 
\begin{equation*}
\begin{split}
\frac{1}{\hat\pi_1(z)}=\frac{\hat{f}_Z(z)}{\hat{f}_{X,Z(1)}(c,z)} = \frac{(nh_2)^{-1}\sum_{i=1}^nK(\frac{z-z_i}{h_2})}{2 \cdot (nh_1^2)^{-1}\sum_{x_i>c}K_1(\frac{c-x_i}{h_1}, \frac{z-z_i}{h_1})}.
\end{split}
\end{equation*}
Start with the following minimization problem: 
\begin{equation}
\begin{split}
\Big(\hat{\alpha}_1, \hat{\beta}_1\Big) &= \arg \min_{\alpha, \beta} \sum_{i} \frac{T_i}{\hat \pi_1{(z_i)}} \Big( y_i- \alpha - (x_i-c)\beta\Big)^2 K\Big( \frac{x_i-c}{h}\Big). \label{eqopt}
\end{split}
\end{equation}
Recall that for any kernel estimates $\hat f=\hat f(x)$ and $\hat g=\hat g(x)$, if $f$ is bounded away from 0, then 
\begin{equation}\label{kerexp}
\begin{split}
\frac{1}{\hat f} = \frac{1}{f} - \frac{1}{f^2}(\hat f - f)+ O_p(r^2),
\end{split}
\end{equation}
where $\|\hat f - f\|_\infty=O_p(r)$ and $\|\hat g - g\|_\infty=O_p(s)$. Thus, if $g$ is bounded from above, then
\begin{equation*}
\begin{split}
\frac{\hat g}{\hat f} &= \frac{g + (\hat g - g)}{f} - \frac{g + (\hat g - g)}{f^2}(\hat f - f)+O_p(r^2)\\
& = \frac{g}{f} - \frac{g}{f^2}(\hat f - f) + \frac{\hat g - g}{f} +O_p(r^2+rs )\\
& = \frac{g}{f} - \frac{g\hat f}{f^2} + \frac{\hat g}{f}+ O_p(r^2+rs).
\end{split}
\end{equation*}
Following the above discussion, we can show that 
\begin{equation*}
\begin{split}
&\frac{1}{2\cdot(nh_1^2)^{-1}\sum_{j=1}^nT_jK_1(\frac{c-x_j}{h_1}, \frac{z_i-z_j}{h_1})} \\
&=  - \frac{1}{f_{X, Z(1)}(c^+, z_i)^2}\left(2(nh_1^2)^{-1}\sum_{j=1}^nT_jK_1(\frac{c-x_j}{h_1}, \frac{z_i-z_j}{h_1}) - f_{X, Z(1)}(c^+, z_i)\right)\\
& + \frac{1}{f_{X, Z(1)}(c^+, z_i)}+  O_p\left(r^2\right),
\end{split}
\end{equation*}
where Lemma \ref{lemkernel} implies
$$
r=\sqrt{\frac{\log n}{nh_1^2}} + h_1.
$$
The gradient of (\ref{eqopt}) can be written as the following $U$-statistic:
\begin{equation*}
\begin{split}
&\frac{1}{n}\sum_{i=1}^n \frac{T_i \hat{f}_Z(z_i)}{\hat{f}_{X,Z(1)}(c,z_i)}  K\Big( \frac{x_i-c}{h}\Big) R_i = \frac{1}{n}\sum_{i=1}^n\frac{T_i(nh_2)^{-1}\sum_{j=1}^nK(\frac{z_i-z_j}{h_2})}{2 \cdot(nh_1^2)^{-1}\sum_{j=1}^nT_jK_1(\frac{c-x_j}{h_1}, \frac{z_i-z_j}{h_1})}K\Big( \frac{x_i-c}{h}\Big)R_i\\
& = \frac{1}{n}\sum_{i=1}^n \frac{T_if_Z(z_i)}{f_{X, Z(1)}(c^+, z_i)}K\Big( \frac{x_i-c}{h}\Big)R_i + \frac{T_i}{f_{X, Z(1)}(c^+, z_i)}\left((nh_2)^{-1}\sum_{j=1}^nK(\frac{z_i-z_j}{h_2})  \right)K\Big( \frac{x_i-c}{h}\Big)R_i \\
&- \frac{T_if_Z(z_i)}{f_{X, Z(1)}(c^+, z_i)^2}\left(2\cdot (nh_1^2)^{-1}\sum_{j=1}^nT_jK_1(\frac{c-x_j}{h_1}, \frac{z_i-z_j}{h_1}) \right)K\Big( \frac{x_i-c}{h}\Big)R_i + O_p\left(r^2 +rs\right)\\
& = \frac{1}{n^2}\sum_{i=1}^n\sum_{j=1}^n \Big\{\frac{T_i}{f_{X, Z(1)}(c^+, z_i)} (h_2)^{-1}K(\frac{z_i-z_j}{h_2}) K\Big( \frac{x_i-c}{h}\Big)R_i\\
& -2 \cdot\frac{T_iT_jf_Z(z_i)}{f_{X, Z(1)}(c^+, z_i)^2}(h_1^2)^{-1}K_1(\frac{c-x_j}{h_1}, \frac{z_i-z_j}{h_1})K\Big( \frac{x_i-c}{h}\Big)R_i + \frac{T_if_Z(z_i)}{f_{X, Z(1)}(c^+, z_i)}K\Big( \frac{x_i-c}{h}\Big)R_i\Big\}+ O_p\left(r^2+rs\right),
\end{split}
\end{equation*}
where 
$$
s=\sqrt{\frac{\log n}{nh_2}} + h_2^2,
$$
implied by Lemma \ref{lemkernel}. Define 
\begin{equation*}
\begin{split}
\phi_{i,j}  &= \frac{T_i}{f_{X, Z(1)}(c^+, z_i)} (h_2)^{-1}K(\frac{z_i-z_j}{h_2}) K\Big( \frac{x_i-c}{h}\Big)R_i\\
 &-2 \cdot \frac{T_iT_jf_Z(z_i)}{f_{X, Z(1)}(c^+, z_i)^2}(h_1^2)^{-1}K_1(\frac{c-x_j}{h_1}, \frac{z_i-z_j}{h_1})K\Big( \frac{x_i-c}{h}\Big)R_i + \frac{T_if_Z(z_i)}{f_{X, Z(1)}(c^+, z_i)}K\Big( \frac{x_i-c}{h}\Big)R_i.
\end{split}
\end{equation*}
Thus,
\begin{equation}
\begin{split}
\frac{1}{n}\sum_{i=1}^n \frac{T_i \hat{f}_Z(z)}{\hat{f}_{X|Z(1)}(c,z)}  K\Big( \frac{x_i-c}{h}\Big) R_i &= \frac{1}{n^2}\sum_{i}\sum_{j}\phi_{i,j}+O_p\left(r^2+rs\right) \\
&= \frac{1}{2n^2}\sum_{i}\sum_{j}(\phi_{i,j}+\phi_{j,i}) +O_p\left(r^2+rs\right)\\
&= \frac{1}{n^2}\sum_{i<j}(\phi_{i,j}+\phi_{j,i})+\frac{1}{n^2}\sum_{i}\phi_{i,i}+ O_p\left(r^2+rs\right),\label{eqgradient}
\end{split}
\end{equation}
where the first (leading) term is a $U$-statistic after rescaling. By lemma \ref{lemmahoff} and Theorem 12.3 in \cite{van2000asymptotic}, we have
\begin{equation}\label{equstat}
\frac{n^{1/2}}{h^{1/2}}\Big(\frac{1}{n(n-1)}\sum_{i<j}(\phi_{i,j}+\phi_{j,i})-\delta/2-\frac{1}{n}\sum_{i=1}^n \{\EE(\phi_{i,j}+\phi_{j,i}|i)-\delta\}\Big)=o_p(1),
\end{equation}
where $\delta=\EE(\phi_{i,j}+\phi_{j,i})$, and we use $\EE(\cdot |i)$ to denote the conditional expectation given the $i$th sample. In the following, we approximate $\EE(\phi_{i,j}+\phi_{j,i}|i)$. Define $d(x_i, z_i) = m(x_i, z_i) - \alpha_1$. By Lemma \ref{phiij}, we have
\[\mathbb{E}(\phi_{i,j}) =  O_p(h^3).\]
The central limit theorem implies
\begin{equation}\label{eqclt}
\frac{1}{(nh)^{1/2}}\sum_{i=1}^n \{\EE(\phi_{i,j}+\phi_{j,i}|i)-\delta\}+O_p(\frac{n^{1/2}}{h^{1/2}}(h^3))\rightarrow_d N(0,\xi^2/h),
\end{equation}
where $\delta =  O_p(h^3)$ and $\xi^2=\EE\{(\EE(\phi_{i,j}+\phi_{j,i}|i)-\delta)^2\}=\mathbb{E}\left\{\mathbb{E}(\phi_{ij}|i)^2 + \mathbb{E}(\phi_{ji}|i)^2 + 2\mathbb{E}(\phi_{ij}|i)\mathbb{E}(\phi_{ji}|i) \right\}-\delta^2$. We now calculate the asymptotic variance as follows. Since $\mathbb{E}(\phi_{j,i}|i)=\mathbb{E}(\phi_{i,j}|j)$, from lemma \ref{phiij}, we have
\begin{equation*}
\begin{split}
\mathbb{E}(\phi_{ij}|j)  &=  \frac{h}{2} d_1(c^+, z_j) -T_j\frac{h}{h_1} K(\frac{c-x_j}{h_1})  \frac{f_Z(z_j)}{f_{X, Z(1)}(c^+,z_j)}d_1(c^+, z_j) + O_p(h^2+hh_1+hh_2).
\end{split}
\end{equation*}
Similarly, we can show that 
\[\mathbb{E}(\phi_{i,j}|i) = \frac{T_if_Z(z_i)}{f_{X, Z(1)}(c^+, z_i)}K\Big( \frac{x_i-c}{h}\Big)R_i +O_p(h_2^2 + h_1).\]
Recall that  $\sigma^2=\EE(Y(1)-m_1(X, Z(1)))^2$. Since $h\asymp \sqrt{h_1}\asymp h_2$, after some tedious calculation we can show that 
\begin{equation*}
\begin{split}
\frac{1}{n} \mathbb{E}\left(\mathbb{E}(\phi_{ji}|i)^2 \right)&  = \frac{1}{n}\frac{h^2}{h_1^2} \int_{z_i}\int_{x_i}T_iK(\frac{c-x_i}{h_1})^2  \frac{f_Z(z_i)^2}{f_{X, Z(1)}(c^+,z_i)^2}d(c^+, z_i)^2f_{X, Z(1)}(x_i, z_i)dx_idz_i + O\left(\frac{h^2}{n}\right)\\
& = \frac{h^2}{h_1} \frac{\kappa_{20}}{n}\int_{z_i}\frac{f_Z(z_i)^2}{f_{X, Z(1)}(c^+, z_i)}d(c^+, z_i)^2 dz_i+O\left(\frac{h^2}{n}\right)\\
& = \frac{h^2}{h_1} \frac{\kappa_{20}}{n}\mathbb{E}_Z\left(\frac{f_Z(z_i)}{f_{X, Z(1)}(c^+, z_i)}d(c^+, z_i)^2 \right)+O\left(\frac{h^{2}}{n}\right),
\end{split}
\end{equation*}
And
$$
\xi^2 = \kappa_{20} \underbrace{\mathbb{E}_Z\left(\frac{f_Z(z_i)}{f_{X, Z(1)}(c^+, z_i)}d_1(c^+, z_i)^2 \right)}_{\omega}  + O(h).
$$
Combining (\ref{equstat}) and (\ref{eqclt}), 
$$
n^{1/2}\Big(\frac{1}{n(n-1)}\sum_{i<j}(\phi_{i,j}+\phi_{j,i})-\delta/2\Big)+O\Big(n^{1/2}h^3\Big)\rightarrow_d N(0,\xi^2).
$$
Finally, note that in (\ref{eqgradient}), 
$$
\frac{1}{n^2}\sum_{i=1}^n\phi_{i,i}\lesssim \frac{1}{n}\EE(\phi_{i,j}) = O_p\left( \frac{h^2}{n}\right),
$$
and therefore we obtain that
$$
\frac{1}{n^{1/2}}\sum_{i=1}^n \frac{T_i \hat{f}_Z(z)}{\hat{f}_{X, Z(1)}(c,z)}  K\Big( \frac{x_i-c}{h}\Big)R_i+\phi\rightarrow_d N(n^{1/2}\delta/2,\xi^2),
$$
where  
$$
\phi=O_p\Big(n^{1/2}(h^3+r^2+rs)\Big).
$$
Following the similar argument, we can show the joint convergence
$$
\frac{1}{n^{1/2}}\sum_{i=1}^n \frac{T_i \hat{f}_Z(z)}{\hat{f}_{X, Z(1)}(c,z)}  K\Big( \frac{x_i-c}{h}\Big) R_i [1, (x_i-c)]^T\rightarrow_d N\left(n^{1/2} \left(\begin{matrix} O_p(h^3) \\ O_p(h^4)
\end{matrix}\right), \omega\left(\begin{matrix}
 \kappa_{20} & h\kappa_{21} \\
  h\kappa_{21} & h^2\kappa_{22} \\
 \end{matrix}\right)
\right).
$$
By the least squared formulation, the estimator $\hat\alpha_1$ satisfies
$$
\sqrt{nh^2}(\hat{\alpha}_1 - \alpha_1)=-e_1^T\left(\begin{matrix}
\hat C_n/(nh) & -\hat B_n/(nh)  \\ -\hat B_n/(nh)  & \hat A_n/(nh) 
\end{matrix}\right)^{-1} \frac{1}{n^{1/2}}\sum_{i=1}^n \frac{T_i \hat{f}_Z(z)}{\hat{f}_{X, Z(1)}(c,z)}  K\Big( \frac{x_i-c}{h}\Big) R_i [1, (x_i-c)]^T,
$$
where $e_1^T=(1,0)$.
From lemma \ref{lemABCD} and the matrix inversion formula, 
$$
\left(\begin{matrix}
\hat C_n/(nh) & -\hat B_n/(nh)  \\ -\hat B_n/(nh)  & \hat A_n/(nh) 
\end{matrix}\right)^{-1}=\frac{1}{\hat D_n/(nh)^2} \left(\begin{matrix}
\hat A_n/(nh) & \hat B_n/(nh) \\ \hat B_n/(nh) & \hat C_n/(nh)
\end{matrix}\right)\rightarrow_p \frac{1}{h^2(\kappa_2/2 -\kappa_1^2)}\left(\begin{matrix}
h^2 \kappa_2 & -h \kappa_1 \\ -h \kappa_1 & \frac{1}{2}
\end{matrix}\right).
$$
Thus, the asymptotic bias of $\sqrt{nh^2}(\hat{\alpha}_1 - \alpha_1)$ is
$$
\frac{-e_1^T}{h^2(\kappa_2/2 -\kappa_1^2)}\left(\begin{matrix}
h^2 \kappa_2 & -h \kappa_1 \\ -h \kappa_1 & \frac{1}{2}
\end{matrix}\right)n^{1/2} \left(\begin{matrix}
O_p(h^3) \\ O_p(h^4)
\end{matrix}\right)=O(n^{1/2}h^{3})=o(1).
$$
Similarly, the asymptotic variance  of $\sqrt{nh^2}(\hat{\alpha}_1 - \alpha_1)$ is
\begin{align*}
&\frac{\omega }{h^4(\kappa_2/2 -\kappa_1^2)^2}e_1^T\left(\begin{matrix}
h^2 \kappa_2 & -h \kappa_1 \\ -h \kappa_1 & \frac{1}{2}
&\end{matrix}\right)\left(\begin{matrix}
 \kappa_{20} & h\kappa_{21} \\
  h\kappa_{21} & h^2\kappa_{22} \\
 \end{matrix}\right)
\left(\begin{matrix}
h^2 \kappa_2 & -h \kappa_1 \\ -h \kappa_1 & \frac{1}{2}
\end{matrix}\right)e_1\\
&=\mathbb{E}_Z\left(\frac{f_Z(z_i)}{f_{X, Z(1)}(c^+, z_i)}d(c^+, z_i)^2 \right)\cdot C_v,
\end{align*}
where 
\[C_v = \frac{\kappa_2^2\kappa_{20}  + \kappa_1^2\kappa_{22} - 2\kappa_1\kappa_2\kappa_{21}}{\left(\frac{1}{2}\kappa_2-\kappa_1^2\right)^2}.\]
This completes the proof.

\begin{lemma} \label{lemM}
Under the same condition as in Theorem \ref{thm2}, 
\begin{equation*}
\begin{split}
M(c, z_i) &:= \frac{1}{h}\int T_iK\left(\frac{x_i-c}{h}\right)\EE(R_i|x_i,z_i) f_{X, Z(1)}(x_i, z_i) dx_i \\
& = \frac{1}{2} m_1(c^+, z_i) f_{X, Z(1)}(c^+, z_i) + h \kappa_1 \frac{\partial m_1(c^+, z_i)}{\partial x_i} f_{X, Z(1)}(c^+, z_i) + h \kappa_1 m_1(c^+, z_i) \frac{\partial f_{X, Z(1)}(c^+, z_i)}{\partial x_i}\\
& - \frac{1}{2}\alpha_1f_{X, Z(1)}(c^+, z_i) -  h \kappa_1\beta_1f_{X, Z(1)}(c^+, z_i) - h\kappa_1 \alpha_1\frac{\partial f_{X, Z(1)}(c^+, z_i)}{\partial x_i} + O_p(h^2),
\end{split}
\end{equation*}
where $O_p$ terms are valid uniformly over $i$.
\end{lemma}

\begin{proof}
Following the standard Taylor expansion, we can show that
\begin{align*}
& M(c, z_i) = \frac{1}{h}\int T_iK\left(\frac{x_i-c}{h}\right)\EE(R_i|x_i,z_i) f_{X, Z(1)}(x_i, z_i) dx_i \\
& = \int_{u>0} K(u)\left(m_1(c+uh, z_i) - \alpha_1-  uh\beta_1  \right) f_{X, Z(1)}(c+uh, z_i) du\\
& = \int_{u>0} K(u)\left(m_1(c^+, z_i) + \frac{\partial m_1(c^+, z_i)}{\partial x_i}uh +\frac{\partial^2 m_1(c^+, z_i)}{\partial x_i^2}\frac{u^2h^2}{2} +\frac{\partial^3 m_1(\tilde c^{(2)}, z_i)}{\partial x_i^3}\frac{u^3h^3}{3!} - \alpha_1-  uh\beta_1  \right) \\
&\left(f_{X, Z(1)}(c^+, z_i) + \frac{\partial f_{X, Z(1)}(c^+, z_i)}{\partial x_i}uh + \frac{\partial^2 f_{X, Z(1)}(c^+, z_i)}{\partial x_i^2}\frac{u^2h^2}{2}+ \frac{\partial^3 f_{X, Z(1)}(\tilde c^{(3)}, z_i)}{\partial x_i^3}\frac{u^3h^3}{3!} \right) du \\
& = \frac{1}{2} (m_1(c^+, z_i) - \alpha_1) f_{X, Z(1)}(c^+, z_i) + h \kappa_1\left(\frac{\partial m_1(c^+, z_i)}{\partial x_i} - \beta_1\right)f_{X, Z(1)}(c^+, z_i)\\
&  + h \kappa_1(m_1(c^+, z_i) - \alpha_1)\frac{\partial f_{X, Z(1)}(c^+, z_i)}{\partial x_i} +O_p(h^2),
\end{align*}
where $O_p$ terms are valid uniformly over $i$ as the (mixed) third derivatives of $f_{X, Z(1)}(x_i, z_i)$ are all bounded. 
\end{proof}

\begin{lemma}\label{phiij}
Recall that 
\begin{equation*}
\begin{split}
\phi_{i,j}  &= \frac{T_i}{f_{X, Z(1)}(c^+, z_i)} (h_2)^{-1}K(\frac{z_i-z_j}{h_2}) K\Big( \frac{x_i-c}{h}\Big)R_i\\
 &-2 \cdot \frac{T_iT_jf_Z(z_i)}{f_{X, Z(1)}(c^+, z_i)^2}(h_1^2)^{-1}K_1(\frac{c-x_j}{h_1}, \frac{z_i-z_j}{h_1})K\Big( \frac{x_i-c}{h}\Big)R_i + \frac{T_if_Z(z_i)}{f_{X, Z(1)}(c^+, z_i)}K\Big( \frac{x_i-c}{h}\Big)R_i.
\end{split}
\end{equation*}
Under the same condition as in Theorem \ref{thm2}, and when $h_1 = h^2$
\begin{align*}
&\mathbb{E}(\phi_{i,j}|j) = \frac{h}{2} d(c^+, z_j) -T_j\frac{h}{h_1} K(\frac{c-x_j}{h_1})  \frac{f_Z(z_j)}{f_{X, Z(1)}(c^+,z_j)}d_1(c^+, z_j) + O_p(h^2+hh_1+hh_2),\\
&\mathbb{E}(\phi_{i,j}|i) = \frac{T_if_Z(z_i)}{f_{X, Z(1)}(c^+, z_i)}K\Big( \frac{x_i-c}{h}\Big)R_i +O_p(h_2^2 + h_1), \\
&\mathbb{E}(\phi_{i,j}) =  O_p(h^3) 
\end{align*}
where $O_p$ terms are valid uniformly over $i$ or $j$.
\end{lemma}

\begin{proof}
\begin{equation*}
\begin{split}
\mathbb{E}(\phi_{i,j}|j)  &=\underbrace{\mathbb{E}\left( \frac{T_i}{f_{X, Z(1)}(c^+, z_i)}  (h_2)^{-1}K(\frac{z_i-z_j}{h_2})  K\Big( \frac{x_i-c}{h}\Big)R_i\Big|z_j\right)}_{\text{(Part.BI)}}\\
 &-2T_j \cdot\underbrace{\mathbb{E}\left(  \frac{T_if_Z(z_i)}{f_{X, Z(1)}(c^+, z_i)^2} (h_1^2)^{-1}K_1(\frac{c-x_j}{h_1}, \frac{z_i-z_j}{h_1}) K\Big( \frac{x_i-c}{h}\Big)R_i\Big|z_j, x_j \right)}_{\text{(Part.BII)}}\\
 & + \underbrace{\mathbb{E}\left(\frac{T_if_Z(z_i)}{f_{X, Z(1)}(c^+, z_i)}K\Big( \frac{x_i-c}{h}\Big)R_i\Big| z_j,x_j \right)}_{\text{(Part.BIII)}}.
\end{split}
\end{equation*}
From lemma \ref{lemM},
\begin{equation*}
\begin{split}
\text{(Part.BI)} & = \mathbb{E}\left( \frac{T_i}{f_{X, Z(1)}(c^+, z_i)}  (h_2)^{-1}K(\frac{z_i-z_j}{h_2})  K\Big( \frac{x_i-c}{h}\Big)R_i\Big|z_j\right)\\
& = h\int_{z_i} \frac{1}{f_{X, Z(1)}(c^+, z_i)}  (h_2)^{-1}K(\frac{z_i-z_j}{h_2}) M(c, z_i) dz_i\\
& = h\int_{z_i}   (h_2)^{-1}K(\frac{z_i-z_j}{h_2})\frac{1}{2} (m_1(c^+, z_i) - \alpha_1) dz_i + O_p(h^2)\\
& = \frac{h}{2} (m_1(c^+, z_j) - \alpha_1) + O_p(h^2 +h_2h),
\end{split}
\end{equation*}
where $O_p$ terms are valid uniformly over $j$.  Similarly, for Part.BII and Part.BIII, we can show that
\begin{equation*}
\begin{split}
&\text{(Part.BII)} = \mathbb{E}\left(  \frac{T_if_Z(z_i)}{f_{X, Z(1)}(c^+, z_i)^2} (h_1^2)^{-1}K_1(\frac{c-x_j}{h_1}, \frac{z_i-z_j}{h_1}) K\Big( \frac{x_i-c}{h}\Big)R_i\Big|z_j, x_j \right) \\
& = h\int_{z_i} \frac{f_Z(z_i)}{f_{X, Z(1)}(c^+, z_i)^2}  (h_1^2)^{-1}K_1(\frac{c-x_j}{h_1}, \frac{z_i-z_j}{h_1}) M(c, z_i)dz_i\\
& = \frac{h}{2}\int_{z_i}  \frac{f_Z(z_i)}{f_{X, Z(1)}(c^+, z_i)} (h_1^2)^{-1}K_1(\frac{c-x_j}{h_1}, \frac{z_i-z_j}{h_1})  (m_1(c^+, z_i) - \alpha_1) dz_i + O_p(h^2)\\
& = \frac{h}{2h_1} K(\frac{c-x_j}{h_1})  \frac{f_Z(z_j)}{f_{X, Z(1)}(c^+,z_j)} (m_1(c^+, z_j) - \alpha_1)+ O_p\left(hh_1\right)  + O_p(h^2),
\end{split}
\end{equation*}
\begin{equation*}
\begin{split}
\text{(Part.BIII)} & = \mathbb{E}\left(\frac{T_if_Z(z_i)}{f_{X, Z(1)}(c^+, z_i)}K\Big( \frac{x_i-c}{h}\Big)R_i\Big| z_j,x_j \right)
 = \frac{h}{2}\int_{z_i}  (m_1(c^+, z_i) - \alpha_1)f_Z(z_i) dz_i + O_p(h^2) = O_p(h^2),
\end{split}
\end{equation*}
where the last step follows from the definition of $\alpha_1$. Define $d_1(x_i, z_i) = m_1(x_i, z_i) - \alpha_1$. Combining the Part BI, BII and BIII, we obtain 
\begin{equation*}
\begin{split}
\mathbb{E}(\phi_{ij}|j)  &=  \frac{h}{2} d_1(c^+, z_j) -T_j\frac{h}{h_1} K(\frac{c-x_j}{h_1})  \frac{f_Z(z_j)}{f_{X, Z(1)}(c^+,z_j)}d_1(c^+, z_j) + O_p(h^2+hh_1+hh_2).
\end{split}
\end{equation*}
Following the similar calculation, by lemma \ref{lemkernel} we have 
\begin{equation*}
\begin{split}
\mathbb{E}(\phi_{i,j}|i)  &= \underbrace{\frac{T_i}{f_{X, Z(1)}(c^+, z_i)} \mathbb{E}\left( (h_2)^{-1}K(\frac{z_i-z_j}{h_2}) \Big|z_i\right) K\Big( \frac{x_i-c}{h}\Big)R_i}_{\text{Part. AI}}\\
 &-\underbrace{2 \cdot \frac{T_if_Z(z_i)}{f_{X, Z(1)}(c^+, z_i)^2}\mathbb{E}\left(  (h_1^2)^{-1}T_jK_1(\frac{c-x_j}{h_1}, \frac{z_i-z_j}{h_1}) \Big| z_i,x_i \right)K\Big( \frac{x_i-c}{h}\Big)R_i}_{\text{Part. AII}} \\
 &+ \underbrace{\frac{T_if_Z(z_i)}{f_{X, Z(1)}(c^+, z_i)}K\Big( \frac{x_i-c}{h}\Big)R_i}_{\text{Part. AIII}}\\
&= \frac{T_if_Z(z_i)}{f_{X, Z(1)}(c^+, z_i)}K\Big( \frac{x_i-c}{h}\Big)R_i +O_p(h_1 + h_2^2).
\end{split}
\end{equation*}
Finally, we calculate $\mathbb{E}(\phi_{i,j})$ using Lemma \ref{lemM}, 
\begin{equation*}
\begin{split}
\mathbb{E}(\text{ Part.AIII} ) & = h\int   \frac{ f_Z(z_i)}{f_{X, Z(1)}(c^+, z_i)} M(c, z_i)dz_i  \\
& = \frac{h}{2}\int  (m_1(c^+, z_i) - \alpha_1) f_Z(z_i)dz_i- h^2\kappa_1\beta_1- h^2\kappa_1 \alpha_1 \int \frac{\partial f_{X, Z(1)}(c^+, z_i)}{\partial x_i} \frac{f_Z(z_i)}{f_{X, Z(1)}(c^+, z_i)} dz_i\\
& + h^2 \kappa_1\int\left(\frac{\partial m_1(c^+, z_i)}{\partial x_i}f_Z(z_i)+m_1(c^+, z_i) \frac{\partial f_{X, Z(1)}(c^+, z_i)}{\partial x_i}\frac{f_Z(z_i)}{f_{X, Z(1)}(c^+, z_i)} \right)dz_i   +O_p( h^3)\\
& = -h^2\kappa_1 \alpha_1 \int \frac{\partial f_{X, Z(1)}(c^+, z_i)}{\partial x_i} \frac{f_Z(z_i)}{f_{X, Z(1)}(c^+, z_i)} dz_i +O_p(h^3)
\end{split}
\end{equation*}
where the $O_p$ terms are valid uniformly over $i$ and the last equality follows as 
 \[\alpha_1 = \int   m_1(c^+, z) f_{Z}(z)d{z},\]
and
\[\beta_1 = \int  \frac{\partial m_1(c^+, z_i)}{\partial x_i} f_Z(z_i) dz_i + \int m_1(c^+, z_i) \frac{\partial f_{X, Z(1)}(c^+, z_i)}{\partial x_i} \frac{f_Z(z_i)}{f_{X, Z(1)}(c^+, z_i)} dz_i.\]  
\\
From Lemma \ref{lemkernel}, some tedious calculation implies
\begin{equation*}
\begin{split}
\mathbb{E}(\text{Part.AI} ) & =  \mathbb{E}(\text{Part.AIII} )+O_p(hh_2^2),
\end{split}
\end{equation*}
and 
\begin{equation*}
\begin{split}
\mathbb{E}(\text{Part.AII} ) & =  \mathbb{E}(\text{Part.AIII} )+2h_1\kappa_1\mathbb{E}\left(\frac{T_if_Z(z_i)}{f_{X, Z(1)}(c^+, z_i)^2}\frac{\partial f_{X, Z(1)}(c^+, z_i)}{\partial x_i}K\Big( \frac{x_i-c}{h}\Big)R_i\right) +O_p(hh_1^2).
\end{split}
\end{equation*}
From Lemma \ref{lemkernel} and Lemma \ref{lemM},
\begin{equation*}
\begin{split}
&h_1\kappa_1\mathbb{E}\left(\frac{T_if_Z(z_i)}{f_{X, Z(1)}(c^+, z_i)^2}\frac{\partial f_{X, Z(1)}(c^+, z_i)}{\partial x_i}K\Big( \frac{x_i-c}{h}\Big)R_i\right)\\
& = h_1\kappa_1\int \int \frac{T_if_Z(z_i)}{f_{X, Z(1)}(c^+, z_i)^2}\frac{\partial f_{X, Z(1)}(c^+, z_i)}{\partial x_i}  K\Big( \frac{x_i-c}{h}\Big)\EE(R_i|x_i,z_i) f_{X, Z(1)}(x_i, z_i)  dx_i dz_i\\
& = \frac{hh_1\kappa_1}{2}\int  \frac{f_Z(z_i)}{f_{X, Z(1)}(c^+, z_i)}\frac{\partial f_{X, Z(1)}(c^+, z_i)}{\partial x_i} (m_1(c^+, z_i) - \alpha_1) dz_i  +O_p(h_1h^2).
\end{split}
\end{equation*}
Thus when $h_1 \asymp h^2$
\[\mathbb{E}(\phi_{i,j}) =  O_p(h^3).\]
This completes the proof. 
\end{proof}

\subsection{Proof of Corollary \ref{corrd}}
In the case when $Z = Z(1) = Z(0)$, we can estimate $\frac{1}{\pi_1(z)} = \frac{f_Z(z)}{f_{X, Z}(c, z)}$ by the following kernel estimator: 
\begin{equation*}
\begin{split}
\frac{1}{\hat\pi_1(z)}=\frac{\hat{f}_Z(z)}{\hat{f}_{X,Z}(c,z)} = \frac{(nh_2)^{-1}\sum_{i=1}^nK(\frac{z-z_i}{h_2})}{(nh_1^2)^{-1}\sum_{i=1}^nK_1(\frac{c-x_i}{h_1}, \frac{z-z_i}{h_1})}.
\end{split}
\end{equation*}
Similar as in theorem \ref{thm2},
\begin{equation*}
\begin{split}
&\frac{1}{(nh_1^2)^{-1}\sum_{j=1}^nK_1(\frac{c-x_j}{h_1}, \frac{z_i-z_j}{h_1})} \\
&=  - \frac{1}{f_{X, Z}(c, z_i)^2}\left((nh_1^2)^{-1}\sum_{j=1}^nK_1(\frac{c-x_j}{h_1}, \frac{z_i-z_j}{h_1}) - f_{X, Z}(c, z_i)\right)\\
& + \frac{1}{f_{X, Z}(c, z_i)}+  O_p\left(r^2\right),
\end{split}
\end{equation*}
where Lemma \ref{lemkernel} implies
$$
r=\sqrt{\frac{\log n}{nh_1^2}} + h_1^2.
$$
The gradient in (\ref{eqopt}) can be written as the following $U$-statistic:
\begin{equation*}
\begin{split}
&\frac{1}{n}\sum_{i=1}^n \frac{T_i \hat{f}_Z(z_i)}{\hat{f}_{X,Z}(c,z_i)}  K\Big( \frac{x_i-c}{h}\Big) R_i = \frac{1}{n}\sum_{i=1}^n\frac{T_i(nh_2)^{-1}\sum_{j=1}^nK(\frac{z_i-z_j}{h_2})}{(nh_1^2)^{-1}\sum_{j=1}^nK_1(\frac{c-x_j}{h_1}, \frac{z_i-z_j}{h_1})}K\Big( \frac{x_i-c}{h}\Big)R_i\\
& = \frac{1}{n}\sum_{i=1}^n \frac{T_if_Z(z_i)}{f_{X, Z}(c, z_i)}K\Big( \frac{x_i-c}{h}\Big)R_i + \frac{T_i}{f_{X, Z}(c, z_i)}\left((nh_2)^{-1}\sum_{j=1}^nK(\frac{z_i-z_j}{h_2})  \right)K\Big( \frac{x_i-c}{h}\Big)R_i \\
&- \frac{T_if_Z(z_i)}{f_{X, Z}(c, z_i)^2}\left( (nh_1^2)^{-1}\sum_{j=1}^nK_1(\frac{c-x_j}{h_1}, \frac{z_i-z_j}{h_1}) \right)K\Big( \frac{x_i-c}{h}\Big)R_i + O_p\left(r^2 +rs\right)\\
& = \frac{1}{n^2}\sum_{i=1}^n\sum_{j=1}^n \Big\{\frac{T_i}{f_{X, Z}(c, z_i)} (h_2)^{-1}K(\frac{z_i-z_j}{h_2}) K\Big( \frac{x_i-c}{h}\Big)R_i\\
& -\frac{T_if_Z(z_i)}{f_{X, Z}(c, z_i)^2}(h_1^2)^{-1}K_1(\frac{c-x_j}{h_1}, \frac{z_i-z_j}{h_1})K\Big( \frac{x_i-c}{h}\Big)R_i + \frac{T_if_Z(z_i)}{f_{X, Z}(c, z_i)}K\Big( \frac{x_i-c}{h}\Big)R_i\Big\}+ O_p\left(r^2+rs\right),
\end{split}
\end{equation*}
where 
$$
s=\sqrt{\frac{\log n}{nh_2}} + h_2^2,
$$
Define
\begin{equation*}
\begin{split}
\phi_{i,j}^{rd}  &= \frac{T_i}{f_{X, Z}(c, z_i)} (h_2)^{-1}K(\frac{z_i-z_j}{h_2}) K\Big( \frac{x_i-c}{h}\Big)R_i\\
 &-  \frac{T_if_Z(z_i)}{f_{X, Z}(c, z_i)^2}(h_1^2)^{-1}K_1(\frac{c-x_j}{h_1}, \frac{z_i-z_j}{h_1})K\Big( \frac{x_i-c}{h}\Big)R_i + \frac{T_if_Z(z_i)}{f_{X, Z}(c, z_i)}K\Big( \frac{x_i-c}{h}\Big)R_i.
\end{split}
\end{equation*}
By Lemma \ref{phiijrd}, we have
\[\mathbb{E}(\phi_{i,j}^{rd}) =  O_p(h^3).\]
The central limit theorem implies
\begin{equation}\label{eqclt}
\frac{1}{(nh)^{1/2}}\sum_{i=1}^n \{\EE(\phi_{i,j}^{rd}+\phi_{j,i}^{rd}|i)-\delta^{rd}\}+O_p(\frac{n^{1/2}}{h^{1/2}}(h^3 +h^2h_1))\rightarrow_d N(0,\xi^2/h),
\end{equation}
where $\delta^{rd} = O_p(h^3)$ and $\xi^2=\EE\{(\EE(\phi_{i,j}^{rd}+\phi_{j,i}^{rd}|i)-\delta)^2\}=\mathbb{E}\left\{\mathbb{E}(\phi_{ij}|i)^2 + \mathbb{E}(\phi_{ji}|i)^2 + 2\mathbb{E}(\phi_{ij}|i)\mathbb{E}(\phi_{ji}|i) \right\}-\delta^2$. We now calculate the asymptotic variance as follows. Since $\mathbb{E}(\phi_{j,i}|i)=\mathbb{E}(\phi_{i,j}|j)$, from lemma \ref{phiij}, we have
\begin{equation*}
\begin{split}
\mathbb{E}(\phi_{ij}|j)  &=  \frac{h}{2} d_1(c^+, z_j) -T_j\frac{h}{h_1} K(\frac{c-x_j}{h_1})  \frac{f_Z(z_j)}{f_{X, Z}(c,z_j)}d_1(c^+, z_j) + O_p(h^2+hh_1+hh_2).
\end{split}
\end{equation*}
Similarly, we can show that 
\[\mathbb{E}(\phi_{i,j}|i) = \frac{T_if_Z(z_i)}{f_{X, Z}(c, z_i)}K\Big( \frac{x_i-c}{h}\Big)R_i +O_p(h_2^2 + h_1).\]
Recall that  $\sigma^2=\EE(Y(1)-m_1(X, Z))^2$. Since $h\asymp h_1\asymp h_2$, after some tedious calculation we can show that 
\begin{equation*}
\begin{split}
\frac{1}{n} \mathbb{E}\left(\mathbb{E}(\phi_{ij}|i)^2 \right)&  = \frac{\sigma^2}{n} \int_{z_i}\int_{x_i}\frac{T_if_Z(z_i)^2}{f_{X, Z}(c, z_i)^2}K\Big( \frac{x_i-c}{h}\Big)^2f_{X, Z}(x_i, z_i)dx_idz_i \\
& +\frac{1}{n} \int_{z_i}\int_{x_i}\frac{T_if_Z(z_i)^2}{f_{X, Z}(c, z_i)^2}K\Big( \frac{x_i-c}{h}\Big)^2f_{X, Z}(x_i, z_i) (m_1(x_i, z_i) - \alpha_1)^2dx_idz_i +O\left(\frac{h^2}{n}\right)\\
& = \frac{\sigma^2}{n}h\kappa_{20} \int_{z_i}\frac{f_Z(z_i)^2}{f_{X, Z}(c, z_i)} dz_i + \frac{h}{n} \kappa_{20}\int_{z_i}\frac{f_Z(z_i)^2}{f_{X, Z}(c, z_i)}d_1(c^+, z_i)^2 dz_i+O\left(\frac{h^2}{n}\right)\\
& = \kappa_{20}\frac{h}{n}\sigma^2 \mathbb{E}_Z\left(\frac{f_Z(z_i)}{f_{X, Z}(c, z_i)} \right) + \kappa_{20}\frac{h}{n} \mathbb{E}_Z\left(\frac{f_Z(z_i)}{f_{X, Z}(c, z_i)}d_1(c^+, z_i)^2 \right)+O\left(\frac{h^2}{n}\right),
\end{split}
\end{equation*}
\begin{equation*}
\begin{split}
\frac{1}{n} \mathbb{E}\left(\mathbb{E}(\phi_{ji}|i)^2 \right)&  = \frac{1}{n}\frac{h^2}{h_1^2} \int_{z_i}\int_{x_i}T_iK(\frac{c-x_i}{h_1})^2  \frac{f_Z(z_i)^2}{f_{X, Z}(c^+,z_i)^2}d_1(c^+, z_i)^2f_{X, Z}(x_i, z_i)dx_idz_i + O\left(\frac{h^2}{n}\right)\\
& = \frac{h^2}{h_1} \frac{\kappa_{20}}{n}\int_{z_i}\frac{f_Z(z_i)^2}{f_{X, Z}(c, z_i)}d_1(c^+, z_i)^2 dz_i+O\left(\frac{h^2}{n}\right)\\
& = \frac{h^2}{h_1} \frac{\kappa_{20}}{n}\mathbb{E}_Z\left(\frac{f_Z(z_i)}{f_{X, Z}(c, z_i)}d_1(c^+, z_i)^2 \right)+O\left(\frac{h^{2}}{n}\right),
\end{split}
\end{equation*}
\begin{equation*}
\begin{split}
\frac{2}{n} \mathbb{E}\left(\mathbb{E}(\phi_{ij}|i) \mathbb{E}(\phi_{ji}|i) \right)&  = -2 \cdot \frac{h}{h_1n} \int_{z_i}\int_{x_i}T_iK(\frac{c-x_j}{h_1})^2  \frac{f_Z(z_j)^2}{f_{X, Z}(c,z_j)^2}d_1(c^+, z_j)^2f_{X, Z}(x_i, z_i)dx_idz_i+O\left(\frac{h^{2}}{n}\right) \\
& = -2 \cdot\kappa_{20}\frac{h}{n} \int_{z_i}\frac{f_Z(z_i)^2}{f_{X, Z}(c, z_i)}d_1(c^+, z_j)^2 dz_i + O\left(\frac{h^{2}}{n}\right)\\
& =  -2 \cdot\kappa_{20}\frac{h}{n} \mathbb{E}_Z\left(\frac{f_Z(z_i)}{f_{X, Z}(c|z_i)}d_1(c^+, z_i)^2 \right) + O\left(\frac{h^{2}}{n}\right).
\end{split}
\end{equation*}
Thus, choosing $h_1 = h$, we have
$$
\xi^2=h\kappa_{20}\underbrace{\sigma^2 \mathbb{E}_Z\left( \frac{f_Z(z_i)}{f_{X, Z}(c, z_i)}\right)}_{\omega}  + O(h^2).
$$
Combining (\ref{equstat}) and (\ref{eqclt}), 
$$
\frac{n^{1/2}}{h^{1/2}}\Big(\frac{1}{n(n-1)}\sum_{i<j}(\phi_{i,j}+\phi_{j,i})-\delta/2\Big)+O\Big(\frac{n^{1/2}}{h^{1/2}}(h^3 +h^2h_1)\Big)\rightarrow_d N(0,\xi^2/h).
$$
Finally, note that in (\ref{eqgradient}), 
$$
\frac{1}{n^2}\sum_{i=1}^n\phi_{i,i}\lesssim \frac{1}{n}\EE(\phi_{i,j}) = O_p\left( \frac{h^2}{n}\right),
$$
and therefore we obtain that
$$
\frac{1}{(nh)^{1/2}}\sum_{i=1}^n \frac{T_i \hat{f}_Z(z)}{\hat{f}_{X, Z}(c,z)}  K\Big( \frac{x_i-c}{h}\Big)R_i+\phi\rightarrow_d N(\frac{n^{1/2}\delta}{2h^{1/2}},\xi^2/h),
$$
where  
$$
\phi=O_p\Big(\frac{n^{1/2}}{h^{1/2}}(h^3+r^2+rs)\Big).
$$
Following the similar argument, we can show the joint convergence
$$
\frac{1}{(nh)^{1/2}}\sum_{i=1}^n \frac{T_i \hat{f}_Z(z)}{\hat{f}_{X, Z}(c,z)}  K\Big( \frac{x_i-c}{h}\Big) R_i [1, (x_i-c)]^T\rightarrow_d N\left(\frac{n^{1/2}}{h^{1/2}} \left(\begin{matrix}
O_p(h^3) \\O_p(h^4)
\end{matrix}\right), \omega\left(\begin{matrix}
 \kappa_{20} & h\kappa_{21} \\
  h\kappa_{21} & h^2\kappa_{22} \\
 \end{matrix}\right)
\right).
$$
By the least squared formulation, the estimator $\hat\alpha_1$ satisfies
$$
\sqrt{nh}(\hat{\alpha}_1 - \alpha_1)=-e_1^T\left(\begin{matrix}
\hat C_n/(nh) & -\hat B_n/(nh)  \\ -\hat B_n/(nh)  & \hat A_n/(nh) 
\end{matrix}\right)^{-1} \frac{1}{(nh)^{1/2}}\sum_{i=1}^n \frac{T_i \hat{f}_Z(z)}{\hat{f}_{X, Z(1)}(c,z)}  K\Big( \frac{x_i-c}{h}\Big) R_i [1, (x_i-c)]^T,
$$
where $e_1^T=(1,0)$.
From lemma \ref{lemABCD} and the matrix inversion formula, 
$$
\left(\begin{matrix}
\hat C_n/(nh) & -\hat B_n/(nh)  \\ -\hat B_n/(nh)  & \hat A_n/(nh) 
\end{matrix}\right)^{-1}=\frac{1}{\hat D_n/(nh)^2} \left(\begin{matrix}
\hat A_n/(nh) & \hat B_n/(nh) \\ \hat B_n/(nh) & \hat C_n/(nh)
\end{matrix}\right)\rightarrow_p \frac{1}{h^2(\kappa_2/2 -\kappa_1^2)}\left(\begin{matrix}
h^2 \kappa_2 & -h \kappa_1 \\ -h \kappa_1 & \frac{1}{2}
\end{matrix}\right).
$$
Thus, the asymptotic bias of $\sqrt{nh}(\hat{\alpha}_1 - \alpha_1)$ is
$$
\frac{-e_1^T}{h^2(\kappa_2/2 -\kappa_1^2)}\left(\begin{matrix}
h^2 \kappa_2 & -h \kappa_1 \\ -h \kappa_1 & \frac{1}{2}
\end{matrix}\right)\frac{n^{1/2}}{h^{1/2}} \left(\begin{matrix}
 O_p(h^3) \\ O_p(h^4)
\end{matrix}\right)=O(n^{1/2}h^{5/2})=o(1).
$$
Similarly, the asymptotic variance  of $\sqrt{nh}(\hat{\alpha}_1 - \alpha_1)$ is
\begin{align*}
&\frac{\omega }{h^4(\kappa_2/2 -\kappa_1^2)^2}e_1^T\left(\begin{matrix}
h^2 \kappa_2 & -h \kappa_1 \\ -h \kappa_1 & \frac{1}{2}
&\end{matrix}\right)\left(\begin{matrix}
 \kappa_{20} & h\kappa_{21} \\
  h\kappa_{21} & h^2\kappa_{22} \\
 \end{matrix}\right)
\left(\begin{matrix}
h^2 \kappa_2 & -h \kappa_1 \\ -h \kappa_1 & \frac{1}{2}
\end{matrix}\right)e_1\\
&=\sigma^2\int  \frac{f_Z(z_i)^2}{f_{X, Z}(c, z_i)}dz_i\cdot C_v,
\end{align*}
where 
\[C_v = \frac{\kappa_2^2\kappa_{20}  + \kappa_1^2\kappa_{22} - 2\kappa_1\kappa_2\kappa_{21}}{\left(\frac{1}{2}\kappa_2-\kappa_1^2\right)^2}.\]
This completes the proof.

\begin{lemma}\label{phiijrd}
Recall that 
\begin{equation*}
\begin{split}
\phi_{i,j}^{rd}  &= \frac{T_i}{f_{X, Z}(c, z_i)} (h_2)^{-1}K(\frac{z_i-z_j}{h_2}) K\Big( \frac{x_i-c}{h}\Big)R_i\\
 &- \frac{T_if_Z(z_i)}{f_{X, Z}(c, z_i)^2}(h_1^2)^{-1}K_1(\frac{c-x_j}{h_1}, \frac{z_i-z_j}{h_1})K\Big( \frac{x_i-c}{h}\Big)R_i + \frac{T_if_Z(z_i)}{f_{X, Z}(c, z_i)}K\Big( \frac{x_i-c}{h}\Big)R_i.
\end{split}
\end{equation*}
Under the same condition as in Theorem \ref{thm2}, and choose $h_1 \asymp h$ 
\begin{align*}
&\mathbb{E}(\phi_{i,j}^{rd}|j) = \frac{h}{2} d_1(c^+, z_j) -\frac{h}{2h_1} K(\frac{c-x_j}{h_1})  \frac{f_Z(z_j)}{f_{X, Z}(c,z_j)}d_1(c^+, z_j) + O_p(h^2+hh_1+hh_2),\\
&\mathbb{E}(\phi_{i,j}^{rd}|i) = \frac{T_if_Z(z_i)}{f_{X, Z}(c, z_i)}K\Big( \frac{x_i-c}{h}\Big)R_i +O_p(h_2^2 + h_1^2), \\
&\mathbb{E}(\phi_{i,j}^{rd}) = O_p(h^3),
\end{align*}
where $O_p$ terms are valid uniformly over $i$ or $j$.
\end{lemma}

\begin{proof}
\begin{equation*}
\begin{split}
\mathbb{E}(\phi_{i,j}^{rd}|j)  &=\underbrace{\mathbb{E}\left( \frac{T_i}{f_{X, Z}(c, z_i)}  (h_2)^{-1}K(\frac{z_i-z_j}{h_2})  K\Big( \frac{x_i-c}{h}\Big)R_i\Big|z_j\right)}_{\text{(Part.EI)}}\\
 &- \underbrace{\mathbb{E}\left(  \frac{T_if_Z(z_i)}{f_{X, Z}(c, z_i)^2} (h_1^2)^{-1}K_1(\frac{c-x_j}{h_1}, \frac{z_i-z_j}{h_1}) K\Big( \frac{x_i-c}{h}\Big)R_i\Big|z_j, x_j \right)}_{\text{(Part.EII)}}\\
 & + \underbrace{\mathbb{E}\left(\frac{T_if_Z(z_i)}{f_{X, Z}(c, z_i)}K\Big( \frac{x_i-c}{h}\Big)R_i\Big| z_j,x_j \right)}_{\text{(Part.EIII)}}.
\end{split}
\end{equation*}
From lemma \ref{lemM},
\begin{equation*}
\begin{split}
\text{(Part.EI)} & = \mathbb{E}\left( \frac{T_i}{f_{X, Z}(c, z_i)}  (h_2)^{-1}K(\frac{z_i-z_j}{h_2})  K\Big( \frac{x_i-c}{h}\Big)R_i\Big|z_j\right)\\
& = h\int_{z_i} \frac{1}{f_{X, Z}(c, z_i)}  (h_2)^{-1}K(\frac{z_i-z_j}{h_2}) M(c, z_i) dz_i\\
& = h\int_{z_i}   (h_2)^{-1}K(\frac{z_i-z_j}{h_2})\frac{1}{2} (m_1(c^+, z_i) - \alpha_1) dz_i + O_p(h^2)\\
& = \frac{h}{2} (m_1(c^+, z_j) - \alpha_1) + O_p(h^2 +h_2h),
\end{split}
\end{equation*}
where $O_p$ terms are valid uniformly over $j$.  Similarly, for Part.EII and Part.EIII, we can show that
\begin{equation*}
\begin{split}
&\text{(Part.EII)} = \mathbb{E}\left(  \frac{T_if_Z(z_i)}{f_{X, Z}(c, z_i)^2} (h_1^2)^{-1}K_1(\frac{c-x_j}{h_1}, \frac{z_i-z_j}{h_1}) K\Big( \frac{x_i-c}{h}\Big)R_i\Big|z_j, x_j \right) \\
& = h\int_{z_i} \frac{f_Z(z_i)}{f_{X, Z}(c, z_i)^2}  (h_1^2)^{-1}K_1(\frac{c-x_j}{h_1}, \frac{z_i-z_j}{h_1}) M(c, z_i)dz_i\\
& = \frac{h}{2}\int_{z_i}  \frac{f_Z(z_i)}{f_{X, Z}(c, z_i)} (h_1^2)^{-1}K_1(\frac{c-x_j}{h_1}, \frac{z_i-z_j}{h_1})  (m_1(c^+, z_i) - \alpha_1) dz_i + O_p(h^2)\\
& = \frac{h}{2h_1} K(\frac{c-x_j}{h_1})  \frac{f_Z(z_j)}{f_{X, Z}(c,z_j)} (m_1(c^+, z_j) - \alpha_1)+ O_p\left(hh_1\right)  + O_p(h^2),
\end{split}
\end{equation*}
\begin{equation*}
\begin{split}
\text{(Part.EIII)} & = \mathbb{E}\left(\frac{T_if_Z(z_i)}{f_{X, Z}(c, z_i)}K\Big( \frac{x_i-c}{h}\Big)R_i\Big| z_j,x_j \right)
 = \frac{h}{2}\int_{z_i}  (m_1(c^+, z_i) - \alpha_1)f_Z(z_i) dz_i + O_p(h^2) = O_p(h^2),
\end{split}
\end{equation*}
where the last step follows from the definition of $\alpha_1$. Define $d_1(x_i, z_i) = m_1(x_i, z_i) - \alpha_1$. Combining the Part EI, EII and EIII, we obtain 
\begin{equation*}
\begin{split}
\mathbb{E}(\phi_{ij}^{rd}|j)  &=  \frac{h}{2} d_1(c^+, z_j) -\frac{h}{2h_1} K(\frac{c-x_j}{h_1})  \frac{f_Z(z_j)}{f_{X, Z}(c,z_j)}d_1(c^+, z_j) + O_p(h^2+hh_1+hh_2).
\end{split}
\end{equation*}
Following the similar calculation, by lemma \ref{lemkernel} we have 
\begin{equation*}
\begin{split}
\mathbb{E}(\phi_{i,j}^{rd}|i)  &= \underbrace{\frac{T_i}{f_{X, Z}(c, z_i)} \mathbb{E}\left( (h_2)^{-1}K(\frac{z_i-z_j}{h_2}) \Big|z_i\right) K\Big( \frac{x_i-c}{h}\Big)R_i}_{\text{Part. FI}}\\
 &-\underbrace{ \frac{T_if_Z(z_i)}{f_{X, Z}(c, z_i)^2}\mathbb{E}\left(  (h_1^2)^{-1}K_1(\frac{c-x_j}{h_1}, \frac{z_i-z_j}{h_1}) \Big| z_i,x_i \right)K\Big( \frac{x_i-c}{h}\Big)R_i}_{\text{Part. FII}} \\
 &+ \underbrace{\frac{T_if_Z(z_i)}{f_{X, Z}(c, z_i)}K\Big( \frac{x_i-c}{h}\Big)R_i}_{\text{Part. FIII}}\\
&= \frac{T_if_Z(z_i)}{f_{X, Z}(c, z_i)}K\Big( \frac{x_i-c}{h}\Big)R_i +O_p(h_1^2 + h_2^2).
\end{split}
\end{equation*}
Finally, we calculate $\mathbb{E}(\phi_{i,j}^{rd})$ using Lemma \ref{lemM}, 
\begin{equation*}
\begin{split}
\mathbb{E}(\text{ Part.FIII} ) & = h\int   \frac{ f_Z(z_i)}{f_{X, Z}(c, z_i)} M(c, z_i)dz_i  \\
& = \frac{h}{2}\int  (m_1(c^+, z_i) - \alpha_1) f_Z(z_i)dz_i- h^2\kappa_1\beta_1- h^2\kappa_1 \alpha_1 \int \frac{\partial f_{X, Z}(c, z_i)}{\partial x_i} \frac{f_Z(z_i)}{f_{X, Z}(c, z_i)} dz_i\\
& + h^2 \kappa_1\int\left(\frac{\partial m_1(c^+, z_i)}{\partial x_i}f_Z(z_i)+m_1(c^+, z_i) \frac{\partial f_{X, Z}(c, z_i)}{\partial x_i}\frac{f_Z(z_i)}{f_{X, Z}(c, z_i)} \right)dz_i   +O_p( h^3)\\
& = -h^2\kappa_1 \alpha_1 \int \frac{\partial f_{X, Z}(c, z_i)}{\partial x_i} \frac{f_Z(z_i)}{f_{X, Z}(c, z_i)} dz_i +O_p(h^3)
\end{split}
\end{equation*}
where the $O_p$ terms are valid uniformly over $i$ and the last equality follows as 
 \[\alpha_1 = \int   m_1(c^+, z) f_{Z}(z)d{z},\]
and
\[\beta_1 = \int  \frac{\partial m_1(c^+, z_i)}{\partial x_i} f_Z(z_i) dz_i + \int m_1(c^+, z_i) \frac{\partial f_{X, Z}(c, z_i)}{\partial x_i} \frac{f_Z(z_i)}{f_{X, Z}(c, z_i)} dz_i.\]  
\\
From Lemma \ref{lemkernel}, some tedious calculation implies
\begin{equation*}
\begin{split}
\mathbb{E}(\text{Part.FI} ) & =  \mathbb{E}(\text{Part.FIII} )+O_p(hh_2^2),
\end{split}
\end{equation*}
and 
\begin{equation*}
\begin{split}
\mathbb{E}(\text{Part.FII} ) & =  \mathbb{E}(\text{Part.FIII} ) +O_p(hh_1^2).
\end{split}
\end{equation*}
Thus when $h_1 \asymp h$
\[\mathbb{E}(\phi_{i,j}) =  O_p(h^3).\]
This completes the proof. 
\end{proof}

\subsection{Proof of Theorem \ref{thm4}}

\begin{proof}
Assume the treatment assignment function for individual $i$ is $T_i(X_i)$ and only depends on $X_i$ directly. The assignment function is not known to the individual but they are aware of a jump in $T_i(X_i)$ if $X_i>c$. Thus a post selection may exist and we can write 

\[Z_i = (X_i>c)Z_i(1) + (X_i\leq c)Z_i(0)\]

\begin{equation*}
\begin{split}
\mathbb{E}(T_iY_i|X_i = c^+,Z_i=z) &= \mathbb{E}(T_iY_i|X_i = c^+,Z_i=z, \tau_\epsilon = co) \mathbb{P}(\tau_\epsilon = co | X_i = c^+, Z_i = z)\\
& + \mathbb{E}(T_iY_i|X_i = c^+,Z_i=z, \tau_\epsilon = at) \mathbb{P}(\tau_\epsilon = at | X_i = c^+, Z_i = z)\\
& + \mathbb{E}(T_iY_i|X_i = c^+,Z_i=z, \tau_\epsilon = nt) \mathbb{P}(\tau_\epsilon = nt | X_i = c^+, Z_i = z)\\
\end{split}
\end{equation*}

\begin{equation*}
\begin{split}
\mathbb{E}(T_iY_i|X_i = c^-,Z_i=z) &= \mathbb{E}(T_iY_i|X_i = c^-,Z_i=z, \tau_\epsilon = co) \mathbb{P}(\tau_\epsilon = co | X_i = c^-, Z_i = z)\\
& + \mathbb{E}(T_iY_i|X_i = c^-,Z_i=z, \tau_\epsilon = at) \mathbb{P}(\tau_\epsilon = at | X_i = c^-, Z_i = z)\\
& + \mathbb{E}(T_iY_i|X_i = c^-,Z_i=z, \tau_\epsilon = nt) \mathbb{P}(\tau_\epsilon = nt | X_i = c^-, Z_i = z)\\
\end{split}
\end{equation*}

Assume $\mathbb{P}(\tau_\epsilon  | X_i = c^+, Z_i(1) = z) = \mathbb{P}(\tau_\epsilon | X_i = c^-, Z_i(0) = z)$


\begin{equation*}
\begin{split}
& \mathbb{E}(T_iY_i|X_i = c^+,Z_i=z) - \mathbb{E}(T_iY_i|X_i = c^-,Z_i=z) \\
&  = \Big(\mathbb{E}(T_iY_i|X_i = c^+,Z_i=z, \tau_\epsilon = co) - \mathbb{E}(T_iY_i|X_i = c^-,Z_i=z, \tau_\epsilon = co)\Big) \mathbb{P}(\tau_\epsilon = co | X_i = c, Z_i = z)\\
& + \Big(\mathbb{E}(T_iY_i|X_i = c^+,Z_i=z, \tau_\epsilon = at) - \mathbb{E}(T_iY_i|X_i = c^-,Z_i=z, \tau_\epsilon = at)\Big) \mathbb{P}(\tau_\epsilon = at | X_i = c, Z_i = z)\\
& + \Big(\mathbb{E}(T_iY_i|X_i = c^+,Z_i=z, \tau_\epsilon = nt) - \mathbb{E}(T_iY_i|X_i = c^-,Z_i=z, \tau_\epsilon = nt)\Big) \mathbb{P}(\tau_\epsilon = nt | X_i = c, Z_i = z)\\
&  = \mathbb{E}(Y_i(1)|Z_i(1)=z, \tau_\epsilon = co)\cdot  \mathbb{P}(\tau_\epsilon = co | X_i = c, Z_i = z)\\
& + \Big(\mathbb{E}(Y_i(1)|Z_i(1)=z, \tau_\epsilon = at) - \mathbb{E}(Y_i(1)|Z_i(0)=z, \tau_\epsilon = at)\Big) \mathbb{P}(\tau_\epsilon = at | X_i = c, Z_i = z)
\end{split}
\end{equation*}

Similarly

\begin{equation*}
\begin{split}
& \mathbb{E}((1-T_i)Y_i|X_i = c^+,Z_i=z) - \mathbb{E}((1-T_i)Y_i|X_i = c^-,Z_i=z) \\
&  = \Big(\mathbb{E}((1-T_i)Y_i|X_i = c^+,Z_i=z, \tau_\epsilon = co) - \mathbb{E}((1-T_i)Y_i|X_i = c^-,Z_i=z, \tau_\epsilon = co)\Big) \mathbb{P}(\tau_\epsilon = co | X_i = c, Z_i = z)\\
& + \Big(\mathbb{E}((1-T_i)Y_i|X_i = c^+,Z_i=z, \tau_\epsilon = at) - \mathbb{E}((1-T_i)Y_i|X_i = c^-,Z_i=z, \tau_\epsilon = at)\Big) \mathbb{P}(\tau_\epsilon = at | X_i = c, Z_i = z)\\
& + \Big(\mathbb{E}((1-T_i)Y_i|X_i = c^+,Z_i=z, \tau_\epsilon = nt) - \mathbb{E}((1-T_i)Y_i|X_i = c^-,Z_i=z, \tau_\epsilon = nt)\Big) \mathbb{P}(\tau_\epsilon = nt | X_i = c, Z_i = z)\\
&  = - \mathbb{E}(Y_i(0)|Z_i(0)=z, \tau_\epsilon = co) \cdot \mathbb{P}(\tau_\epsilon = co | X_i = c, Z_i = z)\\
& + \Big(\mathbb{E}(Y_i(0)|Z_i(1)=z, \tau_\epsilon = nt) - \mathbb{E}(Y_i(0)|Z_i(0)=z, \tau_\epsilon = nt)\Big) \mathbb{P}(\tau_\epsilon = nt | X_i = c, Z_i = z)\\
\end{split}
\end{equation*}

Assume $\mathbb{E}(Y_i(1)|Z_i(c^+)=z, \tau_\epsilon = at) = \mathbb{E}(Y_i(1)|Z_i(c^-)=z, \tau_\epsilon = at)$ and $\mathbb{E}(Y_i(0)|Z_i(c^+)=z, \tau_\epsilon = nt) = \mathbb{E}(Y_i(0)|Z_i(c^-)=z, \tau_\epsilon = nt)$, then

\begin{equation}\label{fuzzyeq1}
\begin{split}
&\mathbb{E}(Y_i|X_i = c^+,Z_i=z) - \mathbb{E}(Y_i|X_i = c^-,Z_i=z)\\
& = \underbrace{\Big(\mathbb{E}(Y_i(1)|Z_i(1)=z, \tau_\epsilon = co) - \mathbb{E}(Y_i(0)|Z_i(0)=z, \tau_\epsilon = co)\Big)}_{:=\Delta} \cdot \mathbb{P}(\tau_\epsilon = co | X_i = c, Z_i = z)
\end{split}
\end{equation}

Apply the same decomposition to $T_i$ 
\begin{equation}\label{fuzzyeq2}
\begin{split}
& \mathbb{E}(T_i|X_i = c^+,Z_i=z) - \mathbb{E}(T_i|X_i = c^-,Z_i=z) \\
&  = \Big(\mathbb{E}(T_i|X_i = c^+,Z_i=z, \tau_\epsilon = co) - \mathbb{E}(T_i|X_i = c^-,Z_i=z, \tau_\epsilon = co)\Big) \mathbb{P}(\tau_\epsilon = co | X_i = c, Z_i = z)\\
& + \Big(\mathbb{E}(T_i|X_i = c^+,Z_i=z, \tau_\epsilon = at) - \mathbb{E}(T_i|X_i = c^-,Z_i=z, \tau_\epsilon = at)\Big) \mathbb{P}(\tau_\epsilon = at | X_i = c, Z_i = z)\\
& + \Big(\mathbb{E}(T_i|X_i = c^+,Z_i=z, \tau_\epsilon = nt) - \mathbb{E}(T_i|X_i = c^-,Z_i=z, \tau_\epsilon = nt)\Big) \mathbb{P}(\tau_\epsilon = nt | X_i = c, Z_i = z)\\
&  =  \mathbb{P}(\tau_\epsilon = co | X_i = c, Z_i = z)\\
\end{split}
\end{equation}

First notice that when integrating both sides with $f_{Z(0)|X}(z|c)$, 
\begin{align*}
&\mathbb{E}_{Z(0)}\left(\mathbb{E}(Y_i|X_i = c^+,Z_i=z) - \mathbb{E}(Y_i|X_i = c^-,Z_i=z)\Big|X_i=c^-\right) \\
&= \int \Delta\cdot \mathbb{P}(\tau_\epsilon = co | X_i = c, Z_i = z)f_{Z(0)|X}(z|c^-)dz\\
& = \mathbb{P}(\tau_\epsilon = co | X_i = c) \int \Delta\cdot f_{Z(0)|X,\tau_\epsilon = co }(z|c^-)dz\\
& = \mathbb{P}(\tau_\epsilon = co | X_i = c) \cdot \mathbb{E}_{Z(0)}\Big(\mathbb{E}(Y_i(1)|Z_i(c^+)=z, \tau_\epsilon = co) - \mathbb{E}(Y_i(0)|Z_i(c^-)=z, \tau_\epsilon = co)\Big|  X_i = c^-, \tau_\epsilon = co\Big)
\end{align*}

And devide \eqref{fuzzyeq1} with \eqref{fuzzyeq2}, we obtain
\begin{align*}
&\mathbb{E}_{Z(0)}\Big(\mathbb{E}(Y_i(1)|Z_i(c^+)=z, \tau_\epsilon = co) - \mathbb{E}(Y_i(0)|Z_i(c^-)=z, \tau_\epsilon = co)\Big|  X_i = c^-, \tau_\epsilon = co\Big)\\
& = \frac{\mathbb{E}_{Z(0)}\left(\mathbb{E}(Y_i|X_i = c^+,Z_i=z) - \mathbb{E}(Y_i|X_i = c^-,Z_i=z)\Big| X_i = c^-\right)}{\mathbb{E}_{Z(0)}\left(\mathbb{E}(T_i|X_i = c^+,Z_i=z) - \mathbb{E}(T_i|X_i = c^-,Z_i=z)\Big| X_i = c^-\right)}
\end{align*}

Similarly, when integrating both sides with $\frac{(f_{Z(0)|X}(z_i, c^-) + f_{Z(1)|X}(z_i, c^+))}{2}$

\begin{align*}
&\int \big(\mathbb{E}(Y_i|X_i = c^+,Z_i=z) - \mathbb{E}(Y_i|X_i = c^-,Z_i=z)\Big)  \frac{f_{Z(0)|X}(z_i|c^-) + f_{Z(1)|X}(z_i|c^+)}{2} dz \\
&= \int \Delta\cdot \mathbb{P}(\tau_\epsilon = co | X_i = c, Z_i = z)\frac{f_{Z(0)|X}(z_i|c^-) + f_{Z(1)|X}(z_i|c^+)}{2} dz\\
& = \mathbb{P}(\tau_\epsilon = co | X_i = c) \int \Delta\cdot \frac{f_{Z(0)|X, \tau_\epsilon = co}(z_i|c^-) + f_{Z(1)|X, \tau_\epsilon = co}(z_i|c^+)}{2} dz\\
\end{align*}

And we obtain the estimand
\begin{align*}
&\lim_{\epsilon \rightarrow 0} \mathbb{E}\Big(\mathbb{E}(Y_i(1)|Z_i(1)=z, \tau_\epsilon = co) - \mathbb{E}(Y_i(0)|Z_i(0)=z, \tau_\epsilon = co)\Big|X \in \mathcal{N}_\epsilon,  \tau_\epsilon = co\Big)\\
& = \frac{\int \left(\mathbb{E}(Y_i|X_i = c^+,Z_i=z) - \mathbb{E}(Y_i|X_i = c^-,Z_i=z)\right) (f_{Z(0)|X}(z_i|c^-) + f_{Z(1)|X}(z_i|c^+)) dz }{\int \left(\mathbb{E}(T_i|X_i = c^+,Z_i=z) - \mathbb{E}(T_i|X_i = c^-,Z_i=z)\right) (f_{Z(0)|X}(z_i|c^-) + f_{Z(1)|X}(z_i|c^+)) dz}
\end{align*}

Lastly under CIA, one can show that 

\begin{equation}
\begin{split}
&\mathbb{E}(Y_i|X_i = c^+,Z_i=z) - \mathbb{E}(Y_i|X_i = c^-,Z_i=z)\\
& = \underbrace{\Big(\mathbb{E}(Y_i(1)|Z_i(1)=z, \tau_\epsilon = co) - \mathbb{E}(Y_i(0)|Z_i(0)=z, \tau_\epsilon = co)\Big)}_{:=\Delta} \cdot \mathbb{P}(\tau_\epsilon = co |  Z_i = z)
\end{split}
\end{equation}

and

\begin{equation}
\begin{split}
\mathbb{E}(T_i|X_i = c^+,Z_i=z) - \mathbb{E}(T_i|X_i = c^-,Z_i=z)  =  \mathbb{P}(\tau_\epsilon = co |Z_i = z)\\
\end{split}
\end{equation}

when we integrating with both sides with $f_{Z}(z_i)$, 

\begin{align*}
&\int \big(\mathbb{E}(Y_i|X_i = c^+,Z_i=z) - \mathbb{E}(Y_i|X_i = c^-,Z_i=z)\Big)  f_{Z}(z_i) dz \\
&= \int \Delta\cdot \mathbb{P}(\tau_\epsilon = co |  Z_i = z)f_{Z}(z_i) dz\\
& = \mathbb{P}(\tau_\epsilon = co) \int \Delta\cdot f_{Z}(z_i) dz\\
\end{align*}

Thus

\begin{align*}
&\mathbb{E}_{Z}\Big(\mathbb{E}(Y_i(1)|Z_i(1)=z, \tau_\epsilon = co) - \mathbb{E}(Y_i(0)|Z_i(0)=z, \tau_\epsilon = co)\Big)\\
& = \frac{\mathbb{E}_{Z}\left(\mathbb{E}(Y_i|X_i = c^+,Z_i=z) - \mathbb{E}(Y_i|X_i = c^-,Z_i=z)\right)}{\mathbb{E}_{Z}\left(\mathbb{E}(T_i|X_i = c^+,Z_i=z) - \mathbb{E}(T_i|X_i = c^-,Z_i=z)\right)}
\end{align*}

\end{proof}

\begin{lemma}\label{lemmahoff}
When either $h\asymp h_1 \asymp h_2$ or $h\asymp \sqrt{h_1} \asymp h_2$, the condition in Theorem 12.3 in \cite{van2000asymptotic} holds, such that $\mathbb{E}(\phi_{ij}^2)<\infty$ 
\end{lemma}
\begin{proof}
$\mathbb{E}(\phi_{ij}^2)$ is a linear combination of the following quantities:
\begin{align}
 &\mathbb{E}\left( \frac{1}{f_{X, Z(1)}(c^+, z_i)^2} (h_2)^{-2}K^2(\frac{z_i-z_j}{h_2}) K^2\Big( \frac{x_i-c}{h}\Big)R_i^2 \right)\label{eqA1}\\[10pt]
 &\mathbb{E}\left( \frac{f_Z(z_i)^2}{f_{X, Z(1)}(c^+, z_i)^4}(h_1)^{-4}K_1^2(\frac{c-x_j}{h_1}, \frac{z_i-z_j}{h_1})K^2\Big( \frac{x_i-c}{h}\Big)R_i^2 \right)\label{eqA2}\\[10pt]
&\mathbb{E}\left( \frac{f_Z^2(z_i)}{f_{X, Z(1)}^2(c^+, z_i)}K^2\Big( \frac{x_i-c}{h}\Big)R_i^2\right)\label{eqA3}\\[10pt]
& \mathbb{E}\left( \frac{f_Z(z_i)}{f_{X, Z(1)}^3(c^+, z_i)} (h_2)^{-1}(h_1)^{-2}K_1(\frac{c-x_j}{h_1}, \frac{z_i-z_j}{h_1})K(\frac{z_i-z_j}{h_2}) K^2\Big( \frac{x_i-c}{h}\Big)R_i^2\right)\label{eqA4} \\[10pt]
& \mathbb{E}\left( \frac{f_Z(z_i)}{f_{X, Z(1)}^2(c^+, z_i)} (h_2)^{-1}K(\frac{z_i-z_j}{h_2}) K^2\Big( \frac{x_i-c}{h}\Big)R_i^2\right)\label{eqA5} \\[10pt]
&\mathbb{E}\left(\frac{f_Z^2(z_i)}{f_{X, Z(1)}(c^+, z_i)^3}(h_1^2)^{-1}K_1(\frac{c-x_j}{h_1}, \frac{z_i-z_j}{h_1})K^2\Big( \frac{x_i-c}{h}\Big)R_i^2\right)\label{eqA6}\\[10pt] \nonumber
\end{align}

\begin{itemize}

\item From Lemma \ref{lemAM}, equation \ref{eqA1} can be written as
\begin{align*}
&\mathbb{E}\left(\frac{1}{f_{X, Z(1)}(c^+, z_i)^2} (h_2)^{-2}K^2(\frac{z_i-z_j}{h_2}) K^2\Big( \frac{x_i-c}{h}\Big)R_i^2 \Big|z_j\right)\\
& = \int_{z_i}  \frac{1}{f_{X, Z(1)}(c^+, z_i)^2} (h_2)^{-1}K^2(\frac{z_i-z_j}{h_2})  J(x_i, z_i)  dz_i \qquad \text{($z_j \independent (x_i, z_i)$)}\\
& = \kappa_{20}\int_{z_i}  \frac{1}{f_{X, Z(1)}(c^+, z_i)} (h_2)^{-1}K^2(\frac{z_i-z_j}{h_2})  (\sigma^2 + d_1(c^+, z_i)^2)     dz_i  +O_p(h)\\
& =   \frac{\kappa_{20}^2}{f_{X, Z(1)}(c^+, z_j)}  (\sigma^2 + d_1(c^+, z_j)^2)   +O_p(h)
\end{align*}

Thus
\begin{align*}
&\mathbb{E}\left(\frac{1}{f_{X, Z(1)}(c^+, z_i)^2} (h_2)^{-2}K^2(\frac{z_i-z_j}{h_2}) K^2\Big( \frac{x_i-c}{h}\Big)R_i^2 \right)\\
& = \kappa_{20}^2\left(\sigma^2 \int \frac{f_Z(z_j)}{f_{X, Z(1)}(c^+, z_j)}dz_j + \int d_1(c^+, z_j)^2 \frac{f_Z(z_j)}{f_{X, Z(1)}(c^+, z_j)}dz_j\right) +O_p(h) < \infty
\end{align*}

\item When $h\asymp h_1 \asymp h_2$, by law of iterated expectation and Lemma \ref{lemAM}, equation \ref{eqA2} can be written as

\begin{align*}
&\mathbb{E}\left( \frac{f_Z(z_i)^2}{f_{X, Z(1)}(c^+, z_i)^4}(h_1)^{-4}K_1^2(\frac{c-x_j}{h_1}, \frac{z_i-z_j}{h_1})K^2\Big( \frac{x_i-c}{h}\Big)R_i^2 \Big| x_j, z_j\right)\\
& = \int_{z_i} \frac{f_Z(z_i)^2}{f_{X, Z(1)}(c^+, z_i)^4}(h_1)^{-3}K_1^2(\frac{c-x_j}{h_1}, \frac{z_i-z_j}{h_1}) J(c, z_i) dz_i\\
& =  \kappa_{20}\int_{z_i} \frac{f_Z(z_i)^2}{f_{X, Z(1)}(c^+, z_i)^3}(h_1)^{-3}K_1^2(\frac{c-x_j}{h_1}, \frac{z_i-z_j}{h_1}) (\sigma^2 + d_1(c^+, z_i)^2)  dz_i +O_p(h)\\
& = \kappa_{20}^2 \frac{f_Z(z_j)^2}{f_{X, Z(1)}(c^+, z_j)^3} (\sigma^2 + d_1(c^+, z_j)^2)(h_1)^{-1}K^2(\frac{c-x_j}{h_1})  +O_p(h) 
\end{align*}

Thus 
\begin{align*}
&\mathbb{E}\left( \frac{f_Z(z_i)^2}{f_{X, Z(1)}(c^+, z_i)^4}(h_1)^{-4}K_1^2(\frac{c-x_j}{h_1}, \frac{z_i-z_j}{h_1})K^2\Big( \frac{x_i-c}{h}\Big)R_i^2 \right)\\
& = \kappa_{20}^2 \int \int \frac{f_Z(z_j)^2}{f_{X, Z(1)}(c^+, z_j)^3} (\sigma^2 + d_1(c^+, z_j)^2)(h_1)^{-1}K^2(\frac{c-x_j}{h_1}) f_{X, Z(1)}(x_j, z_j)dz_jdx_j   +O_p(h) \\
& = \kappa_{20}^3 \int \frac{f_Z(z_j)^2}{f_{X, Z(1)}(c^+, z_j)^2} (\sigma^2 + d_1(c^+, z_j)^2) dz_j   +O_p(h) \\
& = \kappa_{20}^3\left(\sigma^2 \int \frac{f_Z(z_j)^2}{f_{X, Z(1)}(c^+, z_j)^2}dz_j + \int d_1(c^+, z_j)^2 \frac{f_Z(z_j)^2}{f_{X, Z(1)}(c^+, z_j)^2}dz_j\right) +O_p(h) < \infty
\end{align*}

When $h\asymp \sqrt{h_1} \asymp h_2$

\begin{align*}
&\mathbb{E}\left( \frac{f_Z(z_i)^2}{f_{X, Z(1)}(c^+, z_i)^4}(h_1)^{-4}K_1^2(\frac{c-x_j}{h_1}, \frac{z_i-z_j}{h_1})K^2\Big( \frac{x_i-c}{h}\Big)R_i^2 \Big| x_j, z_j\right)\\
& = \int_{z_i} \frac{f_Z(z_i)^2}{f_{X, Z(1)}(c^+, z_i)^4}(h_1)^{-2}K_1^2(\frac{c-x_j}{h_1}, \frac{z_i-z_j}{h_1}) J(c, z_i) dz_i\\
& =  \kappa_{20}\int_{z_i} \frac{f_Z(z_i)^2}{f_{X, Z(1)}(c^+, z_i)^3}(h_1)^{-2}K_1^2(\frac{c-x_j}{h_1}, \frac{z_i-z_j}{h_1}) (\sigma^2 + d_1(c^+, z_i)^2)  dz_i +O_p(h)\\
& = \kappa_{20}^2 \frac{f_Z(z_j)^2}{f_{X, Z(1)}(c^+, z_j)^3} (\sigma^2 + d_1(c^+, z_j)^2)K^2(\frac{c-x_j}{h_1})  +O_p(h) 
\end{align*}
Thus 
\begin{align*}
&\mathbb{E}\left( \frac{f_Z(z_i)^2}{f_{X, Z(1)}(c^+, z_i)^4}(h_1)^{-4}K_1^2(\frac{c-x_j}{h_1}, \frac{z_i-z_j}{h_1})K^2\Big( \frac{x_i-c}{h}\Big)R_i^2 \right)\\
& = \kappa_{20}^2 \int \int \frac{f_Z(z_j)^2}{f_{X, Z(1)}(c^+, z_j)^3} (\sigma^2 + d_1(c^+, z_j)^2)K^2(\frac{c-x_j}{h_1}) f_{X, Z(1)}(x_j, z_j)  +O_p(h) \\
& = O_p(h) < \infty
\end{align*}
\item From Lemma \ref{lemAM}, equation \ref{eqA3} can be written as
\begin{align*}
&\mathbb{E}\left( \frac{f_Z(z_i)^2}{f_{X, Z(1)^2}(c^+, z_i)}K^2\Big( \frac{x_i-c}{h}\Big)R_i^2\right)\\
& = h\int_{z_i} \frac{f_Z(z_i)^2}{f_{X, Z(1)}(c, z_i)^2}J(c, z_i) dz_i\\
& = h\kappa_{20}\int_{z_i} \frac{f_Z(z_i)^2}{f_{X, Z(1)}(c, z_i)}(\sigma^2 + d_1(c^+, z_j)^2)  dz_i +O_p(h^2)\\
& = O_p(h) < \infty
\end{align*}

\item From Lemma \ref{lemAM}, and using the fact that the kernel is bounded by a constant $\mathcal{K}$, equation \ref{eqA4} can be written as


\begin{align*}
 &\mathbb{E}\left( \frac{f_Z(z_i)}{f_{X, Z(1)}(c^+, z_i)^3} (h_2)^{-1}(h_1)^{-2}K_1(\frac{c-x_j}{h_1}, \frac{z_i-z_j}{h_1})K(\frac{z_i-z_j}{h_2}) K^2\Big( \frac{x_i-c}{h}\Big)R_i^2\Big| x_j, z_j\right)\\
 & = \int_{z_i}\frac{f_Z(z_i)}{f_{X, Z(1)}(c^+, z_i)^3}(h_1)^{-2} K_1(\frac{c-x_j}{h_1}, \frac{z_i-z_j}{h_1})K(\frac{z_i-z_j}{h_2}) J(c, z_j) dz_i\\
 & = \kappa_{20}K(\frac{c-x_j}{h_1})\int_{z_i}\frac{f_Z(z_i)}{f_{X, Z(1)}(c^+, z_i)^2}(h_1)^{-2} K( \frac{z_i-z_j}{h_1})K(\frac{z_i-z_j}{h_2})  (\sigma^2 + d_1(c^+, z_i)^2) dz_i +O_p(h)\\
 & \leq \mathcal{K}\kappa_{20}K(\frac{c-x_j}{h_1})(h_1)^{-1}\frac{f_Z(z_j)}{f_{X, Z(1)}(c^+, z_j)^2}  (\sigma^2 + d_1(c^+, z_j)^2)  +O_p(h)
\end{align*}

Thus 
\begin{align*}
&\mathbb{E}\left( \frac{f_Z(z_i)}{f_{X, Z(1)}(c^+, z_i)^3} (h_2)^{-1}(h_1)^{-2}K_1(\frac{c-x_j}{h_1}, \frac{z_i-z_j}{h_1})K(\frac{z_i-z_j}{h_2}) K^2\Big( \frac{x_i-c}{h}\Big)R_i^2\right)\\
& \leq \mathcal{K}\kappa_{20} \int \int K(\frac{c-x_j}{h_1})(h_1)^{-1}\frac{f_Z(z_j)}{f_{X, Z(1)}(c^+, z_j)^2}  (\sigma^2 + d_1(c^+, z_j)^2) f_{X, Z(1)}(x_j, z_j)dz_jdx_j  +O_p(h) \\
& = \mathcal{K}\kappa_{20} \int \frac{f_Z(z_j)}{f_{X, Z(1)}(c^+, z_j)} (\sigma^2 + d_1(c^+, z_j)^2) dz_j  +O_p(h) \\
& = \mathcal{K}\kappa_{20} \left(\sigma^2 \int \frac{f_Z(z_j)}{f_{X, Z(1)}(c^+, z_j)}dz_j + \int d_1(c^+, z_j)^2 \frac{f_Z(z_j)}{f_{X, Z(1)}(c^+, z_j)}dz_j\right) +O_p(h) < \infty
\end{align*}

\item From Lemma \ref{lemAM}, equation \ref{eqA5} can be written as
\begin{align*}
 &\mathbb{E}\left( \frac{f_Z(z_i)}{f_{X, Z(1)}(c^+, z_i)^2} (h_2)^{-1}K(\frac{z_i-z_j}{h_2}) K^2\Big( \frac{x_i-c}{h}\Big)R_i^2\Big|x_j, z_j\right)\\
 & = \int \frac{f_Z(z_i)}{f_{X, Z(1)}(c^+, z_i)^2} K(\frac{z_i-z_j}{h_2}) J(c, z_i)dz_i \\
 & =  \kappa_{20}\int \frac{f_Z(z_i)}{f_{X, Z(1)}(c^+, z_i)} K(\frac{z_i-z_j}{h_2})  (\sigma^2 + d_1(c^+, z_i)^2) dz_i +O_p(h) \\
 & = O_p(h)
\end{align*}
Thus
\begin{align*}
 &\mathbb{E}\left( \frac{f_Z(z_i)}{f_{X, Z(1)}(c^+, z_i)^2} (h_2)^{-1}K(\frac{z_i-z_j}{h_2}) K^2\Big( \frac{x_i-c}{h}\Big)R_i^2\right) < \infty
\end{align*}
\item From Lemma \ref{lemAM}, equation \ref{eqA6} can be written as
\begin{align*}
 &\mathbb{E}\left(\frac{f_Z(z_i)^2}{f_{X, Z(1)}(c^+, z_i)^3}(h_1^2)^{-1}K_1(\frac{c-x_j}{h_1}, \frac{z_i-z_j}{h_1})K^2\Big( \frac{x_i-c}{h}\Big)R_i^2\Big|x_j, z_j\right)\\
 & = h\int \frac{f_Z(z_i)^2}{f_{X, Z(1)}(c^+, z_i)^3}(h_1^2)^{-1}K_1(\frac{c-x_j}{h_1}, \frac{z_i-z_j}{h_1})J(c, z_i) dz_i\\
 & = h\kappa_{20}K(\frac{c-x_j}{h_1})\int \frac{f_Z(z_i)^2}{f_{X, Z(1)}(c^+, z_i)^2}(h_1^2)^{-1} K(\frac{z_i-z_j}{h_1})(\sigma^2 + d_1(c^+, z_i)^2) dz_i +O_p(h)\\
 & = hh_1^{-1}\kappa_{20}K(\frac{c-x_j}{h_1}) \frac{f_Z(z_j)^2}{f_{X, Z(1)}(c^+, z_j)^2} (\sigma^2 + d_1(c^+, z_j)^2) dz_i +O_p( h)
\end{align*}
Thus when $h\asymp h_1 \asymp h_2$
\begin{align*}
&\mathbb{E}\left( \frac{f_Z(z_i)^2}{f_{X, Z(1)}(c^+, z_i)^2}K_1^2(\frac{c-x_j}{h_1}, \frac{z_i-z_j}{h_1})K^2\Big( \frac{x_i-c}{h}\Big)R_i^2 \right)\\
& = \kappa_{20}^2 \int \int \frac{f_Z(z_j)^2}{f_{X, Z(1)}(c^+, z_j)^3} (\sigma^2 + d_1(c^+, z_j)^2)K(\frac{c-x_j}{h_1}) f_{X, Z(1)}(x_j, z_j)dz_jdx_j   +O_p(h) \\
& = O_p(h)<\infty
\end{align*}
and when $h\asymp \sqrt{h_1} \asymp h_2$
\begin{align*}
&\mathbb{E}\left( \frac{f_Z(z_i)^2}{f_{X, Z(1)}(c^+, z_i)^2}K_1^2(\frac{c-x_j}{h_1}, \frac{z_i-z_j}{h_1})K^2\Big( \frac{x_i-c}{h}\Big)R_i^2 \right) = O_p(\sqrt{h})<\infty
\end{align*}
\end{itemize}

\end{proof}

\begin{lemma} \label{lemAM}
Under the same condition as in Theorem \ref{thm2}, 
\begin{equation*}
\begin{split}
J(c, z_i) &:= \frac{1}{h}\int K^2\left(\frac{x_i-c}{h}\right)\EE(R_i^2|x_i,z_i) f_{X, Z(1)}(x_i, z_i) dx_i \\
& =(\sigma^2 + d_1(c^+, z_i)^2) \cdot f_{X, Z(1)}(c^+, z_i)\kappa_{20}  +O_p(h),,
\end{split}
\end{equation*}
where $O_p$ terms are valid uniformly over $i$.
\end{lemma}

\begin{proof}
Recall $R_i = \epsilon_i + m_1(x_i, z_i) - \alpha_1 - (x_i-c)\beta_1 = \epsilon_i + d_1(x_i, z_i)  - (x_i-c)\beta_1$, thus

\begin{align*}
\EE(R_i^2|x_i,z_i) &= d_1(x_i, z_i)^2 + \EE(\epsilon_i^2|x_i, z_i) + (x_i-c)^2\beta_1^2\\
& + 2d_1(x_i, z_i)\EE(\epsilon_i|x_i, z_i) +2(x_i-c)\beta_1\EE(\epsilon_i|x_i, z_i) + 2d_1(x_i, z_i) (x_i-c)\beta_1 \\
& = \sigma^2 + d_1(x_i, z_i)^2 + (x_i-c)^2\beta_1^2 + 2d_1(x_i, z_i) \cdot (x_i-c)\beta_1
\end{align*}

Following the standard Taylor expansion, we can show that
\begin{align*}
& J(c, z_i) = \frac{1}{h}\int K^2\left(\frac{x_i-c}{h}\right)\EE(R_i^2|x_i,z_i) f_{X, Z(1)}(x_i, z_i) dx_i \\
& = \sigma^2\int_{u} K^2(u) f_{X, Z(1)}(c+uh, z_i) du + \int_{u} K^2(u) d_1(c+uh, z_i)^2 f_{X, Z(1)}(c+uh, z_i) du \\
& + \beta_1^2h^2\int_{u} K^2(u) u^2 f_{X, Z(1)}(c+uh, z_i) du + 2h\beta_1 \cdot \int_{u} K^2(u) ud_1(c+uh, z_i) f_{X, Z(1)}(c+uh, z_i) du\\
& =(\sigma^2 + d_1(c^+, z_i)^2) \cdot f_{X, Z(1)}(c^+, z_i)\int_{u} K^2(u)  du  +O_p(h),
\end{align*}
where $O_p$ terms are valid uniformly over $i$ as the (mixed) third derivatives of $f_{X, Z(1)}(x_i, z_i)$ are all bounded. 
\end{proof}

\bibliographystyle{ims}
\bibliography{balancedRD}

\end{document}